\documentclass{elsarticle}

\usepackage[british]{babel}
\usepackage[utf8]{inputenc}
\usepackage[T1]{fontenc}

\usepackage{amsmath}
\usepackage{amsfonts}
\usepackage{amssymb}
\usepackage{stmaryrd}

\usepackage{url}
\usepackage{hyperref}
\usepackage{array}
\usepackage{tikz}
\usetikzlibrary{snakes,arrows,shapes}
\usetikzlibrary{matrix}
\usetikzlibrary{backgrounds}

\usepackage{subfig}
\usepackage{multirow}

\usepackage{enumitem}

\usepackage{comment}

\newdefinition{definition}{Definition}
\newdefinition{condition}{Condition}
\newtheorem{theorem}{Theorem}
\newtheorem{corollary}{Corollary}
\newtheorem{lemma}{Lemma}
\newtheorem{property}{Property}

\newproof{example}{Example}
\newproof{remark}{Remark}
\newproof{rawproof}{Proof}

\newenvironment{proof}{\begin{rawproof}}{\hfill$\Box$\end{rawproof}}

\def\DEF{\stackrel{\Delta}=}
\def\EQDEF{\stackrel{\Delta}\Leftrightarrow}

\def\powerset#1{2^{#1}}

\def\f#1{\operatorname{\bf #1}}

\def\mv#1{m_{#1}}
\def\range#1#2{\{#1,\ldots,#2\}}
\def\domv#1{\range 0 {m_{#1}}}

\def\subst#1#2#3{ {#1}_{[\scriptstyle #2 \mapsto #3]}}
\def\update#1#2#3#4{ {#1}_{[\scriptstyle #2 \,#3=\, #4]}}

\def\P{\mathcal P}
\def\PG{\mathbb P(G_m)}

\def\ub#1{\lceil #1 \rceil}

\def\lb#1{\lfloor #1 \rfloor} 

\def\Pub{U}
\def\Plb{L}

\def\card#1{|#1|}
\def\innodes#1{n^-(#1)}
\def\tool{\emph{Pawn}}
\def\sputnik{\texttt{SPuTNIk}}
\def\emptylattice{\varnothing}

\def\fpite#1#2{\operatorname{itefix}_{#1} #2}

\def\restrict{\nabla}
\def\o{\mathrm{o}}

\def\pabs{p^\#}

\def\pre#1{ {}^\bullet #1}
\def\post#1{ #1 {}^\bullet }
\def\cfg{\mathcal E}
\def\conflict{\#}
\def\mincfg#1{\lfloor #1 \rfloor}

\def\toset#1{\widetilde{#1}}

\def\cfglt{\lessdot}
\def\palt{<}

\def\cut#1{\f{cut}(#1)}
\def\cfp{\mathcal N}
\def\coff{\mathit{cutoffs}}

\begin{document}
\begin{frontmatter}

\title{Parameter Space Abstraction and Unfolding Semantics of Discrete Regulatory Networks}

\author[LSV,FI]{Juraj Kol\v{c}\'ak}
\author[FI]{David \v{S}afr\'anek}
\author[LSV]{Stefan Haar}
\author[LRI]{Lo\"ic Paulev\'e}
\address[LSV]{Inria and Universit\'e Paris-Saclay\\
LSV, CNRS \& ENS Paris-Saclay and Universit\'e Paris-Saclay, France}
\address[FI]{Systems Biology Laboratory (Sybila), Masaryk University, Brno, Czech Republic}
\address[LRI]{LRI UMR 8623, Univ. Paris-Sud -- CNRS, Universit\'e
Paris-Saclay, Orsay, France}

\begin{abstract}
The modelling of discrete regulatory networks combines a graph specifying the pairwise influences
between the variables of the system, and a parametrisation from which can be derived a discrete
transition system.
Given the influence graph only, the exploration of admissible parametrisations and the behaviours
they enable is computationally demanding due to the combinatorial explosions of both parametrisation and reachable state space.

This article introduces an abstraction of the parametrisation space and its refinement to account
for the existence of given transitions, and for constraints on the sign and observability of
influences.
The abstraction uses a convex sublattice containing the concrete parametrisation space
specified by its infimum and supremum parametrisations.
It is shown that the computed abstractions are optimal, i.e., no smaller convex sublattice exists.
Although the abstraction may introduce over-approximation, it has been proven to be conservative with respect to reachability of states.

Then, an unfolding semantics for Parametric Regulatory Networks is defined, taking advantage of
concurrency between transitions to provide a compact representation of reachable transitions.
A prototype implementation is provided: it has been applied to several examples of Boolean and
multi-valued networks, showing its tractability for networks with numerous components.
\end{abstract}
\end{frontmatter}

\section{Introduction}

Qualitative models of dynamics of biological regulatory networks form a
convenient framework for systems biology as they require little parametrisation
on top of knowledge available in the literature.
Regulatory networks account for the intertwined influences, positive and
negative, between components of a system.
In systems biology, these networks usually relate to gene regulation and
signalling pathways.
As it has been widely studied in the literature, the architecture of these
networks, and in particular the presence of feedback loops, contributes to a
complex emerging behaviour \cite{thieffry95,albert03,Demongeot03}.

The modelling of regulatory networks is classically done in two steps:
at first, an influence graph is built from data available in the literature.
This directed graph, where nodes are the components/species of the system, relates the
pairwise influences, positive and/or negative.
In a second step, a dynamical model is built from this influence graph.
In this paper, we focus on discrete models, where the state of each node has a
finite discrete domain, typically of very small size, if not Boolean
\cite{thomas73,jong02,laubenbacher04,BernotSemBRN,wang12}.

The specification of a discrete regulatory network requires additional
parameters on top of the influence graph.
Indeed, whereas the influence graph establishes the potential dependencies (possibly
signed) between the node value changes, they are not sufficient to determine the
function which associates each node with its next value, given the global state of
the network.
In other words, it may be known that two species both have positive influence on the activity of a third species.
However, it is rarely known if both of the activators must be present to
activate the target or if just one is sufficient. In general, an arbitrary
logical function may govern the joint influences.
Hence, the individual target values of a node in possible combinations
of its regulators' activity are (discrete) parameters;
the full set of parameters required to define a concrete regulatory network is
referred to as a \emph{parametrisation}.

A \emph{Parametric Regulatory Network} (PRN) is thus a formal model constructed to represent exactly the available biological knowledge.
It contains all the influence information available in the literature, however, no assumptions are made on the unknown specifics retaining all possibilities via different parametetrisations.

The analysis of PRNs is therefore
necessary to identify which parametrisations give a model satisfying given
dynamical properties (existence of particular sequences of state changes,
attractors, etc.).
However, the exploration of possible dynamics of PRNs is hindered by dual combinatorial
explosion limiting its scalability.
Indeed, not only is the state space exponential in the number of nodes in the
network, but the number of parametrisations is in the worst case doubly
exponential in the number of nodes.

\paragraph{Contribution}
The aim of this paper is to define an abstract semantics for PRNs to address the
combinatorial explosion of the parametrisation space and of possible traces.

First, we propose an abstraction of the parametrisation space by the means of a
convex sublattice that we specify by its bounds.
Our abstraction can then be refined to account for possible state transitions.
This leads to an abstract semantics of PRNs, where each state of the network is combined
with the (abstracted) set of possible parametrisations.

We extend our method to account for monotonicity and observability constraints
issued from the influence graph.
Monotonicity constraints derive from the sign of influences:
if a node is influenced positively (resp. negatively) by another node, a
decrease (resp. increase) of the latter
cannot cause the increase of the former.
Observability constraints specify that there should exist states in which the
related influences have an impact on the dynamics of the regulated node.
Indeed, in general, even if a node is marked as regulated
by another in the influence graph, we admit parametrisations where the state of
the latter node never affects the value of the regulated node.
Marking an influence as observable prevents such a case.

In both settings, we prove that our refinement operators lead to the best
possible abstraction of the parametrisation set by the means of a single convex
sublattice.
This result ensures that, if a state is reached in our abstract semantics,
there exists at least one parametrisation which allows a sequence of concrete
transitions leading to this state.
Therefore, whereas our approach relies on an over-approximation of the
parametrisation set, our method does not introduce spurious transitions.

Finally, we define an unfolding semantics for PRNs which allows building a
partial order representation of all the possible traces (or processes) a PRN can
generate from a given initial state.
Our unfolding semantics associates each process with the (abstracted)
set of parametrisations that can generate it.
Overall, this allows a compact representation of the possible traces of a PRN,
both by exploiting concurrency to avoid redundant exploration of ordering of 
independent transitions; and by sharing prefixes of processes that are identical
for different parametrisations.

A prototype implementation is provided to compute the finite complete prefix of
the unfolding of PRNs with abstract parametrisation space.

\paragraph{Related work}
The first systematic approach for exploring the parametrisation space of
multi-valued regulatory networks has been introduced by Bernot et
al.~\cite{bernot04}, and uses an explicit enumeration of admissible
parametrisations, which are then verified individually against temporal
properties, expressed in CTL (Computational Tree Logic~\cite{ClarkeEmerson81}).

Several works aim at improving the scalability of parameter identification,
which verify a given temporal logic property afterwards.
In~\cite{barnat12,me12,Brim2015} the method called coloured model checking is used
to capitalise on many parametrisations sharing some parts of their behaviour for
checking LTL or CTL.
The parametrisations are explicitly represented by colours (bits) in a binary vector and
the model checking is extended to binary vector operations to keep track of the
satisfying behaviours.
The approach in~\cite{gallet14} explores the state space represented
symbolically in the form of execution trees, coupled with an LTL formula, which
also aims at avoiding redundant analysis of different parametrisations having
identical behaviours.
Other methods employ symbolic representations of the parametrisation space to
enumerate valid parametrisations with respect to expected behaviours,
either with Boolean constraints \cite{fromentin07}
or with logic programming \cite{corblin10,Caspots-BioSystems16}.
Finally, \cite{BernotHoare-CMSB15} extends the Hoare logic to build a symbolic
representation of the parametrisation from (partial) trace specifications.

Contrary to our approach, all the above mentioned methods but
\cite{Caspots-BioSystems16} rely on an exact representation of the parametrisation space, either explicitly
\cite{bernot04,barnat12,me12}, or symbolically
\cite{fromentin07,corblin10,gallet14,BernotHoare-CMSB15};
the approach in \cite{Caspots-BioSystems16} is dedicated to Boolean networks and
does not allow a representation of all the possible processes.

The work in \cite{gallet14} is closest to our work since their symbolic
representation of possible traces is acyclic, similarly to unfoldings.
The encoding of parametrisations is performed using Boolean formulas. Contrary
to our fixed-size encoding, however, the formula continues to expand during the
exploration as a more detailed encoding of parametrisations is required.

\medskip
In this paper, the results of the workshop paper \cite{Kolcak-SASB16} are significantly extended by
generalising the framework to multi-valued regulatory networks (instead of only Boolean)
and providing a proof of optimality of the computed abstraction.
The generalisation also requires stronger abstraction refinement operators to
account for monotonicity and observability constraints in order to guarantee the
optimal abstraction.

\paragraph{Outline}
Sect.~\ref{sec:background} settles the main definitions of influence graph and
Parametric Regulatory Network (PRN).
Sect.~\ref{sec:representation} introduces our abstraction of the parametrisation space
and shows its optimality for abstract interpretation of traces of PRNs.
Sect.~\ref{sec:constraints} extends our abstract interpretation to account for \emph{a priori} constraints on admissible
parametrisation, namely monotonicity and observability.
Again, we show that the abstraction of the parametrisation set we compute is the best possible
abstraction.
Sect.~\ref{sec:unfolding} establishes the unfolding semantics of PRNs abstract
interpretation.
Sect.~\ref{sec:experiments} applies a prototype implementation of our unfolding of PRNs
with abstracted parametrisation space to several biological models from
literature and compares the size of the obtained complete finite prefix with the
symbolic execution tree obtained with the tool from \cite{gallet14}.
Finally, Sect.~\ref{sec:discussion} discusses our results and sketches future research
directions.

\paragraph{Notations}
We use $\prod$ to build Cartesian products between sets.
As the ordering of components matters, $\prod$ is not a commutative operator.
Therefore, we write
$\prod^{\leq}_{x\in X}$
for the product over elements in $X$ according to a total order $\leq$.
To ease notations, when the order is clear from the context, or when either $X$ is a set of
integers, or a set of integer vectors, on which we use the lexicographic
ordering, we simply write $\prod_{x\in X}$.

Given a sequence of $n$ elements
$\pi = (\pi_i)_{1 \leq i \leq n}$,
we write $\toset \pi\DEF\{ \pi_i \mid 1\leq i\leq n\}$ the set of its elements.

We denote by
$\fpite{x_0}{f}$ the fixpoint of the iteration of the
monotonic function $f$ initially applied on $x_0$.

Given a vector $v=\langle v_1,\dots,v_n\rangle$,
we write $\subst v i y$ for the vector equal to $v$ except on the component $i$, which is equal to
$y$.
Moreover, we write $\update v i + y$ and $\update v i - y$ for the vector equal to $v$ except on
component $i$, which is equal to $v_i + y$ and $v_i - y$.


\section{Definitions}
\label{sec:background}

This section settles the definitions of influence graph, parametric regulatory
network, and their concrete semantics.
Our framework is general enough to subsume most of the usual definitions of Boolean and multi-valued
networks with asynchronous semantics.
Note that, at this stage, we do not consider influences to be signed (i.e., negative/positive). These will be introduced in Sect.~\ref{sec:constraints}.


An influence graph is a classical directed graph, where nodes define the
variables of the system.

\begin{definition}
An \emph{Influence Graph} (IG) is a tuple $G=(V,I)$ where $V=\{1,\dots,n\}$ is
the finite set of $n$ nodes (components) and $I\subseteq V\times V$ is the set of directed edges (influences).

For each $v\in V$ we denote the set of \emph{incoming nodes}, also referred to as \emph{regulators}, as $\innodes v$, $\innodes v\DEF \{u\in V\mid (u,v)\in I\}$.
\end{definition}



Given an influence graph $G=(V,I)$ of size $n$,
we define a vector $m$ of $n$ dimensions which associates
to each node $v\in V$, its maximum value $m_v\geq 0$.\footnote{%
In general, $m_v\leq \card{\{(v,u)\mid (v,u)\in I\}}$ (out-degree of $v$)~\cite{Richard06}}

A parametric regulatory network $(G,m)$, also denoted as $G_m$,
is an influence graph $G$ coupled with a vector $m$ of maximal values of each node of $G$.

Let us denote by $\Omega_v \DEF \prod_{u\in\innodes v} \domv u$
the set of \emph{regulator states} of $v$.
A \emph{parameter} associates to each node $v$ and to each of its regulator states
a value in $\domv v$.
Intuitively, a parameter $\langle v,\omega\rangle$ specifies the value towards which the node
$v$ evolves when its regulators are in state $\omega$.
The set of parameters of a network is then given by
$\Omega \DEF \bigcup_{v\in V} \left(\{v\}\times \Omega_v\right)$.
Let us define
the total order $\preceq$ on $\Omega$ as follows: $\langle v_1,\omega_1\rangle\preceq \langle v_2,\omega_2\rangle \EQDEF v_1< v_2\vee (v_1=v_2\wedge \omega_1\trianglelefteq\omega_2)$
where $\trianglelefteq$ is the vector order.

The set of \emph{parametrisations} of a network is then the set of vectors
of dimension $\card{\Omega}$ where each coordinate $\langle v,\omega\rangle\in\Omega$, $\omega\in\Omega_v$, has
value in $\domv v$:
\[
\PG \DEF \textstyle\prod^\preceq_{\langle v,\omega\rangle\in\Omega} \domv v
\]
Given a parametrisation $P\in \PG$, a node $v\in V$, and a context
$\omega\in\Omega_v$,
$P_{v,\omega}\in\domv v$ is the coordinate $\langle v,\omega\rangle$ of the vector $P$.

A \emph{Parametric Regulatory Network} (PRN, Def.~\ref{def:prn}) is therefore defined by an influence
graph $G$ and the maximum values $m$ for the nodes from which derives
the set of all parametrisations $\PG$.
A PRN allows to define the set of node states $S(G_m)$ and the set
of transitions $\Delta(G_m)$
which correspond to the unitary increase or decrease of one and only one node
(asynchronous updating mode).

\begin{definition}
\label{def:prn}
A \emph{Parametric Regulatory Network} (PRN) is a couple $(G,m)$, also denoted
$G_m$, where $G=(V,I)$ is an influence graph and $m\in\mathbb N^n$ is the vector of the maximum value of
each node.

\begin{itemize}
\item The set of \emph{states} of $G_m$ is denoted by $S(G_m) \DEF \prod_{v\in V} \domv v$.
\item The set of \emph{transitions} of $G_m$ is denoted $\Delta(G_m)$ and
defined as a relation $\Delta(G_m)\subseteq S(G_m)\times S(G_m)$ such that
\begin{align*}
x \rightarrow y \in \Delta(G_m)
\EQDEF
\exists v\in V:
&(x_v < m_v \wedge y = \subst x v {x_v+1})\\
\vee &
(x_v > 0 \wedge y = \subst x v {x_v-1})
\end{align*}
\end{itemize}
\end{definition}

Given a transition $x\rightarrow y\in\Delta(G_m)$,
we write
$x\xrightarrow{v,+}y$ if $y=\subst x v {x_v+1}$
and
$x\xrightarrow{v,-}y$ if $y=\subst x v {x_v-1}$.

A Discrete Regulatory Network (DRN, Def.~\ref{def:drn}) can then be defined as a PRN $G_m$ associated with a unique
parametrisation $P\in\PG$.
The transition relation $\Delta(G_m,P)\subseteq\Delta(G_m)$ contains only transitions that
modify the value of a node in a direction consistent with parameter values given by $P$.


\begin{definition}
\label{def:drn}
A \emph{Discrete Regulatory Network} (DRN) is a couple $(G_m,P)$ where $G_m$ is
a PRN with $G=(V,I)$, and $P\in\PG$ a parametrisation.

The transition relation
$\Delta(G_m,P) \subseteq \Delta(G_m)$ is defined as, $\forall x\in S(G_m),
\forall v \in V$,
\begin{align*}
x\xrightarrow{v,+}y
\in\Delta(G_m,P)
& \EQDEF
    P_{v,\omega_v(x)} > x_v
\\
x\xrightarrow{v,-}y
\in\Delta(G_m,P)
&\EQDEF
    P_{v,\omega_v(x)} < x_v
\end{align*}
where
$\omega_v : S(G_m) \to \Omega_v$ with
$\omega_v(x) \DEF \prod_{u\in\innodes v} \{x_u\}$
is the projection of the state to the regulators of $v$.
\end{definition}

Note that $\Delta(G_m)$ is by definition the set of all possible transitions.
This is justified by the existence of at least one $P\in\PG:t\in\Delta(G_m,P)$ for any possible transition $t$.

\begin{example}
\label{ex:running}
Fig.~\ref{fig:running_example} gives the influence graph $G=(V,I)$ and parametrisation space
of the PRN $G_m$ with node $a$ having three values, and nodes $b$ and $c$ being Boolean.
Each parametrisation is composed of 11 parameters.
In total $\PG$ contains $3^3\cdot 2^2 \cdot 2^6=6,912$ different parametrisations.
An instance of DRN $(G_m,P)$ is given with its set of transitions $\Delta(G_m,P)$.
For example, the transition
$000 \xrightarrow{a,+} 100$ derives from the fact that
$\omega_a(000) = \langle a=0\rangle$ (or simply $\langle 0\rangle$)
and as $P_{a,\langle 0\rangle}=2$, the node $a$ can increase its value.

To demonstrate the constructions used within the paper clearly and concisely,
we use a toy example which generates sufficiently simple behaviour,
as opposed to real--world biological systems.
However, interplay of several regulators, the centrepiece of out example, is common in biology.

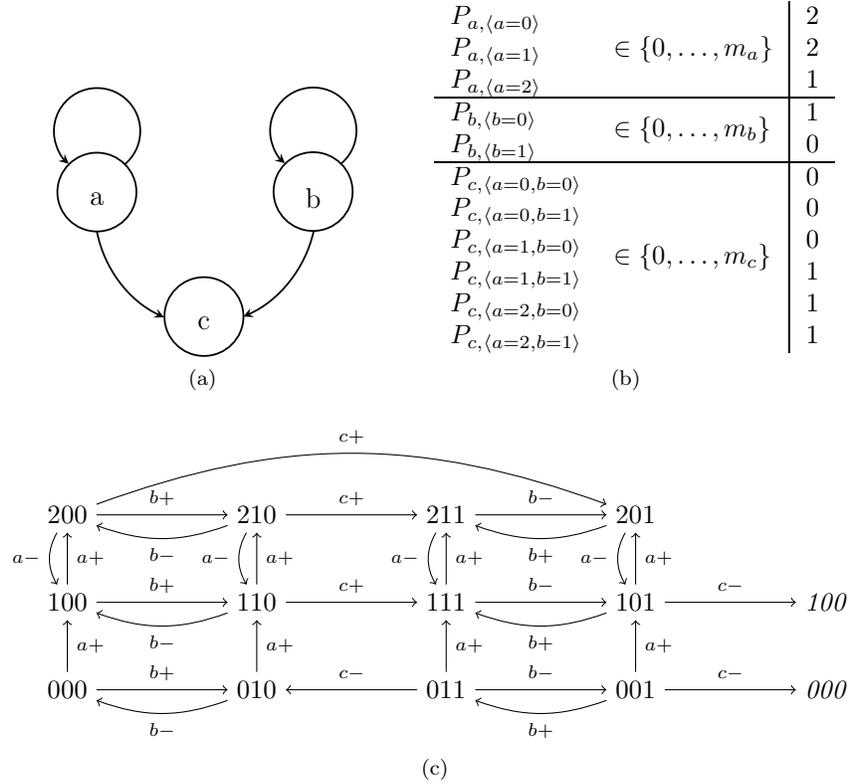
\begin{figure}[t]
\centering
\subfloat[][]{\scalebox{.9}{
\ifx\du\undefined
  \newlength{\du}
\fi
\setlength{\du}{15\unitlength}
\begin{tikzpicture}
\tikzstyle{every node}=[font=\large]
\pgftransformxscale{1.000000}
\pgftransformyscale{-1.000000}
\definecolor{dialinecolor}{rgb}{0.000000, 0.000000, 0.000000}
\pgfsetstrokecolor{dialinecolor}
\definecolor{dialinecolor}{rgb}{1.000000, 1.000000, 1.000000}
\pgfsetfillcolor{dialinecolor}
\definecolor{dialinecolor}{rgb}{1.000000, 1.000000, 1.000000}
\pgfsetfillcolor{dialinecolor}
\pgfpathellipse{\pgfpoint{22.103738\du}{20.108674\du}}{\pgfpoint{1.106838\du}{0\du}}{\pgfpoint{0\du}{1.104174\du}}
\pgfusepath{fill}
\pgfsetlinewidth{0.050000\du}
\pgfsetdash{}{0pt}
\pgfsetdash{}{0pt}
\pgfsetmiterjoin
\definecolor{dialinecolor}{rgb}{0.000000, 0.000000, 0.000000}
\pgfsetstrokecolor{dialinecolor}
\pgfpathellipse{\pgfpoint{22.103738\du}{20.108674\du}}{\pgfpoint{1.106838\du}{0\du}}{\pgfpoint{0\du}{1.104174\du}}
\pgfusepath{stroke}
\definecolor{dialinecolor}{rgb}{0.000000, 0.000000, 0.000000}
\pgfsetstrokecolor{dialinecolor}
\node at (22.103738\du,20.303674\du){a};
\definecolor{dialinecolor}{rgb}{1.000000, 1.000000, 1.000000}
\pgfsetfillcolor{dialinecolor}
\pgfpathellipse{\pgfpoint{28.161238\du}{20.094174\du}}{\pgfpoint{1.106838\du}{0\du}}{\pgfpoint{0\du}{1.104174\du}}
\pgfusepath{fill}
\pgfsetlinewidth{0.050000\du}
\pgfsetdash{}{0pt}
\pgfsetdash{}{0pt}
\pgfsetmiterjoin
\definecolor{dialinecolor}{rgb}{0.000000, 0.000000, 0.000000}
\pgfsetstrokecolor{dialinecolor}
\pgfpathellipse{\pgfpoint{28.161238\du}{20.094174\du}}{\pgfpoint{1.106838\du}{0\du}}{\pgfpoint{0\du}{1.104174\du}}
\pgfusepath{stroke}
\definecolor{dialinecolor}{rgb}{0.000000, 0.000000, 0.000000}
\pgfsetstrokecolor{dialinecolor}
\node at (28.161238\du,20.289174\du){b};
\pgfsetlinewidth{0.050000\du}
\pgfsetdash{}{0pt}
\pgfsetdash{}{0pt}
\pgfsetbuttcap
{
\definecolor{dialinecolor}{rgb}{0.000000, 0.000000, 0.000000}
\pgfsetfillcolor{dialinecolor}
\pgfsetarrowsend{stealth}
\definecolor{dialinecolor}{rgb}{0.000000, 0.000000, 0.000000}
\pgfsetstrokecolor{dialinecolor}
\pgfpathmoveto{\pgfpoint{22.886307\du}{19.327976\du}}
\pgfpatharc{50}{-229}{1.209737\du and 1.209737\du}
\pgfusepath{stroke}
}
\definecolor{dialinecolor}{rgb}{1.000000, 1.000000, 1.000000}
\pgfsetfillcolor{dialinecolor}
\pgfpathellipse{\pgfpoint{25.102938\du}{23.599378\du}}{\pgfpoint{1.106838\du}{0\du}}{\pgfpoint{0\du}{1.104174\du}}
\pgfusepath{fill}
\pgfsetlinewidth{0.050000\du}
\pgfsetdash{}{0pt}
\pgfsetdash{}{0pt}
\pgfsetmiterjoin
\definecolor{dialinecolor}{rgb}{0.000000, 0.000000, 0.000000}
\pgfsetstrokecolor{dialinecolor}
\pgfpathellipse{\pgfpoint{25.102938\du}{23.599378\du}}{\pgfpoint{1.106838\du}{0\du}}{\pgfpoint{0\du}{1.104174\du}}
\pgfusepath{stroke}
\definecolor{dialinecolor}{rgb}{0.000000, 0.000000, 0.000000}
\pgfsetstrokecolor{dialinecolor}
\node at (25.102938\du,23.794378\du){c};
\pgfsetlinewidth{0.050000\du}
\pgfsetdash{}{0pt}
\pgfsetdash{}{0pt}
\pgfsetbuttcap
{
\definecolor{dialinecolor}{rgb}{0.000000, 0.000000, 0.000000}
\pgfsetfillcolor{dialinecolor}
\pgfsetarrowsend{stealth}
\definecolor{dialinecolor}{rgb}{0.000000, 0.000000, 0.000000}
\pgfsetstrokecolor{dialinecolor}
\pgfpathmoveto{\pgfpoint{22.103701\du}{21.212663\du}}
\pgfpatharc{169}{115}{3.362903\du and 3.362903\du}
\pgfusepath{stroke}
}
\pgfsetlinewidth{0.050000\du}
\pgfsetdash{}{0pt}
\pgfsetdash{}{0pt}
\pgfsetbuttcap
{
\definecolor{dialinecolor}{rgb}{0.000000, 0.000000, 0.000000}
\pgfsetfillcolor{dialinecolor}
\pgfsetarrowsstart{stealth}
\definecolor{dialinecolor}{rgb}{0.000000, 0.000000, 0.000000}
\pgfsetstrokecolor{dialinecolor}
\pgfpathmoveto{\pgfpoint{26.209504\du}{23.599501\du}}
\pgfpatharc{66}{13}{3.464593\du and 3.464593\du}
\pgfusepath{stroke}
}
\pgfsetlinewidth{0.050000\du}
\pgfsetdash{}{0pt}
\pgfsetdash{}{0pt}
\pgfsetbuttcap
{
\definecolor{dialinecolor}{rgb}{0.000000, 0.000000, 0.000000}
\pgfsetfillcolor{dialinecolor}
\pgfsetarrowsend{stealth}
\definecolor{dialinecolor}{rgb}{0.000000, 0.000000, 0.000000}
\pgfsetstrokecolor{dialinecolor}
\pgfpathmoveto{\pgfpoint{28.943807\du}{19.313476\du}}
\pgfpatharc{50}{-229}{1.209737\du and 1.209737\du}
\pgfusepath{stroke}
}
\end{tikzpicture}}}
\hspace{5mm}
\subfloat[][]{
    \raisebox{2.3cm}{
    \begin{tabular}{ll|c}
        $P_{a,\langle a=0\rangle}$ & \multirow{3}{*}{$\in \domv a$} & 2\\
        $P_{a,\langle a=1\rangle}$ & & 2\\
        $P_{a,\langle a=2\rangle}$ && 1\\\hline
        $P_{b,\langle b=0\rangle}$ & \multirow{2}{*}{$\in \domv b$} &1\\
        $P_{b,\langle b=1\rangle}$&&0\\\hline
        $P_{c,\langle a=0,b=0\rangle}$ & \multirow{6}{*}{$\in \domv c$}&0\\
        $P_{c,\langle a=0,b=1\rangle}$&&0\\
        $P_{c,\langle a=1,b=0\rangle}$&&0\\
        $P_{c,\langle a=1,b=1\rangle}$&&1\\
        $P_{c,\langle a=2,b=0\rangle}$&&1\\
        $P_{c,\langle a=2,b=1\rangle}$&&1
    \end{tabular}}}

\subfloat[][]{
\begin{tikzpicture}[ampersand replacement=\&]
\matrix (m) [matrix of math nodes,row sep=2em,column sep=5em] {
200 \& 210 \& 211 \& 201 \&\\
100 \& 110 \& 111 \& 101 \&  \it 100\\
000 \& 010 \& 011 \& 001 \&  \it 000\\
};
\path[->,every node/.style={font=\scriptsize}]
(m-3-1) edge node [right] {$a+$} (m-2-1)
(m-2-1) edge node [right] {$a+$} (m-1-1)
(m-1-1) edge[bend right] node [left] {$a-$} (m-2-1)
(m-3-2) edge node [right] {$a+$} (m-2-2)
(m-2-2) edge node [right] {$a+$} (m-1-2)
(m-1-2) edge[bend right] node [left] {$a-$} (m-2-2)
(m-3-3) edge node [right] {$a+$} (m-2-3)
(m-2-3) edge node [right] {$a+$} (m-1-3)
(m-1-3) edge[bend right] node [left] {$a-$} (m-2-3)
(m-3-4) edge node [right] {$a+$} (m-2-4)
(m-2-4) edge node [right] {$a+$} (m-1-4)
(m-1-4) edge[bend right] node [left] {$a-$} (m-2-4)
(m-3-1) edge node[above] {$b+$} (m-3-2)
(m-3-2) edge[bend left=20] node[below] {$b-$} (m-3-1)
(m-2-1) edge node[above] {$b+$} (m-2-2)
(m-2-2) edge[bend left=20] node[below] {$b-$} (m-2-1)
(m-1-1) edge node[above] {$b+$} (m-1-2)
(m-1-2) edge[bend left=20] node[below] {$b-$} (m-1-1)
(m-3-3) edge node[above] {$b-$} (m-3-4)
(m-3-4) edge[bend left=20] node[below] {$b+$} (m-3-3)
(m-2-3) edge node[above] {$b-$} (m-2-4)
(m-2-4) edge[bend left=20] node[below] {$b+$} (m-2-3)
(m-1-3) edge node[above] {$b-$} (m-1-4)
(m-1-4) edge[bend left=20] node[below] {$b+$} (m-1-3)
(m-1-2) edge node[above]{$c+$} (m-1-3)
(m-2-2) edge node[above]{$c+$} (m-2-3)
(m-3-3) edge node[above]{$c-$} (m-3-2)
(m-2-4) edge node[above]{$c-$} (m-2-5)
(m-3-4) edge node[above]{$c-$} (m-3-5)
(m-1-1) edge[bend left=20] node[above]{$c+$} (m-1-4)
;

\end{tikzpicture}
}
\caption{(a) Influence graph $G$ and (b) parametrisation domain of a PRN $G_m$
with $m_a=2$, $m_b=1$, and $m_c=1$.
For readability, we use letters instead of numbers for nodes,
and write explicitly the component name in vectors.
In the rest of the paper, we also use shorter notations, e.g., 
        $P_{c,\langle 2,1\rangle}$
        instead of
        $P_{c,\langle a=2,b=1\rangle}$.
 (c) Transitions of the DRN with parametrisation $P=\langle 2,2,1,1,0,0,0,0,1,1,1\rangle$
corresponding to the right column of (b).
Nodes are states $S(G_m)$ in the order $a,b,c$.
}
\label{fig:running_example}
\end{figure}
\end{example}





\section{Abstraction of Parametrisation Sets}
\label{sec:representation}

The number of candidate parametrisations being exponential with the number of parameters (which is
exponential with the in-degree of nodes), the concrete representation of parametrisation set is a typical bottleneck.

In this section, we introduce an abstraction of a parametrisation set
by the means of a bounded convex sublattice of the lattice of all parametrisations
$\PG$ with respect to the \emph{parametrisation order} with $k=|\Omega|$
(Def.~\ref{def:param-order}) and study its restriction with respect to transitions.
A bounded convex sublattice can be specified solely by its least and its greatest elements
$\Plb\in\PG$, respectively $\Pub\in\PG$, allowing us to uniquely represent
a parametrisation set by only two parametrisations $(\Plb,\Pub)$.
Furthermore, we show that our abstraction introduces no over-approximation of the parametrisation set
unless the model is refined with a set of constraints (Section \ref{sec:constraints}).

\begin{definition}\label{def:param-order}
The parametrisation order $\leq$ on vectors of integers of length $k$ is a partial order
such that $\forall v,w\in\mathbb{N}^k$:
$$v\leq w\EQDEF\forall i\in\{0,\dots,k\}:v_i\leq w_i$$
\end{definition}

First, we consider in Sect.~\ref{sec:concrete} the restriction of a concrete parametrisation set for a given set of transitions
and we analyse its algebraical properties.
Such a restriction allows defining the semantics of Parametric Regulatory Networks, where
states of nodes are coupled with a parametrisation set, and transitions restrict the latter.
Sect.~\ref{sec:abstract} presents its abstract counterpart, and demonstrates that the abstraction is
exact: it preserves all parametrisations for any subset of transitions and introduces no
over-approximation.

\subsection{Concrete parametrisation space}
\label{sec:concrete}

From DRN semantics (Def.~\ref{def:drn}),
we define $\P_t$ the subset of parametrisations of a PRN enabling a transition $t$
(Def.~\ref{def:pt}).
Given a set of transitions $T$, the concrete set of parametrisations enabling all the transitions
in $T$ is the intersection of all the $\P_t$ for $t\in T$.
We denote by $p(T)$ the obtained parametrisation set (Def.~\ref{def:pT}).

\begin{definition}\label{def:pt}
Let $G_m$ be a PRN and $t\in\Delta(G_m)$ a transition.
The \emph{parametrisation set enabling $t$}, denoted $\P_t$, is defined as follows:
\begin{align*}
\P_{x\xrightarrow{v,+} y} &\DEF
\{P\in\PG\mid P_{v,\omega_v(x)} \geq x_v+1\},
\\
\P_{x\xrightarrow{v,-} y} &\DEF
\{P\in\PG\mid P_{v,\omega_v(x)} \leq  x_v-1\}.
\end{align*}
\end{definition}

\begin{definition}[$p(T)$]\label{def:pT}
Let $G_m$ be a PRN. 
Given a set of transitions $T\subseteq\Delta(G_m)$, the \emph{concrete parametrisation set enabling $T$}, denoted
$p(T)$, is defined as follows:
\begin{itemize}
\item[] $p(\emptyset) \DEF \PG$,
\item[] $p(T) \DEF \bigcap_{t\in T} \P_t$ if $T\neq\emptyset$.
\end{itemize}
\end{definition}

Given any sequence of transitions $\pi = x\to \dots\to y$, it follows that
$p(\toset \pi)\neq\emptyset$ if and only if there exists a DRN $(G_m,P)$ where
$\toset\pi\subseteq\Delta(G_m,P)$, i.e., the DRN can produce the trace $\pi$.
This leads to the definition of realisable traces of PRNs.

\begin{definition}[Concrete semantics of PRNs]
\label{def:concrete-PRNs}
Given a PRN $G_m$,
a sequence $\pi$ of transitions in $\Delta(G_m)$
is \emph{realisable} if and only if
$p(\toset\pi)\neq\emptyset$.
\end{definition}

It is important to remark that the set of parametrisations $\PG$
is a bounded lattice with respect to the parametrisation ordering.
It comes from the fact that the set of parametrisations is always finite.
This property is naturally extended to a parametrisation set enabling any set of transitions $T$.
\begin{property}
\label{pty:pT-convex}
$p(T)$ is a bounded convex sublattice of $\PG$.
\end{property}
\begin{proof}
Let $t=x\xrightarrow{v,+}y\in T$ and $P,P',P''\in\PG$ be arbitrary parametrisations such that $P'\leq P\leq P''$ and $P',P''\in\P_t$.
From $P'\in\P_t$ we know $P'_{v,\omega_v(x)}\geq y_v$ and since $P\geq P'$ we get $P_{v,\omega_v(x)}\geq P'_{v,\omega_v(x)}$
thus $P_{v,\omega_v(x)}\geq y_v$.

A symmetric proof can be constructed for decreasing transitions ($t=x\xrightarrow{v,-}y\in T$) using $P''$ to arrive at $P_{v,\omega_v(x)}\leq y_v$.
Surely then, $P$ must belong to $\P_t$, therefore for all $t\in T$, $\P_t$ is a convex sublattice of $\PG$.
The intersection of convex sublattices is a convex sublattice.
Boundedness follows from the fact that $\PG$ is finite.
\end{proof}

\subsection{Abstract parametrisation space}
\label{sec:abstract}

A bounded convex sublattice is fully determined by its least and greatest element.
We write $(\Plb,\Pub)$ to represent the convex sublattice of parametrisation sets, where
$\Plb,\Pub\in\PG$.

\paragraph{Additional notations}
An empty lattice is denoted by $\emptylattice$.
By abuse of notation, we may also write $(\Plb,\Pub)=\emptylattice$
to address the fact that $(\Plb,\Pub)$ represents an empty lattice ($\neg(\Plb\leq\Pub)$).
We use $\lb{A}$ and $\ub{A}$ to denote lower and upper bounds (respectively)
of a bounded lattice generated by a set of elements $A$.
Given two vectors $x,y$ of length $n$, we denote
$\max(x,y) \DEF \langle \max(x_i,y_i) \mid n \in \range 1 n\rangle$
and
$\min(x,y) \DEF \langle \min(x_i,y_i) \mid n \in \range 1 n\rangle$.

\medskip

In the concrete domain, we restrict the parameter set $p(T)$ by a transition $t\notin T$ to obtain $p(T)\cap\P_t=p(T\cup\{t\})$.
In order to obtain the abstract counterpart of the restriction, we define a narrowing operator $\restrict_t$ which
refines an abstract parametrisation space $(\Plb,\Pub)$ according to the specified transition $t$
(Def.~\ref{def:narrow-t}).
If the transition increases the value of $x_v$,
necessarily, all parametrisations $P$ should satisfy $P_{v,\omega_v(x)}\geq x_v+1$.
Therefore, the lower bound $\Plb$ of the parametrisation set at $\langle v,\omega_v(x)\rangle$-coordinate
is at least $x_v + 1$.
The case when the transition decreases the value of $v$ leads to an analogous refinement of the upper
bound $\Pub$.

\begin{definition}\label{def:narrow-t}
Let $(\Plb,\Pub)\in\PG^2$ be the abstraction of parametrisation set of a PRN $G_m$, and
$t\in\Delta(G_m)$ be a transition.
The \emph{narrowing of} $(\Plb,\Pub)$ \emph{by} $t$, $\restrict_t:\PG^2\to \PG^2$, is defined in the following way:
\begin{align}
\restrict_{x\xrightarrow{v,+} y}(\Plb,\Pub)
& \DEF
\left(
\max\left(\Plb, \subst {\lb{\PG}} {v,\omega_v(x)} {x_v+1}\right)
,\Pub
\right),
\\
\restrict_{x\xrightarrow{v,-} y}(\Plb,\Pub)
& \DEF
\left(
\Plb,
    \min\left(\Pub, \subst {\ub{\PG}} {v,\omega_v(x)} {x_v-1}\right)
\right).
\end{align}
\end{definition}

We can then define the abstract counterpart $\pabs(T)$ of $p(T)$ by iteratively
applying $\restrict_t$ for each $t\in T$ starting from the lower and upper bounds of $\PG$.
Iterative application of $\restrict_t$ for all $t\in T$ implicitly requires an order on transitions in $T$.
Due to the use of $\min$ and $\max$ in Def. \ref{def:narrow-t}, however,
the same result is obtained regardless of the order by which transitions in $T$ are explored.
We show that $\pabs(T)=p(T)$ (Theorem~\ref{thm:pabsT}) as a consequence of Property~\ref{pty:pT-convex}.

\begin{definition}[$\pabs(T)$]\label{def:pabsT}
Let $G_m$ be a PRN and $T\subseteq \Delta(G_m)$ a set of transitions.
The \emph{abstract parametrisation set} $\pabs(T)$ is defined inductively as follows:
\begin{itemize}
\item[] $\pabs(\emptyset) \DEF (\lb{\PG},\ub{\PG})$,
\item[] $\pabs(T \cup \{t\}) \DEF
    \restrict_t(\pabs(T))$.
\end{itemize}
\end{definition}

\begin{theorem}\label{thm:pabsT}
Given a PRN $G_m$ and transitions $T\subseteq\Delta(G_m)$,
$\pabs(T) = p(T)$.
\end{theorem}
\begin{proof}
Let $\pabs(T)=(\Plb,\Pub)$.
By Def.~\ref{def:narrow-t} and \ref{def:pabsT},
for each $v\in V$ and each $\omega\in\Omega_v$:
\begin{itemize}
\item[] $\Plb_{v,\omega} = \max(\{0\}\cup\{x_v+1\mid x\xrightarrow{v,+}y\in T \wedge
    \omega = \omega_v(x)\})$
and,

\item[] $\Pub_{v,\omega} = \min(\{\mv v\}\cup\{x_v-1\mid
x\xrightarrow{v,-}y\in T\wedge \omega = \omega_v(x)\})$.
\end{itemize}

By Def.~\ref{def:pt} and \ref{def:pT}, for each $v\in V$ and each $\omega\in\Omega_v$, any parametrisation $P\in p(T)$ satisfies
$\Plb_{v,\omega} \leq P_{v,\omega} \leq \Pub_{v,\omega}$.
Moreover,
for each $v\in V$ and each $\omega\in\Omega_v$,
there exist $P,P'\in p(T)$
such that
$P_{v,\omega}=\Plb_{v,\omega}$
and
$P'_{v,\omega}=\Pub_{v,\omega}$
\end{proof}

From this theorem derives the fact that the abstraction of the parametrisation set for a sequence
of transitions is not empty if and only if
there exists a concrete parametrisation which enables these transitions.

\begin{corollary}
Given a sequence of transitions $\pi = x\to \dots \to y$ in $\Delta(G_m)$,
\[ \pabs(\toset{\pi}) \neq \emptylattice \Longleftrightarrow
    p(\toset{\pi}) \neq \emptyset \]
\end{corollary}

\def\Sa{\scriptscriptstyle\triangle}%
\def\Sb{\scriptscriptstyle\heartsuit}%
\def\Sc{\scriptscriptstyle\clubsuit}%
\def\Sd{\scriptscriptstyle\spadesuit}%
\begin{example}
Fig.~\ref{fig:lattice} gives a sketch of the lattice representation of the parametrisation space of
the PRN $G_m$ introduced in Fig.~\ref{fig:running_example}.
The full parametrisation space $\PG$ is completely characterized by the convex sublattice with lower bound
$L=\langle 00000000000\rangle$ and upper bound $U=\langle 22211111111\rangle$.
\begin{figure}[t]
\centering
\def\myP#1#2#3#4#5#6#7#8#9{%
    \def\nextArgs##1##2{%
$#1\medspace#2\medspace#3\medspace\underset{\Sa}{#4}\thinspace\underset{\Sb}{#5}\medspace#6\medspace#7\medspace#8\medspace\underset{\Sc}{#9}\medspace##1\thinspace\underset{\Sd}{##2}$}%
    \nextArgs}
\begin{tikzpicture}[anchor=west,node distance=15mm]

\node[] (top) {\myP22211111111};
\node[font=\tiny,yshift=-0.5em,anchor=south west] at (top.north west)
{
\rotatebox{90}{$a,\langle 0\rangle$}
\rotatebox{90}{$a,\langle 1\rangle$}
\rotatebox{90}{$a,\langle 2\rangle$}
\rotatebox{90}{$b,\langle 0\rangle$}
\rotatebox{90}{$b,\langle 1\rangle$}
\rotatebox{90}{$c,\langle 0,0\rangle$}
\rotatebox{90}{$c,\langle 0,1\rangle$}
\rotatebox{90}{$c,\langle 1,0\rangle$}
\rotatebox{90}{$c,\langle 1,1\rangle$}
\rotatebox{90}{$c,\langle 2,0\rangle$}
\rotatebox{90}{$c,\langle 2,1\rangle$}
};
\node[below right of=top,xshift=3mm] (U) {{\myP 22210111111}};
\node[below of=U,xshift=1cm] (LR) {\myP 00010000101};
\node[anchor=north west,xshift=-5mm,yshift=-2mm] (LR1) at (LR.south west) {\myP 00000000101};
\node[anchor=north west,xshift=-5mm,yshift=-2mm] (L) at (LR1.south west) {\myP 00000000100};
\node[below left of=L,xshift=-3mm] (bot) {\myP 00000000000};

\node[xshift=-4cm,yshift=-1cm] (h1) at (top.south west) {};
\node[xshift=-4cm,yshift=1cm] (h2) at (bot.north west) {};
\node[xshift=4cm,yshift=-1cm] (h3) at (top.south east) {};
\node[xshift=4cm,yshift=1cm] (h4) at (bot.north east) {};
\node[below left of=U,xshift=-3cm] (h5) {};
\node[above left of=L,xshift=-3cm] (h6) {};
\node[below right of=U,xshift=3cm,yshift=4mm] (h7) {};
\node[above right of=L,xshift=3cm] (h8) {};
\node[xshift=-3mm,yshift=-5mm] (h9) at (U.south west) {};
\node[xshift=2mm,yshift=4mm] (h10) at (LR.north east) {};

\draw
    (U.west)+(0.1,0) -- ([xshift=-1mm]U.east)
    (LR.west)+(0.1,0) -- ([xshift=-1mm]LR.east)
    ([xshift=3mm]U.north west) -- (U.north west) 
        -- ([yshift=2mm]U.south west) -- ([xshift=3mm,yshift=2mm]U.south west)
    ([xshift=3mm]L.north west) -- (L.north west)
        -- ([yshift=2mm]L.south west) -- ([xshift=3mm,yshift=2mm]L.south west)
;

\draw
    (bot) -- (L)
    (L) -- (LR1)
    (LR1) -- (LR)
    (U) -- (top)
;
\draw[dashed]
    (h1) -- (top)
    (h3) -- (top)
    (bot) -- (h2)
    (bot) -- (h4)
    (h5) -- (U.west)
    (h7) -- (U.east)
    (L) -- (h6)
    (L) -- (h8)
    (LR) -- (h9) -- (U)
    (LR) -- (h10) -- (U)
;
\begin{pgfonlayer}{background}
\fill[blue,opacity=0.2]
    (U.west) -- (h5.center) -- (h6.center) 
        -- ([xshift=14mm]L.west) -- ([xshift=-14mm]L.east) -- (h8.center) -- (h7.center)
        -- (U.east) --cycle;
\fill[yellow]
    ([xshift=15mm]U.west) -- (h9.center)
    -- ([xshift=14mm]LR.west) -- ([xshift=-15mm]LR.east)
    -- (h10.center) -- ([xshift=-14mm]U.east) -- cycle;
\end{pgfonlayer}

\end{tikzpicture}
\caption{Sketch of Hasse diagram of the convex sublattice of $\PG$ for the PRN of
Fig.~\ref{fig:running_example} with $L=\langle 00000000000\rangle$ and $U=\langle 22211111111\rangle$.
The two parametrisations with a left bracket correspond to the
refined abstraction $\pabs(T)$ with 
$T=\{ 110\xrightarrow{c,+} 111; 111 \xrightarrow{b,-1} 101 \}$.
The two parametrisations underlined correspond to the further refined abstraction
$\pabs_R(T)$ with $\{(a,c,+1),(b,b,\o)\}\subseteq R$
(Sect.~\ref{sec:constraints}).
The parameter corresponding to $b,\langle b=0\rangle$ is marked with $\Sa$;
$b,\langle b=1\rangle$ with $\Sb$;
$c,\langle a=1,b=1\rangle$ with $\Sc$;
and
$c,\langle a=2,b=1\rangle$ with $\Sd$.}
\label{fig:lattice}
\end{figure}
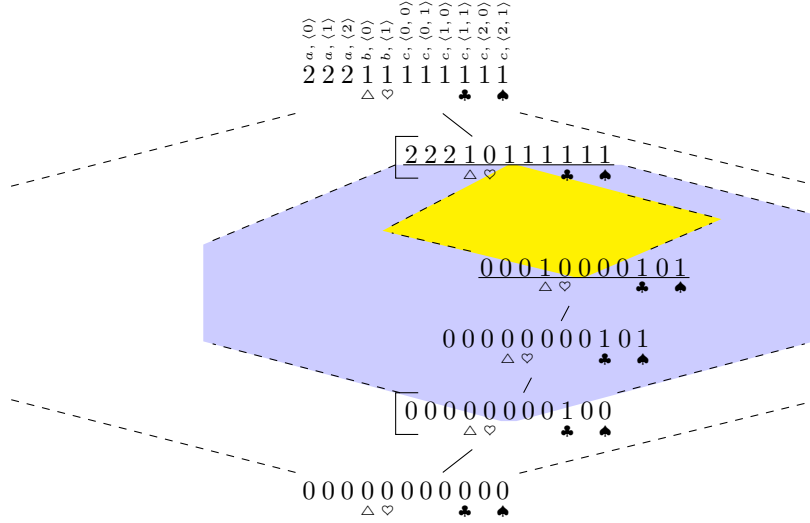

Given the set of transitions
$T=\{ 110\xrightarrow{c,+} 111; 111 \xrightarrow{b,-1} 101 \} \subset \Delta(G_m)$,
$\pabs(T)$ refines the abstraction in the following way:
the transition $110\xrightarrow{c,+} 111$ imposes that
$P_{c,\langle a=1,b=1\rangle} \geq 1$, hence the lower bound of the parametrisation space
is adjusted to $\langle 00000000\mathbf{1}00\rangle$;
and
the transition $111\xrightarrow{b,-} 101$ imposes that
$P_{b,\langle b=1\rangle} \leq 0$ which allows to refine the upper bound
to $\langle 2221\mathbf{0}111111\rangle$.
Theorem~\ref{thm:pabsT} states that $\pabs(T) = p(T)$, i.e.,
the set of parametrisations of $G_m$ enabling $T$
is exactly the set of parametrisations with the adjusted bounds of the convex sublattice
(see the light blue area in Figure~\ref{fig:lattice}).
\end{example}

\section{Global Constraints on Parametrisations}
\label{sec:constraints}

The results of Sect.~\ref{sec:representation} apply on PRNs having an influence
graph without signs on edges.
In general, there is usually partial knowledge of sign of some influences
as well as knowledge on the necessity of some influences, the latter to be referred to as \emph{observability}.

Signed influences lead to global constraints on the admissible parametrisations
in the form of \emph{monotonicity constraints} \cite{BernotSemBRN}:
the sole activation of a positive (resp. negative) regulator $u$ of node $v$ cannot cause a
decrease (resp. increase) of its value.
This imposes inequality restrictions among parameters.

Similarly, an observable influence of $u$ on $v$ imposes that, in some state, a
change in the value of $u$ should change the value of $v$.
This is again translated as inequality constraints among parameters~\cite{me12}.

This section extends our abstraction to account for such constraints.

\subsection{Definitions}

Given a PRN $G_m$ with $G=(V,I)$,
we define a \emph{well-formed} set of influence constraints
$R \subseteq V \times V \times \{+1,-1,\o\}$
such that
$\forall (u,v,c)\in R$,
$u\in\innodes v$,
and
$\forall u,v\in V, \{(u,v,+1),(u,v,-1)\}\not\subseteq R$.%
\footnote{In the framework considered in this paper, an influence being both positive- and
negative-monotonic is equivalent to having no influence.}%
~In this setting,
$(u,v,+1)$ means that the influence of $u$ on $v$ is \emph{positive-monotonic};
$(u,v,-1)$ means that the influence of $u$ on $v$ is \emph{negative-monotonic};
and
$(u,v,\o)$ means that the influence of $u$ on $v$ is \emph{observable}.

We say that $u$ has a \emph{positive-monotonic influence} on $v$
only if, for any $P\in\P$,
\begin{align*}
\forall \omega\in\Omega_v
\forall x_u\in\{1,\cdots\mv u\},
P_{v,\subst\omega u {x_u}} \geq P_{v,\subst\omega u {x_u-1}}
\end{align*}
i.e., the sole increase of the activator $u$ cannot cause a decrease of the regulated node $v$.

Similarly, $u$ has a \emph{negative-monotonic influence} on $v$
only if, for any $P\in\P$,
\begin{align*}
\forall \omega\in\Omega_v
\forall x_u\in\{1,\cdots\mv u\},
P_{v,\subst\omega u {x_u}} \leq P_{v,\subst\omega u {x_u-1}}
\end{align*}
i.e., the sole increase of the inhibitor $u$ cannot cause an increase of the regulated node $v$.

Finally, we say that $u$ has an \emph{observable influence} on $v$
only if, for any $P\in\P$,
\begin{align*}
\exists \omega\in\Omega_v
\exists x_u\in\{1,\cdots\mv u\},
P_{v,\subst\omega u {x_u}} \neq P_{v,\subst\omega u {x_u-1}}
\end{align*}
i.e., there exists a state where the sole change of the regulator $u$ triggers a change of the
regulated node $v$.

Note that the definitions are complementary, it is indeed often the case in biology that an influence is both monotonic (either positive or negative) and observable.

The concrete set of parametrisations satisfying a constraint $r\in R$ is
therefore characterised as follows.
\begin{definition}[$\P_r$]
\label{def:Pr}
Given $r\in R$, $\P_r\subseteq\PG$ is the subset of parametrisations satisfying
the influence constraint $r$ with:
\begin{align*}
\P_{(u,v,+1)} &\DEF
\{ P \in \PG \mid
    \forall \omega\in\Omega_v,
    \forall x_u\in \range 1 {\mv u},
        P_{v,\subst \omega u {x_u}} \geq
        P_{v,\subst \omega u {x_u-1}} \}
\\
\P_{(u,v,-1)} &\DEF
\{ P \in \PG \mid
    \forall \omega\in\Omega_v,
    \forall x_u\in \range 1 {\mv u},
        P_{v,\subst \omega u {x_u}} \leq
        P_{v,\subst \omega u {x_u-1}} \}
\\
\P_{(u,v,\o)} &\DEF
\{ P \in \PG \mid
    \exists \omega\in\Omega_v,
    \exists x_u\in \range 1 {\mv u},
        P_{v,\subst \omega u {x_u}} \neq
        P_{v,\subst \omega u {x_u-1}} \}
\end{align*}
\end{definition}

Given a node $v$, the monotonicity constraints allow to define a partial order over its regulator
states $\Omega_v$: $\omega\in\Omega_v$ is $\preceq_v$-smaller than $\omega'\in\Omega_v$
if for every parametrisation $P$ that satisfies the monotonicity constraints we have $P_{v,\omega}\leq P_{v,\omega'}$.

\begin{definition}[$\preceq_{v}$]
	Let $R$ be an arbitrary well-formed set of constraints.
    The \emph{monotonicity order} $\preceq_{v}\subseteq {\Omega_v}^2$
	on the regulatory contexts of $v$ is the partial order
    such that $\forall \omega, \omega'\in\Omega_v$:
	\begin{align*}
		\omega \preceq_v \omega' \EQDEF
            \forall u\in\innodes v,
\begin{cases}
\omega_u\leq\omega'_u &\text{if }(u,v,+1)\in R\\
\omega_u\geq\omega'_u &\text{if }(u,v,-1)\in R\\
\omega_u=\omega'_u & \text{otherwise.}
\end{cases}
	\end{align*}
   We write $\omega\parallel_v\omega'$ if and only if $\omega$ and $\omega'$ are not comparable
   according to $\preceq_v$. This is the case notably when $\omega_u\neq\omega'_u$ for some $u$ such that the influence $(u,v)$ is not monotonic.
\end{definition}

\subsection{Concrete parametrisation space}

The set of parametrisations which satisfy both constraints $R$ and enable a set of transitions $T$
can be directly derived by the intersection of $p(T)$ (Def.~\ref{def:pT}) with the parametrisations
satisfying $R$ (Def.~\ref{def:Pr}).
\begin{definition}\label{def:pRT}
Let $G_m$ be a PRN, and $R$ a well-formed set of influence constraints.
Given a set of transitions $T\subseteq\Delta(G_m)$, the \emph{parametrisation set enabling $T$ under $R$}, denoted
$p_R(T)$, is given by:
\begin{align*}
p_R(T) \DEF p(T) \cap
\bigcap_{(u,v,s)\in R} \P_{(u,v,s)}
\end{align*}
\end{definition}

\begin{remark}
$p_R(T)$ is no longer a convex sublattice.
\end{remark}

\subsection{Abstract parametrisation space}

Given the lower and upper bounds $(L,U)$ of the convex sublattice of parametrisations, this section
introduces narrowing operators $\nabla_r$ to account for the influence constraints $r\in R$ and adjust
these boundaries accordingly.
We demonstrate that the narrowing operators lead to the optimal abstraction of the concrete
parametrisation set, i.e., it is equal to the smallest convex sublattice which includes $p_R(T)$.

\paragraph{Ensure monotonicity}
Given an influence constraint $(u,v,s)\in R$ with $s\in\{+1,-1\}$ we define the operator
$\restrict_{(u,v,s)} : \PG^2 \to \PG^2$ which increases the lower bound and decreases the upper bound
until the $s$-monotonicity constraint is satisfied:
\begin{equation}
\restrict_{(u,v,s)}(\Plb,\Pub) \DEF
    \fpite{(\Plb,\Pub)}{f}
\end{equation}
where $f(\Plb,\Pub)\DEF(\Plb',\Pub')$ with $\forall \omega\in\Omega_v$,
\begin{align*}
\Plb'_{v,\omega} &=
 \max\left(\{\Plb_{v,\omega}\} \cup
    \left\{\Plb_{v,\subst \omega u {\omega_u-s}} \mid
        \omega_u - s\in\domv u \right\}
       \right)
\\
\Pub'_{v,\omega} &=
 \min\left(\{\Pub_{v,\omega}\} \cup
    \left\{
    \Pub_{v,\subst \omega u {\omega_u+s}} \mid
        \omega_u + s\in\domv u \right\}
       \right)
\end{align*}
and for all $a\in V, a\neq v$, and for all $\omega\in\Omega_a$,
$\Plb'_{a,\omega} = \Plb_{a,\omega}$
and
$\Pub'_{a,\omega} = \Pub_{a,\omega}$.

By iterating over regulator states $\Omega_v$ in $\preceq_v$ order, the fixpoint of $L$ can be
computed in $\card{\Omega_v}$ steps; and similarly for the fixpoint of $U$ by following the
anti-$\preceq_v$ order.

\paragraph{Ensure observability}
Given an influence constraint $(u,v,\o)$ and the boundaries $(L,U)$ of the convex sublattice of
parametrisations, the operator $\restrict_{(u,v,\o)}:\PG^2\to\PG^2$ refines $(L,U)$ to ensure the
satisfiability of the observability constraint.

\def\cst{cst}
In its simplest form, the observability criterion can be applied when for all the regulator states
$\omega\in\Omega_v$ but one, $L_{v,\omega}=U_{v,\omega}=\cst$ where $\cst\in\domv v$: in that case, it should be ensured
that for the remaining unique regulator state $\omega'\in\Omega_v$,
$L_{v,\omega'}\neq \cst$ and $U_{v,\omega'}\neq \cst$.
Intuitively, if all parameters of $v$ but one are fixed to the same value $\cst$, the remaining parameter
should take a different value, and hence, neither its upper or lower bound can be equal to $\cst$.
Although this simple measure ensures all influences of $v$ are observable,
it is only applicable in cases where having the value of the last not fixed regulator state $\omega'$ set to $\cst$ would lead to no influence of $v$ being observable.

Our definition generalises this reasoning to take into account the state of other regulators of $v$ and the monotonicity constraints.
This is especially true for the case when an influence $(u,v)$ is both observable and monotonic.
In such a case it is enough to ensure that the $\preceq_v$-minimal element $\omega$ has a lower value than the $\preceq_v$-maximal element $\omega'$,
which can be achieved by increasing the value of $L_{v,\omega'}$ and/or decreasing the value of $U_{v,\omega}$.

The formal definition is a little technical as it also accommodates the case when an observable $(u,v)$ is not monotonic,
which is achieved similarly to the simple case described in the beginning:
\begin{equation}
\restrict_{(u,v,o)}(\Plb,\Pub)\DEF
\begin{cases}
\emptylattice &\text{if $A_{u,v}(\Plb,\Pub)=\emptyset$}\\
(L',U') &\text{otherwise,}
\end{cases}
\end{equation}
where
\begin{align*}
A_{u,v}(\Plb,\Pub)&\DEF\{\omega\in\Omega_v\mid
\exists x_u\in\range 1 {\mv u}:\\
&\qquad L_{v,\subst \omega u {x_u}} <
U_{v,\subst \omega u {x_u-1}}\vee
U_{v,\subst \omega u {x_u}} >
L_{v,\subst \omega u {x_u-1}}\}
\end{align*}
and
$L' = \begin{cases}
	\update \Plb {v,\omega} + 1
	&\text{if $\overline{B}=\{\omega\}$}\\
	\Plb &\text{otherwise}
\end{cases}$
and
$U'=\begin{cases}
	\update \Pub {v,\omega} - 1
	&\text{if $\underline{B}=\{\omega\}$}\\
	\Pub &\text{otherwise}
\end{cases}$, with
\begin{align*}
\overline{B}&\DEF\{\omega\in A_{u,v}(\Plb,\Pub) \mid
\Plb_{v,\omega}<\Pub_{v,\omega}\wedge
\forall\omega'\in A_{u,v}(\Plb,\Pub): \\
& \qquad
\Plb_{v,\omega}=\Plb_{v,\omega'}
\wedge (\omega'\preceq_v\omega \vee\omega\parallel_{v}\omega')
\}\\
\underline{B}&\DEF\{\omega\in A_{u,v}(\Plb,\Pub) \mid
\Plb_{v,\omega}<\Pub_{v,\omega}\wedge
\forall\omega'\in A_{u,v}(\Plb,\Pub):\\
&\qquad
\Pub_{v,\omega}=\Pub_{v,\omega'}
\wedge (\omega\preceq_v\omega'\vee\omega\parallel_v\omega')
\}
\end{align*}

The set $A_{u,v}(\Plb,\Pub)$ is the set of regulator states $\omega\in\Omega_v$ for which there exist
parametrisations within $(\Plb,\Pub)$ where changing the value of $u$ changes the value of $v$.
Note that if this set is empty, $u$ has no observable influence: the empty lattice is then returned.

$\overline B$ (resp. $\underline B$) is the set of $\preceq_v$-maximal (resp. $\preceq_v$-minimal) elements $\omega\in A_{u,v}(\Plb,\Pub)$ such that
$\Plb_{v,\omega}<\Pub_{v,\omega}$, or an empty set if regulator states in $A_{u,v}(\Plb,\Pub)$ do not have the same lower (resp. upper) bound value.
Increasing (resp. decreasing) the lower (resp. upper) bound of any $\preceq_v$-maximal (resp. $\preceq_v$-minimal) element in $A_{u,v}(\Plb,\Pub)$ ensures observability of $(u,v)$ while respecting the monotonicity restrictions.
No restriction is made in case several maximal (resp. minimal) regulator states exist, in order to preserve all possible behaviours at the cost of an over-approximation.
Thus, the lower (resp. upper) bound is only modified if a unique $\preceq_v$-maximal (resp. $\preceq_v$-minimal) element exists. Note that $\preceq_v$-maximal (resp. $\preceq_v$-minimal) $\omega$ such that $\Plb_{v,\omega}=\Pub_{v,\omega}$ are excluded since their lower (resp. upper) bound cannot be increased (resp. decreased) any further.

The condition for all lower (resp. upper) bounds of regulator states in $A_{u,v}(\Plb,\Pub)$ to be equal is in place to prevent restrictions if $(u,v)$ is already observable under $\Plb$ (resp. $\Pub$). More precisely, $(u,v)$ is observable under $\Plb$ (resp. $\Pub$) when lower (resp. upper) bounds of elements in $A_{u,v}(\Plb,\Pub)$ differ and there exists a unique $\preceq_v$-maximal (resp. $\preceq_v$-minimal) regulatory state $\omega\in A_{u,v}(\Plb,\Pub)$ such that $\Plb_{v,\omega}<\Pub_{v,\omega}$.
We can, however, assume the existence of such $\omega$ without loss of generality, as no restriction takes place otherwise, regardless of lower (resp. upper) bound equality.

\begin{example}
Consider the PRN from figure \ref{fig:running_example} with the simple modification of node $a$ being boolean.
We enrich the PRN with a set of constraints $R=\{(b,c,\o)\}$ such that only the interaction $(b,c)$ is observable.
Let us assume that the boundary parametrisations $L,U$ have the following values for parameters of node $c$ regulation:
\begin{center}
\begin{tabular}{cc||cc}
$\omega_a$ & $\omega_b$ & $\Plb_{c,\omega}$ & $\Pub_{c,\omega}$\\\hline
0 & 0 & 0 & 1\\
0 & 1 & 0 & 0\\\hline
1 & 0 & 1 & 1\\
1 & 1 & 1 & 1\\
\end{tabular}
\end{center}
As $(b,c)$ interaction is observable, there exists no admissible parametrisation with $P_{c,\langle 00\rangle}=0$
(otherwise the value of $b$ never has an effect on the value of $c$).
Therefore $L_{c,\langle 00\rangle}$ should be changed to $1$.
With our definition, $A_{b,c}(\Plb,\Pub)=\{\langle 00\rangle,\langle 01\rangle\}$,
$\overline B = \{\langle 00\rangle\}$ and $\underline B=\emptyset$.
Observe that since no influences are monotonic in this case, all the elements of $A_{u,v}(\Plb,\Pub)$ are $\preceq_v$-maximal and minimal at the same time.
The observability enforcement, therefore, defaults to the simple mode of choosing the last regulator state with different lower and upper bounds, $\langle 00\rangle$.
$\underline B$ remains empty as not all regulator states in $A_{u,v}(\Plb,\Pub)$ share the same upper bound.

Let us now consider the same example, but with larger set of constraints $R=\{(a,c,+1),(b,c,+1),(b,c,\o)\}$, i.e., $a$ and $b$ now have
a positive-monotonic influence on $c$.
Let us again assume boundary parametrisations:
\begin{center}
\begin{tabular}{cc||cc}
$\omega_a$ & $\omega_b$ & $\Plb_{c,\omega}$ & $\Pub_{c,\omega}$\\\hline
0 & 0 & 0 & 1\\
0 & 1 & 0 & 1\\\hline
1 & 0 & 0 & 1\\
1 & 1 & 0 & 1\\
\end{tabular}
\end{center}
In this case $A_{b,c}(L,U) = \Omega_c$.
Remark that all the lower and upper bounds are identical.
Due to the monotonicity constraints, $\langle 00\rangle$ is the unique $\preceq_c$-minimal element
of $\Omega_c$, and $\langle 11\rangle$ is the unique $\preceq_c$-maximal element.
Hence $\overline{B}=\{\langle 11\rangle\}$ and $\underline B=\{\langle 00\rangle\}$.
Therefore, $L_{c,\langle 11\rangle}$ is set to $1$ and $U_{c,\langle 00\rangle}$ is set to $0$.
\end{example}

\paragraph{Ensure influence constraints $R$}
Finally, given the full set of influence constraints $R$,
the narrowing operator $\nabla_R:\PG^2\to\PG^2$ iteratively applies the operators
$\nabla_r$ for reach $r\in R$ until fixpoint:
\begin{equation}
\restrict_R(\Plb,\Pub) \DEF
\fpite{(\Plb,\Pub)}
    {\left(\bigcirc_{(u,v,r)\in R} \restrict_{(u,v,r)}\right)}
\end{equation}
where
$\bigcirc_{e\in\{e_1,\dots,e_n\}} f_e \DEF
    f_{e_1} \circ \cdots \circ f_{e_n}$.

\begin{remark}
The fixpoint is unique and is reachable in at most $\card{\Omega_v}^{\mv v+1}$ iterations.
\end{remark}

The abstract counterpart $\pabs_R(T)$ of $p_R(T)$ is then defined as follows.
\begin{definition}[$\pabs_R(T)$]
Let $G_m$ be a PRN with well-formed influence constraints $R$.
The \emph{abstraction of the parametrisation set enabling a set of transitions} $T\subseteq\Delta(G_m)$
is defined inductively as follows:
\begin{align*}
\pabs_R(\emptyset) &\DEF \restrict_R(\lb{\PG},\ub{\PG}) \\
\pabs_R(T \cup \{t = x\xrightarrow{v,s}y\}) &\DEF
    \restrict_{\{(u,v',r) \in R\mid v=v'\}}
        \left(\restrict_t(\pabs_R(T))\right)
\end{align*}
\end{definition}

As stated in the following theorem, the defined narrowing operators actually lead to the best possible abstraction of the concrete
$p_R(T)$ by the means of a convex sublattice.
\begin{theorem}
$\pabs_R(T)$
is the smallest convex sublattice including $p_R(T)$
\label{thm:pabsTR}
\end{theorem}
The theorem can be proven using mathematical induction on the size of the set $T$. The induction corresponds to the actual application where restrictions generally happen by adding one transition at a time. Within the induction step, the proof is split into two branches.
First containing the proof of soundness (i.e. the smallest convex sublattice containing $p_R(T)$ is included within $\pabs_R(T)$), once again using induction on application of $\restrict_r$ for individual constraints $r\in R$. The inner induction step details a discussion showing that for every restriction that occurs, the smallest convex sublattice must also be accordingly smaller compared to the one obtained for one transition less.
The second branch proves that $\pabs_R(T)$ is the best over-approximation (i.e. $\pabs_R(T)$ is contained within the smallest convex sublattice containing $p_R(T)$). Here the discussion goes the other way saying that if the smallest convex sublattice with the extra transition is smaller, then a restriction must have occurred to reflect the change in $\pabs_R(T)$.
Together the two branches give the coveted equality.

The proof relies on several properties of the interplay between the parametrisation set and influence constraints.
Namely, a very important property could be referred to as the density of the parametrisation set.
More precisely, given a parametrisation set $p_R(T)$ for some $T$, two arbitrary regulator states $\omega,\omega'\in\Omega_v$ for some $v$ and arbitrary parameter values $k\in\{\lb{p_R(T)}_{v,\omega},\dots,\ub{p_R(T)}_{v,\omega}\}$ for $\omega$ and $l\in\{\lb{p_R(T)}_{v,\omega'},\dots,\ub{p_R(T)}_{v,\omega'}\}$ for $\omega'$, it is only under specific conditions imposed by the constraints in $R$, that no parametrisation $P$ such that $P_{v,\omega}=k$ and $P_{v,\omega'}=l$ belongs to $p_R(T)$.

As the proof contains a considerable amount of technical discussion on different constraint types and constraint-transition relations it has been omitted within this section and is instead given in~\ref{app:proof-best-approx} alongside auxiliary lemmas.



Because of the optimality of the abstraction, one can then derive that
$\pabs_R(T)$ is not the empty lattice if and only if
$p_R(T)$ is not empty:
for any set of transitions $T\subseteq\Delta(G_m)$,
there exists a parametrisation $P$ such that $T\subseteq\Delta(G_m,P)$
if and only if
$\pabs_R(T)\neq\emptylattice \Leftrightarrow p_R(T)\neq\emptyset$.
\begin{corollary}
Given a sequence of transitions $\pi = x\to \dots \to y$ in $\Delta(G_m)$,
\[ \pabs_R(\tilde{\pi}) \neq \emptylattice \Longleftrightarrow
    p_R(\tilde{\pi}) \neq \emptyset \]
\end{corollary}
Contrary to the case without influence constraints, there is no guarantee that each parametrisation
within the convex sublattice is in $p_R(T)$.
However, it is guaranteed that all the parametrisations in $p_R(T)$ are in $\pabs_R(T)$.

\begin{example}
Let us reconsider the example at the end of Sect.~\ref{sec:representation}
for the PRN $G_m$ of Fig.~\ref{fig:running_example} with transitions
$T=\{ 110\xrightarrow{c,+} 111; 111 \xrightarrow{b,-1} 101 \}$.
Recall that
$\pabs(T)=(L,U)=(\langle 00000000100\rangle,\langle 22210111111\rangle)$,
as illustrated by Fig.~\ref{fig:lattice}.

Let us assume the influence constraints $R=\{(a,c,+1),(b,b,\o)\}$.
Considering first the positive-monotonic influence of $a$ on $c$, it results that
$P_{c,\langle 21\rangle}$ should be greater or equal to $P_{c,\langle 11\rangle}$
(in particular $\langle 01\rangle\preceq_c\langle 11\rangle\preceq_c\langle 21\rangle$).
Because $L_{c,\langle 11\rangle}=1$, the operator $\restrict_{(a,c,+1)}(L,U)$ increases the lower
bound $L_{c,\langle 21\rangle}$ from 0 to 1
(parametrisation $\langle 00000000101\rangle$ illustrated in Fig.~\ref{fig:lattice}, the component
$c,\langle 21\rangle$ being marked with $\Sd$).

No further adjustments can be made based on $(a,c,+1)$ at this point, thus the observability constraint $(b,b,\o)$ is applied.
Because $L_{b,\langle 1\rangle}=U_{b,\langle 1\rangle}=0$,
the operator $\restrict_{(b,b,\o)}$ increases the lower bound
$L_{b,\langle 0\rangle}$ from 0 to 1.
At this point we should return to $(a,c,+1)$ to check if the intermediate modifications allow for more restrictions,
however, in this example, the fixed point is reached with $\restrict_r$ being applied only once for both constraints in $R$.

It results that $\pabs_R(T)=(\langle 00010000101\rangle,\langle 22210111111\rangle)$, which is
illustrated by the yellow area in Fig.~\ref{fig:lattice}.
\end{example}

\section{Unfolding Semantics for Parametric Regulatory Networks}
\label{sec:unfolding}

Unfolding semantics \cite{esparza02} for (safe) Petri nets
are used for exploring feasible sequences of transitions without the redundancy of 
investigating different interleavings of the same process
which differ only in the ordering of concurrent transitions.
Here, we say that two distinct transitions $t$ and $t'$ are \emph{concurrent}
if, from any state that enables both of them,
one may fire $t$ followed by $t'$, $t'$ followed by $t$,
or both at the same time, and still reach the same final state.
Unfoldings simply keep track of such concurrent occurrences
by storing them in a partial order built recursively
from a representation of the initial marking by a set of places,
and applying the Petri net dynamics locally;
see below for a formal definition.
The resulting structure is called an \emph{occurrence net}:
a bipartite, acyclic graph with some additional properties,
giving a partial order representation of the net's semantics.

Large biological networks are expected to show a high degree of concurrency, as
the value of each node typically depends on only a few nodes compared to the
size of the network.
Therefore, concurrency-aware methods can enhance greatly the tractability of the
analysis of the reachable state space \cite{chatain14}.

In this section, we introduce an unfolding semantics for Parametric Regulatory
Networks with the aim of exploring the reachable transitions and
associated parametrisation sets while avoiding the combinatorial explosion of
interleaving due to concurrent transitions.

Sect.~\ref{sec:onet} gives the definition of an occurrence net which will be
generated from PRNs unfolding semantics.
Sect.~\ref{sec:def-unfolding} establishes the unfolding semantics of PRNs with
concrete and abstract parametrisation space.
Usually, an unfolding is infinite (as soon as a cycle is possible).
Sect.~\ref{sec:complix} details how a complete finite prefix of this unfolding
can be derived in order to obtain a finite occurrence net from which can be
extracted all reachable states and associated concrete or abstract parametrisation space.

\subsection{Occurrence net: events, conditions, and configurations}
\label{sec:onet}

Here we give a brief definition of \emph{occurence net} as a special type of \emph{event structure}~\cite{nielsen81}.

\begin{definition}
An  \emph{occurrence net} $\mathcal O=\langle E, C, F, C_0\rangle$
is a bipartite acyclic digraph between
\emph{events} $E$ and \emph{conditions} $C$
with edges $F\subseteq (E\times C) \cup (C\times E)$
and set of \emph{initial conditions} $C_0\subseteq C$ on which we define:
\begin{itemize}
\item the \emph{pre-set} and \emph{post-set}
of a node $n\in E\cup C$ as
$\pre n\DEF\{ m \in E\cup C\mid (m,n)\in F\}$
and
$\post n\DEF\{m \in E\cup C\mid (n,m)\in F\}$,
respectively;
\item the \emph{causal relation} $\prec\, \subseteq E\times E$ among events
such that
$e'\prec e \EQDEF$ there exists a non-empty path from $e'$ to $e$ in $F$;
\item the \emph{conflict relation} $\conflict \subseteq E\times E$ among events such that
$e'\conflict e\EQDEF \exists u,v\in E$ s.t. $u\neq v$,  $u=e'\vee u\prec e'$,
and $v=e\vee v\prec e$ with $\pre u\cap\pre v\neq\emptyset$;
\end{itemize}
and which satisfy:
\begin{itemize}
\item for all conditions $c\in C$,
    $\card{\pre c}\leq 1$;
    and $\pre c=\emptyset \Leftrightarrow c\in C_0$;
\item $\forall e\in E$, $\neg(e\conflict e)$.
\end{itemize}

A set of events $\cfg\subseteq E$ is a \emph{configuration}
if and only if
for all events $e\in\cfg$,
$\{ e'\in E \mid e'<e\}\subseteq\cfg$
and for all events $e,e'\in\cfg$,
$\neg(e\conflict e')$.
Given an event $e\in E$, we denote the \emph{minimal configuration} containing $e$ by
$\mincfg e=\{e'\in E\mid e'\leq e\}$.
\end{definition}

We extend pre- and post-set notations to sets of nodes:
for any $N\subseteq E\cup C$,
$\pre N\DEF\cup_{n\in N}\pre n$
and
$\post N\DEF\cup_{n\in N}\post n$.

\subsubsection{Application to PRNs}

Let us assume a PRN $G_m$ with $G=(V,I)$ and influence constraints $R$.

A \emph{condition} $c\in C$ is characterised by a triplet
$\langle e,v,j\rangle$
where $e\in E\cup\{\bot\}$ is the parent event of $c$,
or $\bot$ if $c\in C_0$,
$v\in V$ is a node of the PRN,
and $j\in\domv v$ one of its possible values.

An \emph{event} $e\in E$ corresponds to the increase or decrease of a node
value,
and is characterised by a triplet
$\langle C',v,s\rangle$
where
$C'\subseteq C$ is the set of parents conditions of $e$, referred to as
\emph{pre-conditions},
$v\in V$ is a node of the PRN,
and $s\in\{+1,-1\}$,
which satisfy:
\begin{itemize}
\item $\card{C'}=\card{\{v\}\cup\innodes v}$,
i.e., the number of pre-conditions is the number of regulators of $v$ plus $v$
itself;
\item $\forall u\in \innodes v$, $\exists \langle e',u',j\rangle\in C'$
with $u=u'$,
i.e., for each regulator $u$ of $v$ there is a corresponding condition in $C'$;
\item $\exists \langle e',v',j\rangle\in C'$ with $v=v'$
and such that
$j+s \in\domv v$,
i.e., there is a pre-condition corresponding to a value of $v$ which allows the
change by $s$.
\end{itemize}
We denote by $\omega_v(C') \DEF\prod_{u\in\innodes v} \{ j \mid (e',u,j) \in C'\} \in\Omega_v$ the
state of regulators of $v$.
Remark that
from every state $x\in S(G_m)$ of the PRN $G_m$ where
$\omega_v(x) = \omega_v(C')$
and
$x_v = j$
with $\langle e',v,j\rangle\in C'$,
$t = x\xrightarrow{v,s}\subst x v {x_v+s}$ 
is a transition of the PRN.
Moreover, remark that since the value of $s$ is fixed by $e$, any such $t$ has the same $\P_t$ and $\restrict_t$ as both values depend solely on $\omega_v(C')$ contrary to the whole state.
Therefore, we use $\P_e$ and $\restrict_e$ to refer to the common values of $\P_t$ and $\restrict_t$ respectively.

Subsequently, given any configuration $\cfg$, we use
$p(\cfg)$, $p_R(\cfg)$, $\pabs(\cfg)$, and $\pabs_R(\cfg)$
to denote the corresponding
$p(T)$, $p_R(T)$, $\pabs(T)$, and $\pabs_R(T)$.

Finally, we denote the terminal set of conditions of a configuration $\cfg\subseteq E$ as $\cut\cfg = (C_0 \cup \post\cfg)\setminus\pre\cfg$.
Remark that due to the nature of the events, there is a unique condition in the $\cut\cfg$ for every node of the PRN.
This allows us to define the state reached by application of a configuration $\cfg\subseteq E$ as $X(\cfg)\DEF x\in S(G_m)$ such that
$\forall \langle e,v,j\rangle\in\cut\cfg, x_v=j$.

\subsection{Unfolding of Parametric Regulatory Networks}
\label{sec:def-unfolding}

Given an initial state $x^0\in S(G_m)$ of the PRN $G_m$ with influence
constraints $R$,
its unfolding is the unique maximal occurrence net
$\mathcal U=\langle E,C,F,C_0\rangle$
with $C_0 = \{ \langle\bot,v,x^0_v\rangle\mid v\in V\}$,
$C_0 \subseteq C$, and
such that for any event $e\in E$,
$p_R(\mincfg e) \neq \emptyset$~\cite{engelfriet91}.

The unfolding $\mathcal U$ is typically infinite, and its set of events $E$ and
conditions $C$ can be defined inductively as follows:
\begin{enumerate}[label=(\roman*)]
\item Start with $C:=C_0$ and $E:=\emptyset$.
\item An event $e=\langle C',v,s\rangle$ is a possible extension of the
unfolding if and only if $C'\subseteq C$, 
$\mincfg e$ is a configuration, and $p_R(\mincfg e)\neq\emptyset$.
In such a case, $e$ is added as child of each condition $c\in C'$,
together with
new conditions $\langle e,u,j\rangle$ for each $\langle e',u,j\rangle\in C'$, $u\neq v$,
and the condition
$\langle e,v,j+s\rangle$ where $\langle e',v,j\rangle\in C'$,
all being children of $e$.
\end{enumerate}

Remark that $p_R(\mincfg e)$ can be computed inductively as:
\[
p_R(\mincfg{e=\langle C',v,s\rangle}) =
\P_e \cap \left(\bigcap_{c\in C'} p_R(\mincfg{\pre c})\right)
\]
Therefore, one can store with each event $e$ its parametrisation space
$p_R(\mincfg e)$ and re-use it when computing causally related events.

Similarly, we can relax $p_R$ with its abstraction $\pabs_R$.
Also remark that $\pabs_R(\mincfg e)$ can be computed inductively as:
\[
\pabs_R(\mincfg{e=\langle C',v,s\rangle}) =
    \restrict_{\{(u,v',r) \in R\mid v=v'\}}
        \left(\restrict_e
\left(\bigcap_{c\in C'} \pabs_R(\mincfg{\pre c})\right)
        \right)
\]
where
$(L,U)\cap(L',U') \DEF (\max(L,L'),\min(U,U'))$.

\medskip

The construction ensures that for each sequence of transitions $\pi=x^0\to\dots\to y$
which is realisable for the concrete semantics of PRNs with constraints $R$,
i.e., such that $p_R(\toset \pi)\neq\emptyset$,
there exists a configuration $\cfg$ of $\mathcal U$ composed of the
corresponding events and such that $X(\cfg)=y$ and
$p_R(\cfg)\neq\emptyset$, or equivalently, $\pabs_R(\cfg)\neq\emptylattice$.

\tikzstyle{condition}=[circle,draw]
\tikzstyle{event}=[rectangle,draw]
\tikzstyle{cutoff event}=[event,dashed]
\tikzstyle{initial condition}=[condition,fill=blue!30]
\begin{example}
Fig.~\ref{fig:unfolding} shows a partial unfolding of the PRN $G_m$ from Fig.~\ref{fig:running_example}
with the influence constraints $R=\{(a,a,-),(b,b,-),(a,c,+),(b,c,+)\}$,
i.e., $a$ and $b$ auto-inhibit themselves and both activate $c$.
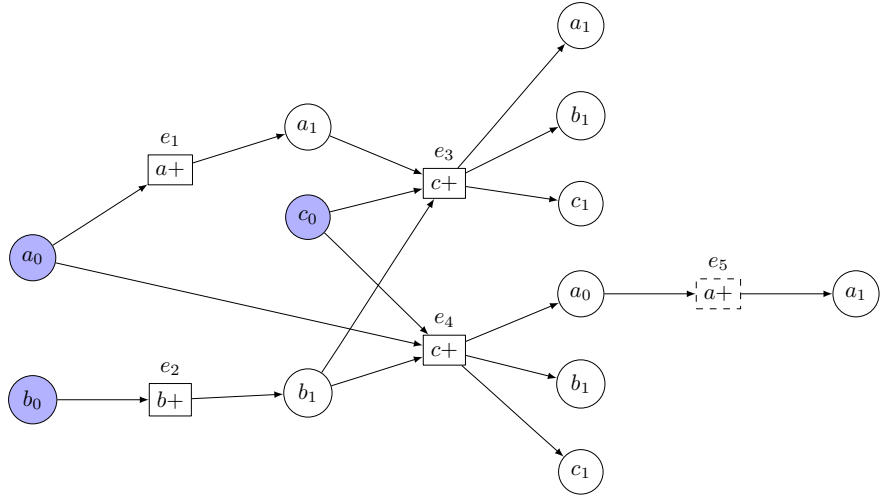
\begin{figure}[t]
\centering
\scalebox{0.85}{
\begin{tikzpicture}[>=latex,line join=bevel,]
\node (c22) at (254.0bp,129.5bp) [condition] {$c_1$};
  \node (c23) at (254.0bp,168.5bp) [condition] {$b_1$};
  \node (c0) at (11.0bp,105.5bp) [initial condition] {$a_0$};
  \node (c24) at (376.0bp,89.5bp) [condition] {$a_1$};
  \node (c11) at (254.0bp,89.5bp) [condition] {$a_0$};
  \node (e14) at (315.0bp,89.5bp) [cutoff event,label=above:$e_5$] {$a+$};
  \node (e0) at (72.0bp,144.5bp) [event,label=above:$e_1$] {$a+$};
  \node (c13) at (254.0bp,49.5bp) [condition] {$b_1$};
  \node (c21) at (254.0bp,208.5bp) [condition] {$a_1$};
  \node (c12) at (254.0bp,10.5bp) [condition] {$c_1$};
  \node (c3) at (133.0bp,163.5bp) [condition] {$a_1$};
  \node (c2) at (133.0bp,123.5bp) [initial condition] {$c_0$};
  \node (e7) at (193.5bp,138.5bp) [event,label=above:$e_3$] {$c+$};
  \node (e6) at (193.5bp,64.5bp) [event,label=above:$e_4$] {$c+$};
  \node (e1) at (72.0bp,42.5bp) [event,label=above:$e_2$] {$b+$};
  \node (c1) at (11.0bp,42.5bp) [initial condition] {$b_0$};
  \node (c4) at (133.0bp,45.5bp) [condition] {$b_1$};
  \draw [->] (e14) ..controls (336.55bp,89.5bp) and (346.06bp,89.5bp)  .. (c24);
  \draw [->] (e1) ..controls (93.218bp,43.528bp) and (103.27bp,44.039bp)  .. (c4);
  \draw [->] (e6) ..controls (216.24bp,44.434bp) and (228.43bp,33.178bp)  .. (c12);
  \draw [->] (c0) ..controls (50.057bp,96.856bp) and (130.44bp,78.596bp)  .. (e6);
  \draw [->] (e7) ..controls (215.39bp,149.21bp) and (226.16bp,154.74bp)  .. (c23);
  \draw [->] (e6) ..controls (214.93bp,59.26bp) and (224.98bp,56.684bp)  .. (c13);
  \draw [->] (e7) ..controls (214.85bp,135.37bp) and (224.76bp,133.84bp)  .. (c22);
  \draw [->] (e6) ..controls (215.02bp,73.269bp) and (225.2bp,77.621bp)  .. (c11);
  \draw [->] (e7) ..controls (215.7bp,163.87bp) and (229.46bp,180.33bp)  .. (c21);
  \draw [->] (c1) ..controls (28.894bp,42.5bp) and (38.913bp,42.5bp)  .. (e1);
  \draw [->] (c4) ..controls (148.23bp,68.121bp) and (166.03bp,96.425bp)  .. (e7);
  \draw [->] (c3) ..controls (150.61bp,156.4bp) and (160.87bp,152.02bp)  .. (e7);
  \draw [->] (c2) ..controls (150.53bp,127.74bp) and (160.65bp,130.33bp)  .. (e7);
  \draw [->] (c4) ..controls (150.41bp,50.829bp) and (160.89bp,54.233bp)  .. (e6);
  \draw [->] (c2) ..controls (149.01bp,108.37bp) and (161.89bp,95.371bp)  .. (e6);
  \draw [->] (e0) ..controls (93.914bp,151.24bp) and (103.99bp,154.48bp)  .. (c3);
  \draw [->] (c11) ..controls (272.02bp,89.5bp) and (281.69bp,89.5bp)  .. (e14);
  \draw [->] (c0) ..controls (28.166bp,116.19bp) and (39.237bp,123.5bp)  .. (e0);
\end{tikzpicture}}
\caption{Excerpt of the unfolding of the PRN from Fig.~\ref{fig:running_example} with monotonic constraints $R=\{(a,a,-),(b,b,-),(a,c,+),(b,c,+)\}$.
Conditions are drawn as circles and are labelled with the corresponding node value;
events are drawn as boxes and are labelled with the corresponding node value increase or decrease.
Initial conditions are filled with light blue.
Dashed event ($e_5$) will be declared as \emph{cut-off} (Sect.~\ref{sec:complix}).
}
\label{fig:unfolding}
\end{figure}

First, notice that events $e_1$ and $e_2$ are concurrent: one can apply them in any order and end in the
state $\langle a=1,b=1,c=0\rangle$.
In a classical state graph computation, this would generate 4 transitions ($a+$ then $b+$; $b+$ then
$a+$) instead of these two events.
The states reached by each of these two events are
$X(\{e_1\}) = \langle 1,0,0\rangle$
and
$X(\{e_2\}) = \langle 0,1,0\rangle$.

Let us consider the configuration $\mincfg{e_3} = \{e_1,e_2,e_3\}$.
The abstract parametrisation space
$\pabs_R(\mincfg{e_3})=(L^{e_3},U^{e_3})$
results in
$L^{e_3}_{a,\langle 0\rangle}=1$ (due to $e_1$),
$L^{e_3}_{b,\langle 0\rangle}=U^{e_3}_{b,\langle 0\rangle}=1$ (due to $e_2$),
$L^{e_3}_{c,\langle 1,1\rangle}=U^{e_3}_{b,\langle 1,1\rangle}=1$ (due to $e_3$),
and
$L^{e_3}_{c,\langle 2,1\rangle}=U^{e_3}_{b,\langle 2,1\rangle}=1$ (due to $(a,c,+)\in R$).

Then, let us consider the configuration $\mincfg{e_4}=\{e_2,e_4\}$.
The abstract parametrisation space
$\pabs_R(\mincfg{e_4})=(L^{e_4},U^{e_4})$
results in
$L^{e_3}_{b,\langle 0\rangle}=U^{e_3}_{b,\langle 0\rangle}=1$ (due to $e_2$),
$L^{e_3}_{c,\langle 0,1\rangle}=U^{e_3}_{b,\langle 0,1\rangle}=1$ (due to $e_4$),
and,
$L^{e_3}_{c,\langle 1,1\rangle}=U^{e_3}_{b,\langle 1,1\rangle}=L^{e_3}_{c,\langle 2,1\rangle}=U^{e_3}_{b,\langle 2,1\rangle}=1$ (due to $(a,c,+)\in R$).

It results that $\pabs_R(\mincfg{e_3})$ and  $\pabs_R(\mincfg{e_4})$ are incomparable:
indeed, whereas the first has more restrictions on the parameters of node $a$,
the latter has more constraints on the parameter $P_{c,\langle 0,1\rangle}$, i.e., when $b$ is
active and not $a$.
Intuitively, the configuration $\mincfg{e_4}$ corresponds to the case when the logic of $c$
activation is a disjunction between $a$ and $b$; where $\mincfg{e_3}$
matches with an \emph{and} logic for $c$ activation, but has observed an increase of $a$.

Finally, the extension of $\mincfg{e_4}$ with $e_5$ leads to
a parametrisation space included in $\pabs_R(\mincfg{e_3})$:
indeed,
$\pabs_R(\mincfg{e_5})$ refines the lower bound of $\pabs_R(\mincfg{e_4})$
for the parameter $P_{a,\langle 0\rangle}$, similarly to $\pabs_R(\mincfg{e_1})$.
Moreover, remark that $X(\mincfg{e_3})=X(\mincfg{e_5})=\langle 1,1,1\rangle$.
Therefore, to any extension of $\mincfg{e_5}$ corresponds an equivalent extension of $\mincfg{e_3}$.
\end{example}

\subsection{Complete finite prefix}
\label{sec:complix}

In the general case, the unfolding of a PRN is infinite.
As the unfolding is a representation of all the processes of the network and the number of states is
finite, there exist finite prefixes of the unfolding from which all the
configurations can be reconstructed, and in particular all reachable states can be recovered.
As in~\cite{esparza02},
we refer to such a prefix as a \emph{complete finite prefix} (CFP), and show below a possible construction.

Our construction follows the same principle as the construction of CFP for safe
Petri nets \cite{esparza02}, with the additional care of parametrisation spaces.
Essentially, the main idea is to detect during the construction
configurations from which can be derived the equivalent extensions.
In such cases, only one configuration should be extended, and the others
stopped: their last event is marked as a \emph{cut-off}.
As demonstrated in \cite{esparza02}, the completeness of a prefix can be guaranteed as soon as the
computation of extensions and cut-offs is performed in a specific order,
so-called \emph{total adequate order}.

In the remainder of this section,
we extend the total adequate order and cut-off used for safe Petri net unfolding to
the PRN unfolding, which results in an algorithm for the CFP of PRN unfolding.

\subsubsection{A total adequate ordering of PRN configurations}


We construct a total adequate order of PRN unfolding configurations based on the total adequate order over configurations of Petri net unfolding as introduced in~\cite{esparza02}.
The Petri net total adequate order relies on recording the number of instances of each Petri net transition in the configuration in a structure similar to a Parikh vector.
Furthermore, Foata normal forms are used to refine the records with respect to causality in cases where records of two configurations are identical.
Our approach differs solely in recording the number of node $v$ value changes per regulatory context $\omega$, for each $(v,\omega)\in\Omega$ instead of keeping record for each Petri net transition.


We define $\varphi(\cfg)$ as the Parikh vector associated to a configuration $\cfg\subseteq E$
as a $\card{\Omega}$ dimensional vector, where, for each $(v,\omega)\in\Omega$,
we associate the number of corresponding events in $\cfg$:
\[
\varphi(\cfg)_{v,\omega} \DEF \card{\{\langle C',v,s\rangle\in\cfg\mid \omega = \omega_v(C')\}}
\]

The Foata normal form serves to distinguish between configurations based on causal constraints.
Before the definition of Foata normal form, we introduce a partition of a configuration $\cfg$ into causal layers defined iteratively as follows:
\begin{enumerate}[label=(\roman*)]
\item $E^\cfg_1\DEF\{e\in \cfg\mid\forall e'\in \cfg:\neg(e'\prec e)\}$
\item For $1<i\in\mathbb{N}:E^\cfg_i=\{e\in\cfg\setminus\bigcup_{j<i}E^\cfg_j\mid\forall e'\in \cfg:e'\prec e\Rightarrow e'\in\bigcup_{j<i}E^\cfg_j\}$
\end{enumerate}
The Foata normal form of a configuration $\cfg$ is defined as a vector $FC(\cfg)\DEF(\varphi(E^\cfg_1),\dots,\varphi(E^\cfg_k))$,
where $k\in\mathbb{N}$ is the larget natural number such that $E^\cfg_k\neq\emptyset$.
Such $k$ is guaranteed to exist as $\cfg$ is finite.
Intuitively, the Foata normal form $FC(\cfg)$ is a layered representation
of $\cfg$ in respect to causality relation and represents steps in which events of
$\cfg$ can fire if all concurrent events fire synchronously.

We can then define the total ordering $\cfglt$ over configurations of $\mathcal U$ as follows,
where we use the lexicographic order $<$ to compare Parikh vectors and Foata normal forms.

\begin{definition}[$\cfglt \subset \powerset E\times\powerset E$]
Let $\mathcal{U}=\langle E,C,F,C_0\rangle$ be the unfolding of PRN $G_m$ and let
$\cfg_1,\cfg_2\subseteq E$ be two finite configurations of $\mathcal{U}$.
We say that $\cfg_1\cfglt \cfg_2$ iff one of the following conditions holds:
    \begin{itemize}
        \item[] $|\cfg_1|<|\cfg_2|$,
        \item[] $|\cfg_1|=|\cfg_2|\wedge\varphi(\cfg_1)\palt\varphi(\cfg_2)$,
        \item[] $|\cfg_1|=|\cfg_2|\wedge\varphi(\cfg_1)=\varphi(\cfg_2)\wedge FC(\cfg_1)\palt FC(\cfg_2)$.
    \end{itemize}
\end{definition}
\begin{property}
$\cfglt$ is an adequate order, i.e.,
\begin{itemize}
\item[] $\cfglt$ is well-founded,
\item[] $\cfg_1\subset\cfg_2$ implies $\cfg_1\cfglt\cfg_2$, and,
\item[] if $\cfg_1\cfglt\cfg_2$ and $X_{\cfg_1}=X_{\cfg_2}$, then for any extension $e$ of
$\cfg_1$ and $f$ of $\cfg_2$ such that $e\equiv f$,
$\cfg_1 \cup\{e\} \cfglt\cfg_2\cup\{f\}$
where $\langle C'_1,v_1,s_1\rangle\equiv\langle C'_2,v_2,s_2\rangle
\EQDEF v_1 = v_2 \wedge s_1 = s_2\wedge \omega_{v_1}(C'_1) = \omega_{v_2}(C'_2)$.
\end{itemize}
\end{property}

The total adequate order $\cfglt$ introduced here is identical to the total adequate order used by Esparza et al.~\cite{esparza02}.
The only difference is the use of regulator states $\langle v,\omega\rangle$ instead of transitions for Parikh vectors,
which has ultimately no influence on the properties of the order itself.

\subsubsection{Cut-offs}

Let us consider two events $e,e'$ of the unfolding such that
$X(\mincfg e)=X(\mincfg{e'})$ and $\pabs(\mincfg e)\subseteq\pabs(\mincfg {e'})$.
Let us assume there exists an event $f=\langle C_1,v,s\rangle$ being an extension of $\mincfg e$,
i.e., $C_1\subset\cut{\mincfg e}$, or equivalently $\omega_v(C_1)=\omega_v(X(\mincfg e))$,
and $\pabs(\mincfg f)\neq\emptylattice$.
It derives that there exists an isomorphic event $f'=\langle C_2,v,s\rangle$
with $\omega_v(C_2)=\omega_v(C_1)=\omega_v(X(\mincfg{e'}))$ being an extension of $\mincfg{e'}$,
$\pabs(\mincfg{f'})\neq\emptylattice$, and $X(\mincfg f)=X(\mincfg{f'})$.

Therefore, every extension of $\mincfg{e}$ has a counterpart extension of $\mincfg{e'}$.
It is then sufficient to compute the extension of the $\cfglt$-smallest of two events to preserve the
completeness of the reachable states \cite{esparza02}.

This leads to the definition of a cut-off event during PRN unfolding which extends the usual
definition for Petri nets with the additional requirement of inclusion of parametrisation space.
\begin{definition}[Cut-off]
    An event $e\in E$ is considered a cut-off event if there exists a different
    event $e'\in E$ such that:
\begin{itemize}
\item[] $X(\mincfg e)=X(\mincfg{e'})$,
\item[] $\pabs_R(\mincfg e)\subseteq \pabs_R(\mincfg{e'})$.
\end{itemize}
\end{definition}

\begin{example}
In Fig.~\ref{fig:unfolding},
the event $e_5$ is a cut-off due to $e_3$ as
$X(\mincfg{e_3})=(1,1,1)=X(\mincfg{e_5})$ and the parametrisation sets
$\pabs_R(\mincfg{e_5})=(\langle 10010010101\rangle,\langle 22211111111\rangle)$ and
$\pabs_R(\mincfg{e_3})=(\langle 10010000101\rangle,\langle 22211111111\rangle)$
giving us $\pabs_R(\mincfg{e_5})\subseteq\pabs_R(\mincfg{e_3})$.
As explained at the end of Sect.~\ref{sec:def-unfolding},
$L_{c,\langle 01\rangle}=1$ (7th position in the vector) enforced by $e_4$,
requires $L_{c,\langle 11\rangle}=1$ (9th position) due to influence $(a,c)$ being positive-monotonic.
The same monotonicity constraint is also responsible for $L_{c,\langle 21\rangle}=1$ (11th position) in both cases.

Note that $\lb{e_3}\cfglt\lb{e_5}$ does not necessarily have to hold. Even if $\lb{e_5}\cfglt\lb{e_3}$ holds, event $e_5$
is declared cut-off once $e_3$ is added to the unfolding.
\end{example}

\subsubsection{Algorithm}

The computation of the CFP extends the one of the unfolding of Sect.~\ref{sec:def-unfolding} by
visiting the candidate extensions
following $\cfglt$ order, and by declaring cut-off events from which no extension is possible.

Following our unfolding semantics, an event $e$ is an extension of an occurrence net
only if $\mincfg e$ is a configuration and $\pabs_R(\mincfg e)\neq\emptylattice$
(or equivalently $p(\mincfg e)\neq\emptyset$).
Additionally, in the case of CFP computation, we require that $e$ is not an extension of a cut-off event.

The CFP $\cfp=\langle E,C,F,C_0,\coff\rangle$ is inductively defined as follows, starting with $C:=C_0$, $E:=\emptyset$,
$\coff:=\emptyset$.\\
Repeat until no possible extension of $\cfp$ exists:
\begin{enumerate}[label=(\roman*)]
\item Let $e=\langle{C',v,s}\rangle$ be the $\cfglt$-smallest extension of $\cfp$
\item For each $e'\in E, e'\neq e$ s.t. $X(\mincfg {e'})=X(\mincfg{e})$
\begin{itemize}
\item if $\pabs_R(\mincfg{e})\subseteq\pabs_R(\mincfg{e'})$,
    mark $e$ as a cut-off ($\coff=\coff\cup\{e\}$)
\item if $\pabs_R(\mincfg{e'})\subsetneq\pabs_R(\mincfg{e})$,
    mark $e'$ as a cut-off ($\coff=\coff\cup\{e'\}$)
\end{itemize}
\end{enumerate}

The latter statement takes care of declaring cut-off events $e'$ due to the newly added extension $e$.
This case can occur as the 
total adequate order $\cfglt$ does not correlate with set inclusion order over parametrisation spaces.
In other words, $\cfg_1\cfglt\cfg_2$ does not guarantee $p(\cfg_1)\subseteq p(\cfg_2)$ and conversely.
Hence, by \emph{a posteriori} declaring $e'$ cut-off, we ensure that none of its extensions will be
considered, as they are redundant.

Following \cite{esparza02},
as $\cfglt$ is a total adequate order, $\cfp$ is complete, i.e., any configuration of the unfolding
$\mathcal U$ can be reconstructed from $\cfp$.
In particular, any state reachable by a configuration of $\mathcal U$ is reachable by a
configuration of $\cfp$.

\begin{example}
Fig.~\ref{fig:prefix} gives a complete finite prefix of the PRN $G_m$ from
Fig.~\ref{fig:running_example}, with
$R=\{(a,a,-),(b,b,-),(a,c,+),(b,c,+),(a,a,\o),(b,b,\o),(a,c,\o),(b,c,\o)\}$, i.e.,
$a$ and $b$ auto-inhibit themselves, both activate $c$, and all influences are observable.
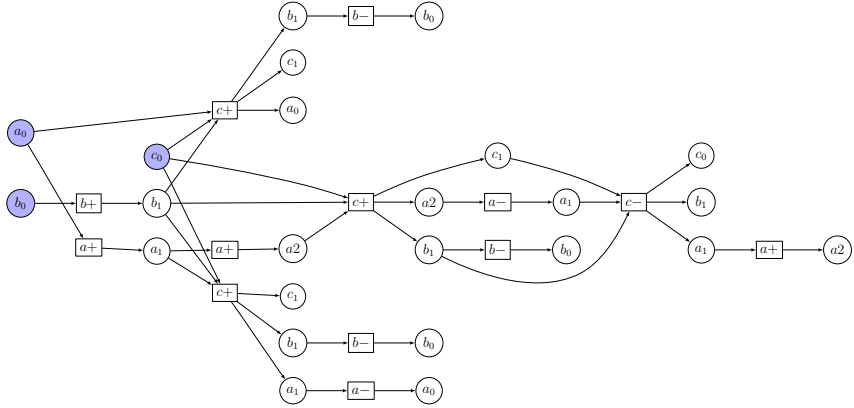
\begin{figure}[tp]
    \centering
    \scalebox{0.45}{\begin{tikzpicture}[>=latex,line join=bevel,font=\large]
\node (c59) at (581.0bp,129.0bp) [condition] {$a_1$};
  \node (c58) at (581.0bp,169.0bp) [condition] {$b_1$};
  \node (c37) at (353.0bp,51.0bp) [condition] {$b_0$};
  \node (c1) at (11.0bp,168.0bp) [initial condition] {$b_0$};
  \node (c57) at (581.0bp,208.0bp) [condition] {$c_0$};
  \node (c0) at (11.0bp,227.0bp) [initial condition] {$a_0$};
  \node (e32) at (410.5bp,129.0bp) [event] {$b-$};
  \node (e31) at (410.5bp,169.0bp) [event] {$a-$};
  \node (e10) at (296.0bp,169.0bp) [event] {$c+$};
  \node (e15) at (296.0bp,325.0bp) [event] {$b-$};
  \node (e39) at (524.5bp,169.0bp) [event] {$c-$};
  \node (e0) at (68.0bp,131.0bp) [event] {$a+$};
  \node (c13) at (239.0bp,325.0bp) [condition] {$b_1$};
  \node (c12) at (239.0bp,286.0bp) [condition] {$c_1$};
  \node (c11) at (239.0bp,246.0bp) [condition] {$a_0$};
  \node (c33) at (410.5bp,208.0bp) [condition] {$c_1$};
  \node (c32) at (353.0bp,169.0bp) [condition] {$a2$};
  \node (e24) at (296.0bp,51.0bp) [event] {$b-$};
  \node (c35) at (353.0bp,11.0bp) [condition] {$a_0$};
  \node (c47) at (468.0bp,129.0bp) [condition] {$b_0$};
  \node (c34) at (353.0bp,129.0bp) [condition] {$b_1$};
  \node (c3) at (125.0bp,128.0bp) [condition] {$a_1$};
  \node (c2) at (125.0bp,207.0bp) [initial condition] {$c_0$};
  \node (e7) at (182.0bp,93.0bp) [event] {$c+$};
  \node (e6) at (182.0bp,246.0bp) [event] {$c+$};
  \node (e1) at (68.0bp,168.0bp) [event] {$b+$};
  \node (c6) at (239.0bp,130.0bp) [condition] {$a2$};
  \node (e3) at (182.0bp,129.0bp) [event] {$a+$};
  \node (c4) at (125.0bp,168.0bp) [condition] {$b_1$};
  \node (c22) at (239.0bp,90.0bp) [condition] {$c_1$};
  \node (c23) at (239.0bp,51.0bp) [condition] {$b_1$};
  \node (c46) at (468.0bp,169.0bp) [condition] {$a_1$};
  \node (c21) at (239.0bp,11.0bp) [condition] {$a_1$};
  \node (e44) at (638.0bp,129.0bp) [event] {$a+$};
  \node (e22) at (296.0bp,11.0bp) [event] {$a-$};
  \node (c25) at (353.0bp,325.0bp) [condition] {$b_0$};
  \node (c67) at (695.0bp,129.0bp) [condition] {$a2$};
  \draw [->] (c3) ..controls (143.24bp,128.31bp) and (153.06bp,128.49bp)  .. (e3);
  \draw [->] (c33) ..controls (433.95bp,202.69bp) and (458.63bp,196.4bp)  .. (479.0bp,189.0bp) .. controls (487.96bp,185.75bp) and (497.63bp,181.49bp)  .. (e39);
  \draw [->] (e15) ..controls (312.64bp,325.0bp) and (322.89bp,325.0bp)  .. (c25);
  \draw [->] (e7) ..controls (196.48bp,72.906bp) and (214.54bp,45.981bp)  .. (c21);
  \draw [->] (c1) ..controls (28.851bp,168.0bp) and (39.083bp,168.0bp)  .. (e1);
  \draw [->] (c4) ..controls (140.87bp,147.72bp) and (158.18bp,124.11bp)  .. (e7);
  \draw [->] (c2) ..controls (141.99bp,218.34bp) and (154.16bp,226.97bp)  .. (e6);
  \draw [->] (c34) ..controls (383.9bp,113.27bp) and (439.15bp,88.941bp)  .. (479.0bp,109.0bp) .. controls (496.9bp,118.01bp) and (509.4bp,138.0bp)  .. (e39);
  \draw [->] (e22) ..controls (313.03bp,11.0bp) and (322.89bp,11.0bp)  .. (c35);
  \draw [->] (c23) ..controls (256.85bp,51.0bp) and (267.08bp,51.0bp)  .. (e24);
  \draw [->] (c2) ..controls (156.23bp,202.24bp) and (207.68bp,193.64bp)  .. (250.0bp,183.0bp) .. controls (258.94bp,180.75bp) and (268.72bp,177.77bp)  .. (e10);
  \draw [->] (c0) ..controls (26.612bp,201.5bp) and (46.163bp,167.37bp)  .. (e0);
  \draw [->] (e6) ..controls (198.55bp,246.0bp) and (208.64bp,246.0bp)  .. (c11);
  \draw [->] (c4) ..controls (162.66bp,168.22bp) and (240.25bp,168.68bp)  .. (e10);
  \draw [->] (c6) ..controls (256.46bp,141.67bp) and (268.32bp,150.08bp)  .. (e10);
  \draw [->] (e44) ..controls (655.03bp,129.0bp) and (664.89bp,129.0bp)  .. (c67);
  \draw [->] (e10) ..controls (313.64bp,156.9bp) and (325.88bp,147.99bp)  .. (c34);
  \draw [->] (c2) ..controls (132.24bp,194.61bp) and (134.23bp,190.63bp)  .. (136.0bp,187.0bp) .. controls (149.22bp,159.98bp) and (164.26bp,128.4bp)  .. (e7);
  \draw [->] (c13) ..controls (256.85bp,325.0bp) and (267.08bp,325.0bp)  .. (e15);
  \draw [->] (e32) ..controls (427.45bp,129.0bp) and (438.11bp,129.0bp)  .. (c47);
  \draw [->] (c32) ..controls (371.48bp,169.0bp) and (381.55bp,169.0bp)  .. (e31);
  \draw [->] (e0) ..controls (85.11bp,130.12bp) and (95.108bp,129.58bp)  .. (c3);
  \draw [->] (e7) ..controls (199.64bp,80.29bp) and (211.88bp,70.945bp)  .. (c23);
  \draw [->] (c34) ..controls (371.09bp,129.0bp) and (381.56bp,129.0bp)  .. (e32);
  \draw [->] (e6) ..controls (199.64bp,258.1bp) and (211.88bp,267.01bp)  .. (c12);
  \draw [->] (e3) ..controls (199.03bp,129.29bp) and (208.89bp,129.47bp)  .. (c6);
  \draw [->] (c59) ..controls (599.24bp,129.0bp) and (609.06bp,129.0bp)  .. (e44);
  \draw [->] (e39) ..controls (541.99bp,180.8bp) and (554.12bp,189.48bp)  .. (c57);
  \draw [->] (e10) ..controls (312.55bp,169.0bp) and (322.64bp,169.0bp)  .. (c32);
  \draw [->] (c46) ..controls (486.37bp,169.0bp) and (496.06bp,169.0bp)  .. (e39);
  \draw [->] (e24) ..controls (312.64bp,51.0bp) and (322.89bp,51.0bp)  .. (c37);
  \draw [->] (c0) ..controls (49.637bp,231.23bp) and (126.7bp,239.89bp)  .. (e6);
  \draw [->] (e39) ..controls (541.9bp,156.96bp) and (553.84bp,148.2bp)  .. (c59);
  \draw [->] (e6) ..controls (196.4bp,265.25bp) and (214.21bp,290.82bp)  .. (c13);
  \draw [->] (e1) ..controls (84.639bp,168.0bp) and (94.893bp,168.0bp)  .. (c4);
  \draw [->] (c4) ..controls (140.97bp,189.23bp) and (158.54bp,214.15bp)  .. (e6);
  \draw [->] (e31) ..controls (427.76bp,169.0bp) and (437.85bp,169.0bp)  .. (c46);
  \draw [->] (e10) ..controls (314.77bp,177.71bp) and (329.13bp,184.36bp)  .. (342.0bp,189.0bp) .. controls (357.92bp,194.75bp) and (376.47bp,199.86bp)  .. (c33);
  \draw [->] (c21) ..controls (257.24bp,11.0bp) and (267.06bp,11.0bp)  .. (e22);
  \draw [->] (e7) ..controls (198.64bp,92.147bp) and (208.89bp,91.588bp)  .. (c22);
  \draw [->] (e39) ..controls (541.0bp,169.0bp) and (551.16bp,169.0bp)  .. (c58);
  \draw [->] (c3) ..controls (142.38bp,117.58bp) and (154.08bp,110.13bp)  .. (e7);
\end{tikzpicture}}
    \caption{The complete finite prefix obtained by unfolding the PRN of
    Fig.~\ref{fig:running_example} with
    $R=\{(a,a,-),(b,b,-),(a,c,+),(b,c,+),(a,a,\o),(b,b,\o),(a,c,\o),(b,c,\o)\}$.
    Cut-off events are not represented.}
    \label{fig:prefix}
\end{figure}

The completeness property, coupled with the result on optimal abstraction of parametrisation space
(Theorem~\ref{thm:pabsTR}), ensures that:
\begin{itemize}
\item for any configuration of prefix, there exists a sequence of transitions in $\Delta(G_m)$ realisable
with respect to the concrete semantics of PRNs (Def.~\ref{def:concrete-PRNs}), i.e., there exists a
parametrisation $P\in\PG$ such that all the transitions are in $\Delta(G_m,P)$.
\item for any parametrisation $P\in\PG$, for any realisable sequence of transitions in $\Delta(G_m,P)$, one can
reconstruct from the prefix (with the cut-off events) a configuration $\cfg$ which contains the
corresponding events and such that $P\in\pabs_R(\cfg)$.
\end{itemize}
\end{example}

Standard complete finite prefixes of Petri nets computed using a total adequate order for extensions
have a number of non-cut-off events which does not exceed the number of reachable states~\cite{esparza02}.
This claim does not hold in our setting, because several events with the same state can exist in our
CFP of PBNs (the cut-offs depend also on the parametrisation space).
However, because of the resulting partial ordering of transitions in the CFP,
one can easily argue that the number of configurations in the CFP is smaller than the
number of concrete traces.
Future work may consider defining an ordering encompassing both configurations and
parametrisation space to avoid redundant exploration of configurations for the prefix computation.


\section{Experiments}
\label{sec:experiments}

Algorithms presented in previous sections have been implemented in a prototype tool \emph{Pawn}
written in Python.\footnote{Pawn is available online: \url{https://github.com/GeorgeKolcak/Pawn}} In
this section, we provide its experimental evaluation performed on several well-known Boolean and
multi-valued regulatory networks that have been studied in the literature. This study extends the
preliminary evaluation provided in~\cite{Kolcak-SASB16}. 

\subsection{Experiment Description}

Several parametrised models were selected varying in size of the network, in
average connectivity of nodes, and in the network type (Boolean vs.
multi-valued). Each experiment is conducted in the way that for a given initial
state the size of its unfolding is computed with respect to full parameter space
provided that all regulations are considered monotonic and observable. The size of the unfolding is characterised by the number of non-cut-off events. This number gives a good figure of the effect of compaction achieved. In models where the reachable state space is sensitive to the initial state, we re-run the experiment for different initial states. By default, we have considered initial states as set in the original model.

To clarify the compaction achieved with unfoldings, we compare the size of the unfoldings with the size of the complete symbolic execution tree achieved from the same initial state. To this end, we employ the tool \sputnik~\cite{gallet14} that implements automata-based LTL model checking of parametric regulatory networks by (finite) symbolic execution of the product automaton. \sputnik{} explicitly traverses the product states in DFS manner while symbolically executing the transitions representing constraints on parameters. To achieve exactly the reachable states of the model state transition graph, we use a B\"uchi automaton with a single state looping over an atomic proposition satisfied in every state of the model.

\sputnik{} implements an additional parameter constraint called {\it Min-Max}. It states that in a state where all the activators (resp. inhibitors) are enabled and all of the inhibitors (resp. activators)
are disabled at the same time, then the regulated node must be at its maximum (resp.
minimum) level. Apparently, in our parameter encoding, it means that the only valid parameter context for such a state is the maximal (or minimal) in the respective component. To this end, we have also included the {\it Min-Max} constraint in \tool.

\subsection{Models}

\begin{figure}[t]
\scalebox{.35}{\begin{tikzpicture}
\usetikzlibrary{shapes,snakes,arrows.meta,calc}
\tikzstyle{every node}=[draw,very thick,ellipse,minimum width=100pt,font=\bf\LARGE,align=center];
\tikzset{edge/.style={-{Stealth[scale=2.5,length=5,width=3.5]},line width=.7pt}};

\begin{scope}[local bounding box=scope1]
\node (coup) { Coup\_fti};
\node[left=200pt,below=50pt] (sp8) { Sp8};
\node[right=200pt,below=100pt,fill=cyan!30] (fgf) { Fgf8};
\node[left=200pt,below=150pt] (pax) { Pax6};
\node[below=200pt] (emx) { Emx2};

\draw[edge] (coup) to node [above,draw=none,near start] { $o-$} (pax); 
\draw[edge] (coup) to node [above,draw=none,near start] {~~~~~~$o+$} (emx); 

\draw[edge] (sp8) to node [above,draw=none,midway] { $o+$~~~~} (pax); 
\draw[edge,bend right=5] (sp8) to node [below,draw=none,midway] { $o-$} (emx); 
\draw[edge,bend right=5] (sp8) to node [below,draw=none,near start] { $o+$} (fgf);
\draw[edge] (sp8) to node [above,draw=none,midway] { $o-$} (coup);

\draw[edge] (fgf) to node [above,draw=none,midway] { $o-$} (coup);  
\draw[edge,bend right=5] (fgf) to node [above,draw=none,near start] { $o+$} (sp8);
\draw[edge,bend right=5] (fgf) to node [above,draw=none,midway] { $o-$} (emx);  
\draw[edge] (fgf) edge[loop right,edge] node [above,draw=none,near start] {$o+$} ();

\draw[edge,bend right=5] (emx) to node [below,draw=none,midway] { $o-$} (fgf); 
\draw[edge,bend right=5] (emx) to node [above,draw=none,midway] { $o-$} (pax); 
\draw[edge,bend right=5] (emx) to node [above,draw=none,midway] { ~~$o-$} (sp8);

\draw[edge,bend right=5] (pax) to node [below,draw=none,midway] { $o-$} (emx); 
\end{scope}

\begin{scope}[shift={($(scope1.east)+(2cm,3)$)}]
\node (cI) {cI, [0..2]};
\node[right=200pt] (cro) {cro, [0..3]};
\node[below=150pt] (cII) {cII, [0..1]};
\node[right=250pt,below=150pt] (N) {N, [0..1]};

\draw[edge,bend right=5] (cI) to node [below,draw=none,near start] {$o-$,$2$} (cro);
\draw[edge,bend right=10] (cI) to node [below,draw=none,near start] {$o-$,$2$~~~~~~~~} (cII);
\draw[edge] (cI) to node [below=5pt,draw=none,near start] {$o-$,$1$} (N);
\draw[edge] (cI) edge[loop left,edge] node [below,draw=none,near start] {$o+$,$2$} ();

\draw[edge,bend right=5] (cro) to node [above,draw=none] {$o-$,$1$} (cI);
\draw[edge,bend left=5] (cro) to node [draw=none,midway] {~~~~~~~~$o-$,$2$} (N);
\draw[edge] (cro) to node [below,draw=none,near start] {$o-$,$3$} (cII);
\draw[edge] (cro) edge[loop right,edge] node [above,draw=none,near start] {$o-$,$3$} ();

\draw[edge,bend left=5] (N) to node [below,draw=none,midway] {$o+$,$1$} (cII);

\draw[edge,bend right=10] (cII) to node [above,draw=none,near start] {~~~~~~~~$o+$,$1$} (cI);
\end{scope}

\begin{scope}[shift={($(scope1.south)+(2cm,-2)$)}]
\node[fill=cyan!30] (tnfa) {tnfa};
\node[right=100pt,fill=cyan!30] (egf) {egf};

\node[left=100pt,below=60pt] (ikb) {ikb};
\node[below=75pt,right=20pt] (pi3k) {pi3k};
\node[below=70pt,right=175pt] (sos) {sos};

\node[below=150pt,left=50pt] (nfkb) {nfkb};
\node[below=150pt,right=20pt] (gsk3) {gsk3};
\node[below=150pt,right=175pt] (map3) {map3k1};
\node[below=150pt,right=300pt] (raf) {raf1};

\node[below=225pt,left=15pt] (p38) {p38};
\node[below=225pt,right=20pt] (ap1) {ap1};
\node[below=225pt,right=175pt] (creb) {creb};
\node[below=225pt,right=335pt] (erk) {erk};

\draw[edge] (tnfa) to node [above,draw=none,midway] { $o-$} (ikb);
\draw[edge] (tnfa) to node [above,draw=none,midway] { $o+$~~} (p38);
\draw[edge] (tnfa) to node [above,draw=none,midway] { ~~$o+$} (pi3k); 

\draw[edge] (egf) to node [above,draw=none,midway] { $o+$} (pi3k);
\draw[edge] (egf) to node [above,draw=none,midway] { ~~$o+$} (sos);

\draw[edge,bend right=20] (ikb) to node [above,draw=none,near end] { $o-$~~~~} (nfkb);
\draw[edge,bend right=20] (nfkb) to node [above,draw=none,near start] { ~~~~$o+$} (ikb);

\draw[edge] (pi3k) to node [above,draw=none,near end] { ~~~~$o-$} (gsk3);
\draw[edge] (pi3k) to node [above,draw=none,midway] { ~$o+$} (map3);

\draw[edge] (sos) to node [above,draw=none,near end] { ~~~~$o+$} (map3);
\draw[edge] (sos) to node [above,draw=none,midway] { ~$o+$} (raf);

\draw[edge] (map3) to node [above,draw=none,midway] { ~~$o+$} (p38);
\draw[edge] (map3) to node [below,draw=none,midway] { ~~$o+$} (ap1);

\draw[edge] (raf) to node [above,draw=none,midway] { ~~$o+$} (creb);
\draw[edge] (raf) to node [above,draw=none,near end] { ~~~~$o+$} (erk);

\draw[edge,bend left=250] (erk) to node [above,draw=none,midway] { ~~$o-$} (sos);
\end{scope}

\end{tikzpicture}}
\caption{\textbf{(left)} A Boolean regulatory network controlling the
	 cortical area development. The state marked in blue has been set to initial value $1$ in one of the experiments. \textbf{(right)} A multi-valued regulatory network of bacteriophage $\lambda$ life cycle. Node labels are extended with ranges describing the value domain. Accordingly, edge labels include threshold levels. \textbf{(bottom)} Model of the signalling pathway of EGF-TNF$\alpha$. The only two states that start with initial value~$1$ are marked in blue.}
    \label{fig:bio_models}
\end{figure}
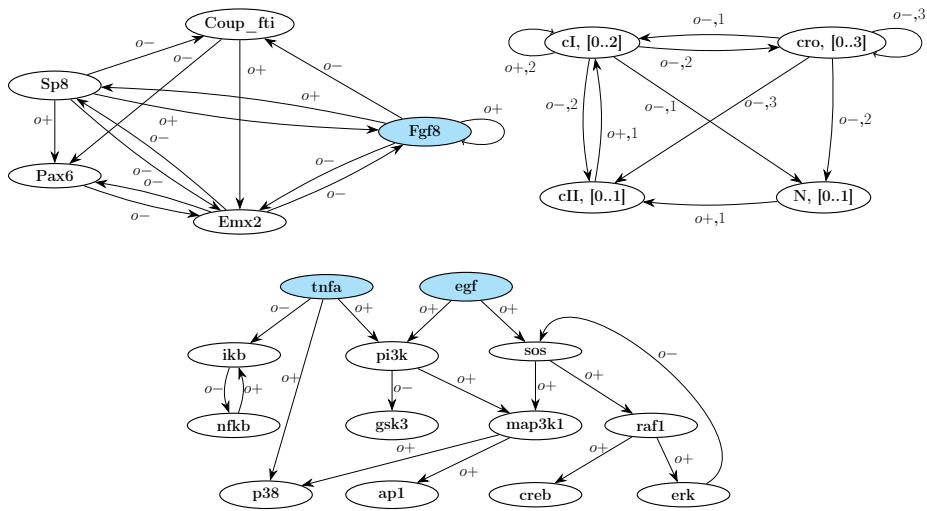

In all considered models, all regulations are defined with observability and monotonicity restrictions. The additional {\it Min-Max} constraint is employed only when explicitly noted.  

First, we use a Boolean model of the gene regulatory network underlying mammalian cortical area development~\cite{giacomantonio10}. The network is depicted in Figure~\ref{fig:bio_models} (left). The unfolding has been computed with respect to two different initial states --- all species inactive (Fgf8=0) and all species inactive with the only exception of Fgf8 (Fgf8=1). 

The smallest multi-valued model we have analysed is the well-studied regulatory network of bacteriophage $\lambda$ life cycle~\cite{thieffry95} ($\lambda$-switch) also analysed in~\cite{gallet14,me12}. The network structure is shown in Figure~\ref{fig:bio_models} (right). The initial state is $0$ for all nodes. The model is considered in two configurations, with and without the {\it Min-Max} constraint.

As an example of a larger Boolean model, we consider a model of EGF-TNF$\alpha$ signalling
pathway~\cite{macnamara12,Caspots-BioSystems16} (Figure \ref{fig:bio_models} (bottom)).
In this case the initial state is set to $\mathit{tnfa}$ and $\mathit{egf}$ nodes active whereas all other nodes are considered inactive.

We have also considered two larger multi-valued models ($>10$ nodes). First, we have analysed a model published in~\cite{mbodj13}. It represents several key signalling pathways of Drosophila including cross-talks. The network has the size of $15$ nodes and its structure is shown in Figure~\ref{fig:bio_models2} (left).

Second, we have analysed a model describing the control of the
developmental process in primary sex determination of placental
mammals. The model has been recently published
in~\cite{sanchez16}. The network is multi-valued and has $14$ nodes
but in contrast to the Drosophila model, it is highly
interconnected. In Figure~\ref{fig:bio_models2} (right) there is shown its basic topology including the information about considered initial states. 

\begin{figure}[h]
\begin{minipage}{.5\textwidth}
\scalebox{.4}{
        {\LARGE
\begin{tikzpicture}[>=latex',line join=bevel,]
  \pgfsetlinewidth{1bp}
\pgfsetcolor{black}
\definecolor{init}{rgb}{0.752941, 0.878431, 1.000000}
\definecolor{init2}{rgb}{1, 0.878431, 0.752941}
  \draw [->] (352.2bp,71.697bp) .. controls (343.4bp,63.135bp) and (332.62bp,52.656bp)  .. (315.59bp,36.104bp);
  \draw [->] (229.49bp,143.7bp) .. controls (240.96bp,134.88bp) and (255.06bp,124.03bp)  .. (275.76bp,108.1bp);
  \draw [->] (171.0bp,287.7bp) .. controls (171.0bp,279.98bp) and (171.0bp,270.71bp)  .. (171.0bp,252.1bp);
  \draw [->] (54.283bp,291.44bp) .. controls (57.193bp,290.2bp) and (60.136bp,289.02bp)  .. (63.0bp,288.0bp) .. controls (122.21bp,266.85bp) and (142.93bp,273.47bp)  .. (215.72bp,248.56bp);
  \draw [->] (217.33bp,143.8bp) .. controls (228.22bp,125.77bp) and (246.0bp,96.692bp)  .. (262.0bp,72.0bp) .. controls (267.9bp,62.899bp) and (274.53bp,53.043bp)  .. (286.14bp,36.06bp);
  \draw [->] (298.0bp,71.697bp) .. controls (298.0bp,63.983bp) and (298.0bp,54.712bp)  .. (298.0bp,36.104bp);
  \draw [->] (287.91bp,291.83bp) .. controls (265.23bp,280.81bp) and (232.64bp,264.97bp)  .. (198.45bp,248.35bp);
  \draw [->] (225.2bp,287.7bp) .. controls (216.4bp,279.14bp) and (205.62bp,268.66bp)  .. (188.59bp,252.1bp);
  \draw [->] (54.085bp,291.83bp) .. controls (76.769bp,280.81bp) and (109.36bp,264.97bp)  .. (143.55bp,248.35bp);
  \draw [->] (298.0bp,143.7bp) .. controls (298.0bp,135.98bp) and (298.0bp,126.71bp)  .. (298.0bp,108.1bp);
  \draw [->] (179.9bp,215.7bp) .. controls (184.0bp,207.73bp) and (188.95bp,198.1bp)  .. (198.2bp,180.1bp);
  \draw [->] (188.8bp,287.7bp) .. controls (197.6bp,279.14bp) and (208.38bp,268.66bp)  .. (225.41bp,252.1bp);
  \draw [->] (359.91bp,291.83bp) .. controls (337.23bp,280.81bp) and (304.64bp,264.97bp)  .. (270.45bp,248.35bp);
  \draw [->] (116.8bp,287.7bp) .. controls (125.6bp,279.14bp) and (136.38bp,268.66bp)  .. (153.41bp,252.1bp);
  \draw [->] (243.0bp,287.7bp) .. controls (243.0bp,279.98bp) and (243.0bp,270.71bp)  .. (243.0bp,252.1bp);
  \draw [->] (201.92bp,36.189bp) .. controls (206.35bp,60.424bp) and (210.09bp,104.89bp)  .. (210.26bp,143.87bp);
  \draw [->] (126.09bp,291.83bp) .. controls (148.77bp,280.81bp) and (181.36bp,264.97bp)  .. (215.55bp,248.35bp);
  \draw [->] (243.0bp,359.7bp) .. controls (243.0bp,351.98bp) and (243.0bp,342.71bp)  .. (243.0bp,324.1bp);
  \draw [->] (201.09bp,143.87bp) .. controls (196.66bp,119.67bp) and (192.92bp,75.211bp)  .. (192.74bp,36.189bp);
  \draw [->] (273.05bp,71.876bp) .. controls (259.96bp,62.893bp) and (243.74bp,51.763bp)  .. (221.05bp,36.19bp);
  \draw [->] (297.2bp,287.7bp) .. controls (288.4bp,279.14bp) and (277.62bp,268.66bp)  .. (260.59bp,252.1bp);
  \draw [->] (234.1bp,215.7bp) .. controls (230.0bp,207.73bp) and (225.05bp,198.1bp)  .. (215.8bp,180.1bp);
\begin{scope}
  \definecolor{strokecol}{rgb}{0.0,0.0,0.0};
  \pgfsetstrokecolor{strokecol}
  \draw (223.0bp,36.0bp) -- (169.0bp,36.0bp) -- (169.0bp,0.0bp) -- (223.0bp,0.0bp) -- cycle;
  \draw (196.0bp,18.0bp) node {Dad};
\end{scope}
\begin{scope}
  \definecolor{strokecol}{rgb}{0.0,0.0,0.0};
  \pgfsetstrokecolor{strokecol}
  \draw[fill=init] (54.0bp,324.0bp) -- (0.0bp,324.0bp) -- (0.0bp,288.0bp) -- (54.0bp,288.0bp) -- cycle;
  \draw (27.0bp,306.0bp) node {Punt};
\end{scope}
\begin{scope}
  \definecolor{strokecol}{rgb}{0.0,0.0,0.0};
  \pgfsetstrokecolor{strokecol}
  \draw[fill=init] (198.0bp,324.0bp) -- (144.0bp,324.0bp) -- (144.0bp,288.0bp) -- (198.0bp,288.0bp) -- cycle;
  \draw (171.0bp,306.0bp) node {Gbb};
\end{scope}
\begin{scope}
  \definecolor{strokecol}{rgb}{0.0,0.0,0.0};
  \pgfsetstrokecolor{strokecol}
  \draw[fill=init] (270.0bp,324.0bp) -- (216.0bp,324.0bp) -- (216.0bp,288.0bp) -- (270.0bp,288.0bp) -- cycle;
  \draw (243.0bp,306.0bp) node {Sog};
\end{scope}
\begin{scope}
  \definecolor{strokecol}{rgb}{0.0,0.0,0.0};
  \pgfsetstrokecolor{strokecol}
  \draw[fill=init2] (342.0bp,324.0bp) -- (288.0bp,324.0bp) -- (288.0bp,288.0bp) -- (342.0bp,288.0bp) -- cycle;
  \draw (315.0bp,306.0bp) node {Dpp};
\end{scope}
\begin{scope}
  \definecolor{strokecol}{rgb}{0.0,0.0,0.0};
  \pgfsetstrokecolor{strokecol}
  \draw (414.0bp,324.0bp) -- (360.0bp,324.0bp) -- (360.0bp,288.0bp) -- (414.0bp,288.0bp) -- cycle;
  \draw (387.0bp,306.0bp) node {Tsg};
\end{scope}
\begin{scope}
  \definecolor{strokecol}{rgb}{0.0,0.0,0.0};
  \pgfsetstrokecolor{strokecol}
  \draw (327.0bp,36.0bp) -- (269.0bp,36.0bp) -- (269.0bp,0.0bp) -- (327.0bp,0.0bp) -- cycle;
  \draw (298.0bp,18.0bp) node {Targets};
\end{scope}
\begin{scope}
  \definecolor{strokecol}{rgb}{0.0,0.0,0.0};
  \pgfsetstrokecolor{strokecol}
  \draw (270.0bp,252.0bp) -- (216.0bp,252.0bp) -- (216.0bp,216.0bp) -- (270.0bp,216.0bp) -- cycle;
  \draw (243.0bp,234.0bp) node {Tkv};
\end{scope}
\begin{scope}
  \definecolor{strokecol}{rgb}{0.0,0.0,0.0};
  \pgfsetstrokecolor{strokecol}
  \draw (198.0bp,252.0bp) -- (144.0bp,252.0bp) -- (144.0bp,216.0bp) -- (198.0bp,216.0bp) -- cycle;
  \draw (171.0bp,234.0bp) node {Sax};
\end{scope}
\begin{scope}
  \definecolor{strokecol}{rgb}{0.0,0.0,0.0};
  \pgfsetstrokecolor{strokecol}
  \draw (325.0bp,108.0bp) -- (271.0bp,108.0bp) -- (271.0bp,72.0bp) -- (325.0bp,72.0bp) -- cycle;
  \draw (298.0bp,90.0bp) node {Brk};
\end{scope}
\begin{scope}
  \definecolor{strokecol}{rgb}{0.0,0.0,0.0};
  \pgfsetstrokecolor{strokecol}
  \draw[fill=init] (270.0bp,396.0bp) -- (216.0bp,396.0bp) -- (216.0bp,360.0bp) -- (270.0bp,360.0bp) -- cycle;
  \draw (243.0bp,378.0bp) node {Tld};
\end{scope}
\begin{scope}
  \definecolor{strokecol}{rgb}{0.0,0.0,0.0};
  \pgfsetstrokecolor{strokecol}
  \draw[fill=init] (397.0bp,108.0bp) -- (343.0bp,108.0bp) -- (343.0bp,72.0bp) -- (397.0bp,72.0bp) -- cycle;
  \draw (370.0bp,90.0bp) node {Nej};
\end{scope}
\begin{scope}
  \definecolor{strokecol}{rgb}{0.0,0.0,0.0};
  \pgfsetstrokecolor{strokecol}
  \draw[fill=init] (126.0bp,324.0bp) -- (72.0bp,324.0bp) -- (72.0bp,288.0bp) -- (126.0bp,288.0bp) -- cycle;
  \draw (99.0bp,306.0bp) node {Scw};
\end{scope}
\begin{scope}
  \definecolor{strokecol}{rgb}{0.0,0.0,0.0};
  \pgfsetstrokecolor{strokecol}
  \draw (241.0bp,180.0bp) -- (173.0bp,180.0bp) -- (173.0bp,144.0bp) -- (241.0bp,144.0bp) -- cycle;
  \draw (207.0bp,162.0bp) node {MadMed};
\end{scope}
\begin{scope}
  \definecolor{strokecol}{rgb}{0.0,0.0,0.0};
  \pgfsetstrokecolor{strokecol}
  \draw[fill=init] (325.0bp,180.0bp) -- (271.0bp,180.0bp) -- (271.0bp,144.0bp) -- (325.0bp,144.0bp) -- cycle;
  \draw (298.0bp,162.0bp) node {Shn};
\end{scope}
\end{tikzpicture}

}
}
\end{minipage}
\begin{minipage}{.5\textwidth}
  \scalebox{.4}{
        {\LARGE
\begin{tikzpicture}[>=latex',line join=bevel,]
  \pgfsetlinewidth{1bp}
\definecolor{init}{rgb}{0.752941, 0.878431, 1.000000}
\pgfsetcolor{black}
  \draw [->] (111.1bp,71.697bp) .. controls (107.0bp,63.728bp) and (102.05bp,54.1bp)  .. (92.797bp,36.104bp);
  \draw [->] (202.24bp,247.49bp) .. controls (212.02bp,247.78bp) and (220.0bp,243.28bp)  .. (220.0bp,234.0bp) .. controls (220.0bp,228.2bp) and (216.88bp,224.27bp)  .. (202.24bp,220.51bp);
  \draw [->] (116.8bp,359.7bp) .. controls (125.6bp,351.14bp) and (136.38bp,340.66bp)  .. (153.41bp,324.1bp);
  \draw [->] (81.099bp,36.104bp) .. controls (77.812bp,44.129bp) and (73.023bp,53.871bp)  .. (62.815bp,71.697bp);
  \draw [->] (73.596bp,143.7bp) .. controls (69.725bp,135.73bp) and (65.048bp,126.1bp)  .. (56.308bp,108.1bp);
  \draw [->] (202.12bp,224.09bp) .. controls (229.98bp,214.74bp) and (274.74bp,199.21bp)  .. (331.0bp,176.09bp);
  \draw [->] (50.983bp,71.697bp) .. controls (54.294bp,63.644bp) and (59.098bp,53.894bp)  .. (69.308bp,36.104bp);
  \draw [->] (92.774bp,180.19bp) .. controls (108.23bp,204.85bp) and (136.82bp,250.46bp)  .. (160.26bp,287.87bp);
  \draw [->] (28.251bp,359.97bp) .. controls (31.751bp,312.29bp) and (41.672bp,177.18bp)  .. (46.729bp,108.31bp);
  \draw [->] (225.2bp,359.7bp) .. controls (216.4bp,351.14bp) and (205.62bp,340.66bp)  .. (188.59bp,324.1bp);
  \draw [->] (171.0bp,359.7bp) .. controls (171.0bp,351.98bp) and (171.0bp,342.71bp)  .. (171.0bp,324.1bp);
  \draw [->] (147.88bp,216.94bp) .. controls (134.89bp,207.93bp) and (119.53bp,196.53bp)  .. (98.956bp,180.21bp);
  \draw [->] (286.0bp,143.7bp) .. controls (286.0bp,135.98bp) and (286.0bp,126.71bp)  .. (286.0bp,108.1bp);
  \draw [->] (109.24bp,179.14bp) .. controls (122.17bp,188.12bp) and (137.44bp,199.45bp)  .. (157.93bp,215.7bp);
  \draw [->] (224.25bp,36.462bp) .. controls (222.17bp,66.09bp) and (210.17bp,128.05bp)  .. (198.0bp,180.0bp) .. controls (195.99bp,188.56bp) and (193.61bp,197.85bp)  .. (187.61bp,215.88bp);
  \draw [->] (54.085bp,363.83bp) .. controls (76.769bp,352.81bp) and (109.36bp,336.97bp)  .. (143.55bp,320.35bp);
  \draw [->] (171.96bp,215.88bp) .. controls (173.61bp,205.52bp) and (177.18bp,192.05bp)  .. (180.0bp,180.0bp) .. controls (190.94bp,133.33bp) and (201.73bp,78.572bp)  .. (213.26bp,36.462bp);
  \draw [->] (166.07bp,287.7bp) .. controls (165.72bp,279.98bp) and (166.02bp,270.71bp)  .. (168.13bp,252.1bp);
  \draw [->] (97.638bp,359.85bp) .. controls (94.697bp,322.83bp) and (87.734bp,235.18bp)  .. (83.369bp,180.23bp);
  \draw [->] (166.3bp,215.68bp) .. controls (156.29bp,196.78bp) and (138.55bp,166.23bp)  .. (118.0bp,144.0bp) .. controls (107.84bp,133.01bp) and (95.18bp,122.66bp)  .. (75.082bp,108.14bp);
  \draw [->] (385.24bp,175.49bp) .. controls (395.02bp,175.78bp) and (403.0bp,171.28bp)  .. (403.0bp,162.0bp) .. controls (403.0bp,156.2bp) and (399.88bp,152.27bp)  .. (385.24bp,148.51bp);
  \draw [->] (340.2bp,143.7bp) .. controls (331.4bp,135.14bp) and (320.62bp,124.66bp)  .. (303.59bp,108.1bp);
  \draw [->] (179.92bp,252.1bp) .. controls (180.28bp,259.79bp) and (179.99bp,269.05bp)  .. (177.9bp,287.7bp);
  \draw [->] (111.06bp,31.793bp) .. controls (126.88bp,40.752bp) and (145.7bp,54.29bp)  .. (156.0bp,72.0bp) .. controls (180.26bp,113.71bp) and (180.55bp,171.76bp)  .. (177.22bp,215.63bp);
  \draw [->] (313.24bp,103.49bp) .. controls (323.02bp,103.78bp) and (331.0bp,99.281bp)  .. (331.0bp,90.0bp) .. controls (331.0bp,84.199bp) and (327.88bp,80.268bp)  .. (313.24bp,76.513bp);
  \draw [->] (270.18bp,71.697bp) .. controls (262.43bp,63.22bp) and (252.96bp,52.864bp)  .. (237.64bp,36.104bp);
\begin{scope}
  \definecolor{strokecol}{rgb}{0.0,0.0,0.0};
  \pgfsetstrokecolor{strokecol}
  \draw[fill=init] (126.0bp,396.0bp) -- (72.0bp,396.0bp) -- (72.0bp,360.0bp) -- (126.0bp,360.0bp) -- cycle;
  \draw (99.0bp,378.0bp) node {Wt1};
\end{scope}
\begin{scope}
  \definecolor{strokecol}{rgb}{0.0,0.0,0.0};
  \pgfsetstrokecolor{strokecol}
  \draw[fill=init]  (75.0bp,108.0bp) -- (21.0bp,108.0bp) -- (21.0bp,72.0bp) -- (75.0bp,72.0bp) -- cycle;
  \draw (48.0bp,90.0bp) node {Dmrt1};
\end{scope}
\begin{scope}
  \definecolor{strokecol}{rgb}{0.0,0.0,0.0};
  \pgfsetstrokecolor{strokecol}
  \draw (147.0bp,108.0bp) -- (93.0bp,108.0bp) -- (93.0bp,72.0bp) -- (147.0bp,72.0bp) -- cycle;
  \draw (120.0bp,90.0bp) node {AF};
\end{scope}
\begin{scope}
  \definecolor{strokecol}{rgb}{0.0,0.0,0.0};
  \pgfsetstrokecolor{strokecol}
  \draw (313.0bp,180.0bp) -- (259.0bp,180.0bp) -- (259.0bp,144.0bp) -- (313.0bp,144.0bp) -- cycle;
  \draw (286.0bp,162.0bp) node {IW};
\end{scope}
\begin{scope}
  \definecolor{strokecol}{rgb}{0.0,0.0,0.0};
  \pgfsetstrokecolor{strokecol}
  \draw[fill=init]  (109.0bp,180.0bp) -- (55.0bp,180.0bp) -- (55.0bp,144.0bp) -- (109.0bp,144.0bp) -- cycle;
  \draw (82.0bp,162.0bp) node {Sf1};
\end{scope}
\begin{scope}
  \definecolor{strokecol}{rgb}{0.0,0.0,0.0};
  \pgfsetstrokecolor{strokecol}
  \draw[fill=init]  (198.0bp,396.0bp) -- (144.0bp,396.0bp) -- (144.0bp,360.0bp) -- (198.0bp,360.0bp) -- cycle;
  \draw (171.0bp,378.0bp) node {AS};
\end{scope}
\begin{scope}
  \definecolor{strokecol}{rgb}{0.0,0.0,0.0};
  \pgfsetstrokecolor{strokecol}
  \draw[fill=init]  (202.0bp,252.0bp) -- (148.0bp,252.0bp) -- (148.0bp,216.0bp) -- (202.0bp,216.0bp) -- cycle;
  \draw (175.0bp,234.0bp) node {Sox9};
\end{scope}
\begin{scope}
  \definecolor{strokecol}{rgb}{0.0,0.0,0.0};
  \pgfsetstrokecolor{strokecol}
  \draw (111.0bp,36.0bp) -- (57.0bp,36.0bp) -- (57.0bp,0.0bp) -- (111.0bp,0.0bp) -- cycle;
  \draw (84.0bp,18.0bp) node {Foxl2};
\end{scope}
\begin{scope}
  \definecolor{strokecol}{rgb}{0.0,0.0,0.0};
  \pgfsetstrokecolor{strokecol}
  \draw[fill=init]  (385.0bp,180.0bp) -- (331.0bp,180.0bp) -- (331.0bp,144.0bp) -- (385.0bp,144.0bp) -- cycle;
  \draw (358.0bp,162.0bp) node {Fgf9};
\end{scope}
\begin{scope}
  \definecolor{strokecol}{rgb}{0.0,0.0,0.0};
  \pgfsetstrokecolor{strokecol}
  \draw[fill=init]  (270.0bp,396.0bp) -- (216.0bp,396.0bp) -- (216.0bp,360.0bp) -- (270.0bp,360.0bp) -- cycle;
  \draw (243.0bp,378.0bp) node {Y};
\end{scope}
\begin{scope}
  \definecolor{strokecol}{rgb}{0.0,0.0,0.0};
  \pgfsetstrokecolor{strokecol}
  \draw[fill=init]  (54.0bp,396.0bp) -- (0.0bp,396.0bp) -- (0.0bp,360.0bp) -- (54.0bp,360.0bp) -- cycle;
  \draw (27.0bp,378.0bp) node {Gata4};
\end{scope}
\begin{scope}
  \definecolor{strokecol}{rgb}{0.0,0.0,0.0};
  \pgfsetstrokecolor{strokecol}
  \draw (198.0bp,324.0bp) -- (144.0bp,324.0bp) -- (144.0bp,288.0bp) -- (198.0bp,288.0bp) -- cycle;
  \draw (171.0bp,306.0bp) node {Sry};
\end{scope}
\begin{scope}
  \definecolor{strokecol}{rgb}{0.0,0.0,0.0};
  \pgfsetstrokecolor{strokecol}
  \draw[fill=init]  (249.0bp,36.0bp) -- (195.0bp,36.0bp) -- (195.0bp,0.0bp) -- (249.0bp,0.0bp) -- cycle;
  \draw (222.0bp,18.0bp) node {b\_cat};
\end{scope}
\begin{scope}
  \definecolor{strokecol}{rgb}{0.0,0.0,0.0};
  \pgfsetstrokecolor{strokecol}
  \draw[fill=init]  (313.0bp,108.0bp) -- (259.0bp,108.0bp) -- (259.0bp,72.0bp) -- (313.0bp,72.0bp) -- cycle;
  \draw (286.0bp,90.0bp) node {Wnt4};
\end{scope}
\end{tikzpicture}

}
     }
\end{minipage}

    \caption{\textbf{(left)} Multi-valued network of signalling
      pathways cross-talk in Drosophila. The states that start with
      initial value $1$ (resp. $2$) are marked in blue
      (resp. red). All nodes considered initally non-zero have been
      set to their maximal level. \textbf{(right)} Multi-valued
      network of mammalian primary sex development. The states with
      initial value $1$ are marked in blue. All the unmarked nodes are initiated $0$. Edge labels are ommitted for sake of simplicity.}
    \label{fig:bio_models2}
\end{figure}
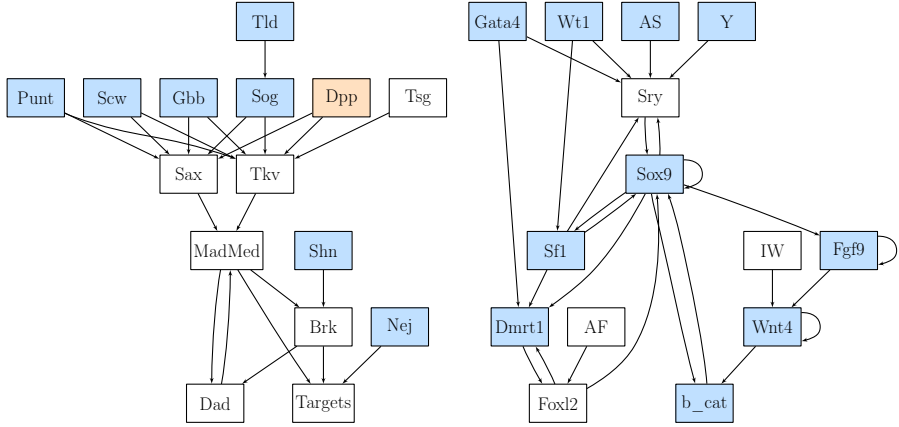 

\subsection{Results}

Computations conducted on all the defined models have led to results shown in
Table~\ref{tab:results}. Unfoldings constructed by \tool{} are characterised by
their size with and without cut-off events. The number of symbolically executed
states computed by \sputnik{} is given for comparison. 

Since both tools are implemented as prototypes without any optimisations, we do not include computation times but rather focus on space which is crucial in this case. However, in all models with the only exception of the Primary Sex Determination model, \tool{} has computed the results in a couple of minutes. In case of the Primary Sex Determination model, \tool{} computed the unfolding in 2 hours whereas sputnik{} has been stopped in 3 days without achieving results.
In case of the Drosophila model, \sputnik{} has been stopped after 2 days of computations whereas \tool{} needed a couple of minutes to compute the unfolding. 
\sputnik{} reached a symbolic execution tree of size at least $7,000,000$ before being timed out in all three relevant cases.

\begin{table}
    \centering
	\begin{scriptsize}
    \begin{tabular}{|l|c|c|r|r|}
        \hline
        Model (init. state) & Type & $\#$ nodes & $\#$ events (incl. cut-offs) & Sym. exec. size\\\hline\hline 
        Cortical Dev. (Fgf8=0) & BN & 5 & 554 (1,939) & 8,312\\\hline 
        Cortical Dev. (Fgf8=1) & BN & 5 & 1,054 (3,530) & 8,312\\\hline
        EGF-TNF$\alpha$ & BN & 13 & 1,057 (2,658) & 534,498\\\hline
        $\lambda$-switch & MN & 4 & 170 (575) & 68,011\\\hline
        $\lambda$-switch w/ Min-Max & MN & 4 & 157 (527) & 15,139\\\hline
        Prim. Sex Det. w/ Min-Max & MN & 14 & 19,954 (88,994) & >7,000,000\\\hline
        Drosophila Signalling & MN & 15 & 781 (2,698) & >7,000,000\\\hline
        Drosophila w/ Min-Max & MN & 15 & 731 (2,507) & >7,000,000\\\hline
    \end{tabular}
	\end{scriptsize}
    \caption{Comparison of the size of the obtained structures between
      unfolding and the symbolic representation for different
      models. The number of unfolding events is specified as a total
      number of non-cut-off events. The number including cut-off
      events is given in brackets. Sym. exec. size is the number of
      states of the complete execution tree constructed by
      \sputnik. The notation '>7,000,000' means the size was at least $7,000,000$ when the
      particular experiment has been stopped after 2 days of computation.}
    \label{tab:results}
\end{table}

Concurrency-aware semantics shows a great improvement in the compactness
of the resulting structure. It is striking in the case of models of signalling pathway cross-talks (Drosophila and EGF-TNF) where concurrency is high due to the low connectivity of the influence graph. The size of unfolding prefixes remains very compact even in cases with more interwoven topology. It is worth noting that the constructed unfoldings preserve the set of reachable states, and any process can be reconstructed from them, with an additional computation cost~\cite{Esparza2001}.

Cortical Development model provides another interesting observation -- the unfolding can be sensitive to the initial state. In this model, the considered initial states give the same reachable state space. However, depending on the initial state, the respective unfoldings have substantially different size.

Theorem \ref{thm:pabsTR} ensures that the set of reachable states in the prefix is exact
despite the over-approximation (for each reachable state there exists at least one
true positive within the computed parametrisation set).
Future work may be aimed at evaluation of the rate of false positives among
the parametrisations for the purposes of parameter identification.

\section{Discussion}
\label{sec:discussion}

This article introduces an abstraction of the parametrisation space of Parametric Regulatory
Networks (PRNs) by the means of two parametrisations, defining the greatest lower bound and least
upper bound of convex sublattice.
We defined narrowing operators to refine the abstracted parametrisation space according to the
existence of transitions, as well as influence monotonicity and observability constraints.
We demonstrate that our operators lead to the smallest approximation of the concrete
parametrisation space attainable by the means of convex lattice.

Our results guarantee that the abstract interpretation of PRN semantics introduce no
over-approximation over the realisable sequences of transitions:
any sequence of transitions allowed by the abstract semantics of PRNs is a realisable sequence of
transitions for the concrete semantics of PRNs, i.e., there exists at least one Discrete Regulatory
Network in which the trace exists.

We also introduce an unfolding semantics for PRNs which takes advantage of the concurrency between
transitions to provide a compact representation of possible behaviours.
The unfolding semantics is built equivalently on the concrete and abstract semantics of PRNs.

Thanks to the compact abstraction of parametrisation spaces and to the unfolding semantics,
preliminary experiments show that our method can explore the full dynamics of PRNs for multi-valued networks
with a dozen of components.

Our approach naturally extends to Gene Regulatory Networks with multiplexes \cite{BernotMultiplexes}
as parameters of multiplex nodes are fully determined (their lower bound and upper
bound parameter values are equal).
Yet, future work may consider extensions of our approach to account for partial parameters
specifications and arbitrary constraints on parametrisations.
Another research direction is the application of our semantics for the parameter identification
problem from temporal properties expressed as LTL or CTL specifications.


{\small
\paragraph{Acknowledgements}
JK, SH, and LP have been partially supported by
ANR-FNR project ``AlgoReCell'' ANR-16-CE12-0034.
and by Labex DigiCosme (project 
ANR-11-LABEX-0045-DIGICOSME) operated by ANR as part of the program 
``Investissement d'Avenir'' Idex Paris-Saclay (ANR-11-IDEX-0003-02).
DS has been supported by the Czech Science Foundation grant No.
GA15-11089S. Main ideas of this paper have been set during JK's and DS's research visits partially supported
by Inria and by Universit\'e Paris-Sud, France.
}

\bibliographystyle{elsarticle-num}
\bibliography{main}

\newpage

\appendix

\section{Proof of Theorem~\ref{thm:pabsTR}}
\label{app:proof-best-approx}

In this appendix, we give the full proof of theorem~\ref{thm:pabsTR} stating that $\pabs_R(T)$ is equal to the smallest convex sublattice containing $p_R(T)$.
It is important to note that in this appendix we will only consider the influence of one node in all the proofs.
We can afford to do this as all the restrictions imposed on $\PG$ to obtain $p_R(T)$ have only local effects, in the sense of the node being regulated.
This is clearly the case for including a transition $t$ in the set $T$, as a transition only affects one associated regulatory state,
however, it also naturally extends to monotonicity and observability constraints.
Marking an influence of node $v$ as monotonic or observable only imposes restriction on the values of parameters that govern the regulation of $v$.
Thanks to this locality, we can analyse the regulation of each node in $V$ separately,
and obtain the final parametrisation set as cartesian sum of the parametrisation sets considered for each node separately.

From here on, all parametrisations are thus considered local to the regulation of the given node,
as well as constraint set $R$ is considered only as the subset of constraints on influences of the given node.

\paragraph{Additional notations}
$[p_R(T)]$ is the smallest convex sublattice including $p_R(T)$.

\medskip

We first introduce several lemmas to aid us in the main proof.
Lemma~\ref{lem:mono-order} describes the intuitive connection between monotonicity constraints and monotonicity order ($\preceq_v$)
while lemma~\ref{lem:mono-order-bounds} extends the intuition to bounds of the abstract parametrisation set.

\begin{lemma}\label{lem:mono-order}
Every parametrisation $P$ that satisfies all monotonic constraints on influences of node $v$,
must also have $P_{v,\omega}\leq P_{v,\omega'}$ for any couple $\omega,\omega'\in\Omega_v$ such that $\omega\preceq_v\omega'$ and vice versa.

Formally, for an arbitrary parametrisation $P\in\PG$ and node $v$:
$$\forall(u,v,s)\in R:(s\in\{-1,1\}\Rightarrow
P\in\P_{(u,v,s)})\Leftrightarrow$$
$$\forall\omega,\omega'\in\Omega_v:
(\omega\preceq_v\omega'\Rightarrow
P_{v,\omega}\leq P_{v,\omega'})$$
\end{lemma}

\begin{proof}
We conduct the proof directly.
$$\forall(u,v,s)\in R:s\in\{-1,1\}\Rightarrow P\in\P_{(u,v,s)}$$
$$\Longleftrightarrow$$
$$\forall(u,v,s)\in R:s\in\{-1,1\}\Rightarrow\forall\omega\in\Omega_v,\forall x_u\in\domv u:$$
$$P_{v,\subst\omega u {x_u}}\geq P_{v,\subst\omega u {x_u-s}}$$
with the obvious exception of $x_u=0$ in case $s=1$, and $x_u=m_u$ when $s=-1$.
$$\Longleftrightarrow$$
$$\forall\omega,\omega'\in\Omega_v:\omega\preceq_v\omega'\wedge$$
$$(\forall\bar\omega\in\Omega_v:\omega\preceq_v\bar\omega\preceq_v\omega'\Rightarrow
\omega=\bar\omega\vee\omega'=\bar\omega)\Rightarrow
P_{v,\omega}\leq P_{v,\omega'}$$
$$\Longleftrightarrow \text{by transitivity of} \preceq_v,\leq$$
$$\forall\omega,\omega'\in\Omega_v:\omega\preceq_v\omega'\Rightarrow P_{v,\omega}\leq P_{v,\omega'}$$
\end{proof}

\begin{lemma}\label{lem:mono-order-bounds}
The lower and upper boundary parametrisations, $\Plb$ and $\Pub$ respectively, of the smallest convex sublattice covering the concrete parametrisation set $p_R(T)$ for some transition set $T$ have, for every node $v$ and any couple $\omega,\omega'\in\Omega_v$ such that $\omega\preceq_v\omega'$, $\Plb_{v,\omega}\leq\Plb_{v,\omega'}$ and $\Pub_{v,\omega}\leq\Pub_{v,\omega}$.

Formally, for a set of transitions $T$ such that $p_R(T)\neq\emptyset$ and $(L,U)=[p_R(T)]$, and an arbitrary node $v$:
$$\forall\omega,\omega'\in\Omega_v:\omega\preceq_v\omega'\Rightarrow
\Plb_{v,\omega}\leq\Plb_{v,\omega'}\wedge
\Pub_{v,\omega}\leq\Pub_{v,\omega'}$$
\end{lemma}

\begin{proof}
Let $\omega,\omega'\in\Omega_v:\omega\prec_v\omega'$ be arbitrary.

We know $\forall P\in p_R(T)$ that $P_{v,\omega}\geq\Plb_{v,\omega}$.
Furthermore by lemma~\ref{lem:mono-order} we have $P_{v,\omega'}\geq P_{v,\omega}$ and thus $P_{v,\omega'}\geq\Plb_{v,\omega}$.
Since this holds for all $P\in p_R(T)$, it must also hold that $\Plb_{v,\omega'}\geq\Plb_{v,\omega}$.

The proof for the upper boundary parametrisation $U$ is symmetrical.
\end{proof}

Lemma~\ref{lem:observ_parity}, on the other hand, captures a very important property of the obeservability constraint.
The property slightly resembles parity, as given some influence of node $v$ not observable under a parametrisation $P$,
changing the value of an arbitrary single parameter of node $v$ regulation in $P$ results in a parametrisation $P'$ such that all influences of $v$ are observable under $P'$.
The connection with parity is especially apparent in case of Boolean networks, where parameters can only have values $0$ or $1$.
In the Boolean case, every influence of $v$ is observable under any parametrisation that has odd number of node $v$ regulation parameters valued $1$.
The reverse does not apply, however, as odd parity of parameters valued $1$ is sufficient, but not necessary condition for observability of all influences in Boolean PRNs.

\begin{lemma}\label{lem:observ_parity}
Given a parametrisation $P\in\PG$ and a constraint $r=(u,v,\o)\in R$ such that $P\notin\P_r$,
then for every other parametrisation $P'$ that differs from $P$ in value of exactly one $\omega\in\Omega_v$,
and for all observability constraints $r'=(u',v,o)\in R$ on $v$, $P'\in\P_{r'}$.

Formally, for a parametrisation $P\in\PG$ and a constraint $r=(u,v,o)\in R$ such that $P\notin\P_r$:
$$\forall\omega\in\Omega_v,\forall x_v\in\domv v:x_v\neq P_{v,\omega}\Rightarrow
\forall(u',v,o)\in R:\subst P {v,\omega} {x_v}\in\P_{(u',v,o)}$$
\end{lemma}

\begin{proof}
Let $P\in\PG$ and $r=(u,v,o)\in R$ be such that $P\notin\P_r$ and
let $\omega\in\Omega_v,x_v\in\domv v$ be arbitrary such that $P_{v,\omega}\neq x_v$.

We denote the modified parametrisation as $P'=\subst P {v,\omega} {x_v}$
and a regulator state identical to $\omega$ up to value of $u$ as $\hat\omega=\subst \omega u {\omega_u+k}$ where $k=\begin{cases}
1 &\omega_u=0\\
-1 &\text{otherwise}
\end{cases}$.

For $r$ we get $P'_{v,\hat\omega}=P_{v,\hat\omega}=P_{v,\omega}\neq P'_{v,\omega}$ and thus $P'\in\P_r$.
Now let us assume $v$ has at least two (observable) influences and let $r'=(u',v,o)\in R$ be arbitrary.

First, we introduce two additional regulator states.
The regulator state identical to $\omega$ up to the value $u'$, denoted $\omega'=\subst \omega {u'} {\omega_{u'}+k'}$
and the regulator state identical to $\omega$ up to the values of both $u$ and $u'$, denoted $\hat\omega'=\subst {\hat\omega} {u'} {\omega_{u'}+k'}$.
Where $k'=\begin{cases}
1 &\omega_{u'}=0\\
-1 &\text{otherwise}
\end{cases}$.

The regulator states $\omega,\hat\omega,\omega'$ and $\hat\omega'$ are now used to show that $(u',v)$ is indeed observable under $P'$.
This is achieved by showing that either $P'_{v,\omega}\neq P'_{v,\omega'}$ or $P'_{v,\hat\omega}\neq P'_{v,\hat\omega'}$ as
both $\omega,\omega'$ and $\hat\omega,\hat\omega'$ differ only in the value of $u'$.
We also use the analogous fact that $\omega,\hat\omega$ and $\omega',\hat\omega'$ differ only in the value of $u$.

The result is trivial if $x_v\neq P_{v,\omega'}$ as $x_v=P'_{v,\omega}\neq P_{v,\omega'}=P'_{v,\omega'}$.
Thus, $P'\in\P_{r'}$.

Let us therefore assume $P_{v,\omega'}=x_v=P'_{v,\omega}$.
Since $P\notin\P_r$ we know that $P_{v,\omega}=P_{v,\hat\omega}$ and $P_{v,\omega'}=P_{v,\hat\omega'}$.
As $P'$ only differs from $P$ on $\omega$, we can expand the previous to obtain $P'_{v,\omega}\neq P_{v,\omega}=P_{v,\hat\omega}=P'_{v,\hat\omega}$
and $P'_{v,\omega'}=P_{v,\omega'}=P_{v,\hat\omega'}=P'_{v,\hat\omega'}$.
Here we use our assumption $P_{v,\omega'}=P'_{v,\omega}$ to obtain $P'_{v,\hat\omega'}=P'_{v,\omega'}=P_{v,\omega'}=P'_{v,\omega}\neq P'_{v,\hat\omega}$
giving us the coveted $P'_{v,\hat\omega'}\neq P'_{v,\hat\omega}$. Thus, $P'\in\P_{r'}$.
\end{proof}

A crucial property of the monotonicity and observability constraints is captured in theorem~\ref{lem:density}.
Given a smallest convex sublattice $(\Plb,\Pub)=[p_R(T)]$ covering a parameter set for some $T$,
there are strict limits on the conditions under which there may exist values $k,k'\in \domv v$ for a couple of regulator states $\omega,\omega'\in\Omega_v$, respectively,
such that $\Plb_{v,\omega}\leq k\leq\Pub_{v,\omega}$ and $\Plb_{v,\omega'}\leq k'\leq\Pub_{v,\omega'}$, but there is no parametrisation $P\in p_R(T)$ with $P_{v,\omega}=k$ and $P_{v,\omega'}=k'$.
The case when no parametrisation with values $k,k'$ for $\omega$ and $\omega'$, respectively, exists in the concrete set solely due to monotonicity constraints,
i.e. $\omega\preceq_v\omega'$ and $k>k'$, is trivial and is not considered in the theorem.

In other words, theorem~\ref{lem:density} ensures a sort of density of the concrete parametrisation set $p_R(T)$,
meaning that by constructing the smallest convex sublattice,
the parametrisations included to satisfy convexity are limited and the observability constraint disqualifying them from the concrete set can be identified.
This result is used in the proof of theorem~\ref{thm:pabsTR} to show that if the smallest convex sublattice covering $p_R(T)$ grows smaller with the inclusion of a new transition $t$,
the abstract counterpart $\pabs_R(T)$ will reflect the change via restrictions $\restrict_t$ and namely $\restrict_R$.

\begin{theorem}
\label{lem:density}
Given a PRN $G_m$, a set of transitions $T\subseteq\Delta(G_m)$
and a well-formed set of constraints $R$ such that $(\Plb,\Pub)=[p_R(T)]\neq\emptylattice$.
Then $\forall v\in V$ and arbitrary couple $\omega,\omega'\in\Omega_v$:
\begin{align*}
&\forall y_v\in\{\Plb_{v,\omega},\dots,\Pub_{v,\omega}\},\\
&\forall z_v\in\{l_{v,\omega'}(\{(\omega,y_v)\}),\dots,u_{v,\omega'}(\{(\omega,y_v)\})\}:\\
(\exists\omega''\in\Omega_v\setminus\{\omega,\omega'\}:&
l_{v,\omega''}(\{(\omega,y_v),(\omega',z_v)\})<
u_{v,\omega''}(\{(\omega,y_v),(\omega',z_v)\})\Rightarrow\\
\exists P\in p_R(T):&P_{v,\omega}=y_v\wedge P_{v,\omega'}=z_v
\end{align*}
where
\begin{align*}
l_{v,\omega}&:2^{\Omega_v\times\domv v}\rightarrow\domv v\\
l_{v,\omega}(\mathcal{O})&\DEF\max(\{k\mid(\omega',k)\in\mathcal{O}:\omega'\prec_v\omega\}\cup
\{\Plb_{v,\omega}\})\\
u_{v,\omega}&:2^{\Omega_v\times\domv v}\rightarrow\domv v\\
u_{v,\omega}(\mathcal{O})&\DEF\min(\{k\mid(\omega',k)\in\mathcal{O}:\omega'\succ_v\omega\}\cup
\{\Pub_{v,\omega}\})
\end{align*}
\end{theorem}

\begin{proof}
We propose a parametrisation $P$ constructed as follows:
\begin{enumerate}
	\item $P_{v,\omega}=y_v\wedge P_{v,\omega'}=z_v$
	\item $P_{v,\omega''}=u_{v,\omega''}(\{(\omega,y_v),(\omega',z_v)\})$
	\item $\forall\bar\omega\in\Omega_v\setminus\{\omega,\omega',\omega''\}:P_{v,\bar\omega}=
		\max(\{P_{v,\bar\omega'}\mid\bar\omega'\in\Omega_v:\bar\omega'\prec_v\bar\omega\}
	\cup\{\Plb_{v,\bar\omega}\})$ computed iteratively in the increasing direction of $\preceq_v$.
\end{enumerate}
And a second parametrisation
$P'=\update P {v,\omega''} {} {l_{v,\omega''}(\{(\omega,y_v),(\omega',z_v)\})}$.

We will now show that $P,P'$ respect the lower and upper bounds $\Plb,\Pub$:
\begin{enumerate}
	\item $\Plb_{v,\omega}\leq y_v=P_{v,\omega}=P'_{v,\omega}\leq\Pub_{v,\omega}$
	\item $\Plb_{v,\omega'}\leq l_{v,\omega'}(\{(\omega,y_v)\})\leq
			z_v=P_{v,\omega'}=P'_{v,\omega'}\leq
			u_{v,\omega'}(\{(\omega,y_v)\})\leq\Pub_{v,\omega'}$
	\item $\Plb_{v,\omega''}\leq l_{v,\omega''}(\{(\omega,y_v),(\omega',z_v)\})=
			P'_{v,\omega''}<P_{v,\omega''}=
u_{v,\omega''}(\{(\omega,y_v),(\omega',z_v)\})\leq\Pub_{v,\omega''}$
	\item $\forall\bar\omega\in\Omega_v\setminus\{\omega,\omega',\omega''\}:
			\Plb_{v,\bar\omega}\leq P_{v,\bar\omega}=P'_{v,\bar\omega}$ by definition.
			
			Using lemma \ref{lem:mono-order-bounds} we have
			$\Pub_{v,\bar\omega}\geq
			max(\{\Pub_{v,\bar\omega'}\mid\bar\omega'\in\Omega_v:
			\bar\omega'\prec_v\bar\omega\})\Rightarrow
			\Pub_{v,\bar\omega}\geq P_{v,\bar\omega}=P'_{v,\bar\omega}$
\end{enumerate}

Furthermore, both $P$ and $P'$ satisfy all the monotonicity constraints.

For $\bar\omega\in\Omega_v\setminus\{\omega,\omega',\omega''\}$ and any other $\bar\omega'\in\Omega_v$ we have
$\hat\omega\preceq_v\hat\omega'\Rightarrow P_{v,\hat\omega}\leq P_{v,\hat\omega'}$
and $\hat\omega\succeq_v\hat\omega'\Rightarrow P_{v,\hat\omega}\geq P_{v,\hat\omega'}$ by definition.

The same holds for $P'$ as long as we prove $P'_{v,\hat\omega}\leq P'_{v,\omega''}$ for $\hat\omega\preceq_v\omega''$ since $P'_{v,\omega''}<P_{v,\omega''}$.
Let us first consider $\neg(\omega\preceq_v\hat\omega)$ and $\neg(\omega'\preceq_v\hat\omega)$.
Then, by definition, $P'_{v,\hat\omega}=\max_{\hat\omega''\preceq_v\hat\omega}(\Plb_{v,\hat\omega''})$ and thanks to lemma~\ref{lem:mono-order-bounds} $P'_{v,\hat\omega}=\Plb_{v,\hat\omega}$.
We know $P'_{v,\omega''}=l_{v,\omega''}(\{(\omega,y_v),(\omega',z_v)\})\geq\Plb_{v,\omega''}$ and thus,
by application of lemma~\ref{lem:mono-order-bounds} again, we obtain $\Plb_{v,\hat\omega}\leq\Plb_{v,\omega''}\leq P'_{v,\omega''}$.

Let us now consider either $\omega\preceq_v\hat\omega$ and/or $\omega'\preceq_v\hat\omega$. This translates to $P_{v,\hat\omega}=\max\{\Plb_{v,\hat\omega},y_v,z_v\}$.
We have $P'_{v,\omega''}\geq\Plb_{v,\hat\omega}$
and $y_v\leq l_{v,\omega''}(\{(\omega,y_v),(\omega',z_v)\})=P'_{v,\omega''}$,
$z_v\leq l_{v,\omega''}(\{(\omega,y_v),(\omega',z_v)\})=P'_{v,\omega''}$ follow directly from the definition of $l_{v,\omega}$.
Note that this also proves monotonicity satisfaction in case $\omega\preceq_v\omega''$ and $\omega'\preceq_v\omega''$
and since $P'_{v,\omega''}<P_{v,\omega''}$ the monotonicity is satisfied in those cases for $P$ as well.

The same reasoning can be applied to show $\omega\preceq_v\omega'\Rightarrow P_{v,\omega}\leq P_{v,\omega'}$
since $P_{v,\omega}=y_v\leq l_{v,\omega'}(\{(\omega,y_v)\})\leq z_v= P_{v,\omega'}$ by definition.
Thanks to $P_{v,\omega}=P'_{v,\omega}$ and $P_{v,\omega'}=P'_{v,\omega'}$, the result is extended to $P'$.

All that is left to prove is that monotonicity is satisfied in case of the opposite direction of $\preceq_v$: $\omega'\preceq_v\omega$, $\omega''\preceq_v\omega$ and $\omega''\preceq_v\omega'$.
It is easy to see from definition of $u_{v,\omega}$, that monotonicity is satisfied in all three cases as $P_{v,\omega'}=z_v\leq y_v=P_{v,\omega}$,
$P_{v,\omega''}=u_{v,\omega''}(\{(\omega,y_v),(\omega',z_v)\})\leq y_v=P_{v,\omega}$
and $P_{v,\omega''}=u_{v,\omega''}(\{(\omega,y_v),(\omega',z_v)\})\leq z_v=P_{v,\omega'}$
holds for the three $\preceq_v$ conditions, respectively.
The same applies for $P'$ due to equality with $P$ on $\omega$ and $\omega'$, and the fact that $P'_{v,\omega''}<P_{v,\omega''}$.


Finally, we show that at least one of the parametrisations $P,P'$ also satisfies all the observability constraints.
Since $P$ and $P'$ differ in the value of exactly one regulator state $\omega''$, the result is obvious from lemma~\ref{lem:observ_parity}.
\end{proof}

Lastly, lemma~\ref{lem:mono-restrict} expands further on the property shown in Lemma~\ref{lem:mono-order-bounds}
to show that the fixed point for monotonic restrictions is reached exactly when both lower and upper boundary parametrisations $\Plb$ and $\Pub$,
respectively comply with the monotonicity order.

\begin{lemma}\label{lem:mono-restrict}
The fixed point of restriction by monotonic constraints is reached for a sublattice of parametrisations $(\Plb,\Pub)$
exactly when the following holds for any couple of node $v$ regulator states $\omega,\omega'\in\Omega_v$ such that $\omega\preceq_v\omega'$:
$\Plb_{v,\omega}\leq\Plb_{v,\omega'}$ and $\Pub_{v,\omega}\leq\Pub_{v,\omega'}$.

Formally, for an arbitrary node $v$ and an arbitrary sublattice of parametrisations $(\Plb,\Pub)\subseteq\PG$:
$$\restrict_{\{(u,v,s)\in R\mid s\in\{-1,1\}\}}(\Plb,\Pub)=(\Plb,\Pub)\Leftrightarrow$$
$$\forall\omega,\omega'\in\Omega_v:\omega\preceq_v\omega'\Rightarrow
\Plb_{v,\omega}\leq\Plb_{v,\omega'}\wedge
\Pub_{v,\omega}\leq\Pub_{v,\omega'}$$

\end{lemma}

\begin{proof}
We conduct the proof directly.
$$\restrict_{\{(u,v,s)\in R\mid s\in\{-1,1\}\}}(\Plb,\Pub)=(\Plb,\Pub)$$
$$\Longleftrightarrow$$
$$\forall(u,v,s)\in R:s\in\{-1,1\}\Rightarrow\forall\omega\in\Omega_v,\forall x_u\in\domv u:$$
$$\Plb_{v,\subst\omega u {x_u}}\geq \Plb_{v,\subst\omega u {x_u-s}}\wedge\Pub_{v,\subst\omega u {x_u}}\geq \Pub_{v,\subst\omega u {x_u-s}}$$
with the obvious exception of $x_u=0$ in case $s=1$, and $x_u=m_u$ when $s=-1$.
$$\Longleftrightarrow \text{by transitivity of} \preceq_v,\leq$$
$$\forall\omega,\omega'\in\Omega_v:\omega\preceq_v\omega'\Rightarrow
\Plb_{v,\omega}\leq\Plb_{v,\omega'}\wedge
\Pub_{v,\omega}\leq\Pub_{v,\omega'}$$
\end{proof}

Finally, we conduct the proof of theorem \ref{thm:pabsTR} itself:
$$\pabs_R(T) = [p_R(T)]$$
\begin{proof}
Let $G_m$ be a PRN with well-formed set of constraints $R$ and let node $v$ be arbitrary.
The proof uses mathematical induction on the size of $T$.

Base case $T=\emptyset$:

Let us first remark that $p(\emptyset)=[p(\emptyset)]=\PG$ contains all possible parametrisations.
Likewise, no restriction by transition, $\restrict_t$, is called on the abstract counterpart $(\lb{\PG},\ub{\PG})=\pabs(\emptyset)=\PG$.
Therefore, only edge constraints in $R$ may be responsible for $p_R(\emptyset)\neq\PG$ or $\pabs_R(\emptyset)\neq\PG$.

For any couple of regulator states $\omega,\omega'\in\Omega_v$,
it holds that $\lb{\PG}_{v,\omega}=\lb{\PG}_{v,\omega'}=0$ and $\ub{\PG}_{v,\omega}=\ub{\PG}_{v,\omega'}=\mv v$.
As such, it is easy to see that any monotonicity constraint is satisfied by both $\lb{\PG}$ and $\ub{\PG}$.
Therefore if all constraints in $R$ are monotonic, both $\lb{\PG}\in p_R(\emptyset)$ and $\ub{\PG}\in p_R(\emptyset)$
and $[p_R(\emptyset)]=\PG=\pabs_R(\emptyset)$, the latter equality being derived from lemma~\ref{lem:mono-restrict}.

Let us now consider there exists at least one observability constraint $r=(u,v,\o)\in R$ on influences of $v$.
Surely the influence $(u,v)$ is not observable under $\lb{\PG}$ and $\ub{\PG}$ so the previous argument does not apply.

During the computation of $\restrict_r$,
the set of all regulator states that allow the value of $v$ to change with the change of value of $u$,
$A_{u,v}(\Plb,\Pub)=\PG$ becomes the whole possible parametrisation set.
(Exluding the pathological case when $\mv v=0$ and thus node $v$ cannot change value.
If $\mv v=0$, $|\PG|=1$ and $A_{u,v}(\Plb,\Pub)=\emptyset$ for any $u\in n^-(v)$.
It is then easy to see that $p_R(\emptyset)=\emptyset\Leftrightarrow\pabs_R(\emptyset)=\emptylattice\Leftrightarrow$ there exists an observable influence of $v$.)

As mentioned above, The parameter values of both $\lb{\PG}$ and $\ub{\PG}$ are equal for all regulator states of $v$.
The sets $\overline{B}$ and $\underline{B}$ thus contain all the $\preceq_v$-maximal, respectively $\preceq_v$-minimal, elements of $\Omega_v$.
If there exists an influence $(u',v)$ that is not monotonic, there will exist at least two distinct $\preceq_v$-maximal and $\preceq_v$-minimal elements
since for any $\omega\in\Omega_v:\subst\omega {u'} 0\parallel_v\subst\omega {u'} {\mv {u'}}$.
As such, no restriction occurs as $\restrict_r(\PG)=\PG=\pabs_R(\emptyset)$.

The same holds for the concrete parametrisation set,
as having two distinct $\preceq_v$-maximal elements $\overline\omega,\overline\omega'$ means
there exist two parametrisations $\overline P=\update{\lb{\PG}} {v,\overline\omega} {} {\mv v}$ and $\overline P'=\update{\lb{\PG}} {v,\overline\omega'} {} {\mv v}$
and both $P\in p_R(\emptyset)$ and $P'\in p_R(\emptyset)$, where constraint satisfaction follows from lemma~\ref{lem:mono-order} for monotonicity, and lemma~\ref{lem:observ_parity} for observability.
Thus $\lb{p_R(\emptyset)}=\lb{\PG}$ as for every $\omega\in\Omega_v$ there exists a parametrisation in $p_R(\emptyset)$ with parameter value $0$ for $\omega$.
A symmetrical construction can be done to show that $\ub{p_R(\emptyset)}=\ub{\PG}$ and thus $[p_R(\emptyset)]=\pabs_R(\emptyset)$.

Finally, let all the influences of $v$ be monotonic in addition to at least one, $(u,v)$, being observable.
All influences being monotonic means that for any couple $\omega,\omega'\in\Omega_v$ we have either $\omega\preceq_v\omega'$ or $\omega\succeq_v\omega'$
giving us a unique $\preceq_v$-minimal and $\preceq_v$-maximal elements, or alternatively, $|\overline{B}|=1=|\underline{B}|$.
As such $\restrict_r(\PG)=(\Plb,\Pub)$ will restrict the unique $\preceq_v$-maximal $\overline\omega$ to have value at least $1$, $\Plb_{v,\overline\omega}=1$,
and analogously, the unique $\preceq_v$-minimal $\underline\omega$ to be at most $\mv v-1$, $\Pub_{v,\underline\omega}=\mv v-1$.

It is also important to note that both $\Plb\in p_R(\emptyset)$ and $\Pub\in p_R(\emptyset)$,
constraint satisfaction again by lemmas~\ref{lem:mono-order} and~\ref{lem:observ_parity},
as this gives us $\pabs_R(\emptyset)\subseteq[p_R(\emptyset)]$.

All that remains to be shown is that $[p_R(\emptyset)]\subseteq\pabs_R(\emptyset)$.
Let $P\in\PG$ be arbitrary parametrisation such that $P_{v,\overline\omega}=0$.
Since $\overline\omega$ is the unique $\preceq_v$-maximal element $P$ has to have all the other parameter values also equal to $0$ in order to meet the monotonic constraints,
thus either $P=\lb{\PG}\notin p_R(\emptyset)$ or $\exists (u',v,s)\in R:P\notin\P_{(u',v,s)}\supseteq p_R(\emptyset)$ where $s\in\{-1,1\}$.
Again, symmetrical conditions apply to arbitrary parametrisation $P'$ with $P'_{v,\underline\omega}=\mv v$ to show $P'\notin p_R(\emptyset)$.

We have thus proven the base case $[p_R(\emptyset)]=\pabs_R(\emptyset)$.

Induction hypothesis: $\pabs_R(T)=[p_R(T)]$ for any $T$ such that $|T|\leq k$ for some $k\in\mathbb{N}$.

We show that $\pabs_R(T\cup\{t\})=[p_R(T\cup\{t\})]$ for arbitrary transition $t\notin T$.
The proof here is split into two separate branches.
We first prove soundness of the abstraction, $[p_R(T\cup\{t\})]\subseteq\pabs_R(T\cup\{t\})$,
and subsequently we prove that we achieve the best over-approximation, $\pabs_R(T\cup\{t\})\subseteq[p_R(T\cup\{t\})]$.

$[p_R(T\cup\{t\})]\subseteq\pabs_R(T\cup\{t\})$ (soundness):

Let us first remark that if $p_R(T\cup\{t\})=\emptyset$
the smallest convex sublattice is also empty $[p_R(T\cup\{t\})]=\emptylattice\subseteq\pabs_R(T\cup\{t\})$ regardless of the value of the abstract parametrisation set $\pabs_R(T\cup\{t\})$.
We therefore assume $p_R(T\cup\{t\})\neq\emptyset$.

To prove soundness we have to show that if a restriction $\restrict_t$ or $\restrict_R$ results in a strictly smaller lattice,
the change is also reflected in the concrete domain $p_R(T\cup\{t\})$ and the smallest convex sublattice covering it $[p_R(T\cup\{t\})]$.

We first show that $[p_R(T\cup\{t\})]$ is contained by the lattice $(\Plb,\Pub)=\restrict_t(\pabs_R(T))$,
i.e. any restriction imposed by $\restrict_t$ is reflected in the concrete parametrisation set.

Let $t=x\xrightarrow{v,s}y$ and let us assume $s=1$.
By definition of $\restrict_t$,
the only change that may occur is the increase of parameter value for $\omega_v(x)$ in the lower boundary parametrisation.
As such, no change occurs if $\lb{\pabs_R(T)}\geq y_v$ and $[p_R(T\cup\{t\})]\subseteq[p_R(T)]=\pabs_R(T)=(\Plb,\Pub)$.
Let us therefore assume $\lb{p_R(T)}<y_v$.
By definition $p_R(T\cup\{t\})=p_R(T)\cap\P_t$.
Furthermore, for all parametrisations $P\in\P_t$ it holds that $P_{v,\omega_v(x)}\geq y_v$,
thus $p_R(T\cup\{t\})$ contains exactly those parametrisations $P\in p_R(T)$ that have $P_{v,\omega_v(x)}\geq y_v$,
leading to $\lb{p_R(T\cup\{t\})}_{v,\omega_v(x)}\geq y_v=\Plb_{v,\omega_v(x)}$.
Coupled with $[p_R(T\cup\{t\})]\subseteq[p_R(T)]$ we obtain the coveted $\lb{p_R(T\cup\{t\})}\geq\Plb$.

The proof of $\ub{p_R(T\cup\{t\})}\leq\Pub$ in case the transition is decreasing, $s=-1$, is symmetrical.

To prove that restrictions enforced by $\restrict_R$ are also reflected in the concrete parameter set
we use induction again. The induction is conducted on the number of calls of $\restrict_r$ for individual constraints $r\in R$.

Formally, applying $\restrict_R$ to $\restrict_t(\pabs_R(T))$ in order to compute $\pabs_R(T\cup\{t\})$
translates into application of finitely many restrictions
$\restrict_{r_1}(\dots\restrict_{r_k}(\restrict_t(\pabs_R(T)))\dots)$ where $k\in\mathbb{N}$
and $\forall i\in\{1,\dots,k\}: r_i\in R$.
Note that the order of restriction application is not fixed.
The same fixpoint is reached regardless of the order, however,
and we consider an arbitrary sequence of restrictions that is valid in the sense of reaching the fixpoint.

To simplify notation, we use $\rho=(r_1,\dots,r_k)$ to denote the chain of constraints used for restriction by $\restrict_R$.
We use $\rho_i$ for $i\in\{0,\dots,k\}$ to denote prefix of $\rho$ of the length of $i$, $\rho_0$ being empty and $\rho_k=\rho$.
Additionally we write $\restrict_{\rho_i}$ to denote the application of the restrictions according to $\rho_i$.
The induction is thus conducted on the length of prefixes of $\rho$.

Base case $i=0$, is trivial as if no additional restriction happens we have $[p_R(T\cup\{t\})]\subseteq\restrict_t(\pabs_R(T))$.

Induction hypothesis: $[p_R(T\cup\{t\})]\subseteq(\Plb,\Pub)=\restrict_{\rho_i}(\restrict_t(\pabs_R(T)))$
 where $i\leq l$ for some $l\in\{0,\dots,k-1\}$.

We now prove that $[p_R(T\cup\{t\})]\subseteq(\Plb',\Pub')=\restrict_{\rho_{l+1}}(\restrict_t(\pabs_R(T)))$.
The result is trivial if $(\Plb',\Pub')=(\Plb,\Pub)$,
we thus assume inequality.

Let $r=(u,v,s)\in R$ be the last constraint in $\rho_{l+1}$.
We now conduct a discussion on the nature of $r$:
\begin{enumerate}[label=(\roman*)]
\item $r$ is monotonicity constraint, $s\in\{-1,1\}$.
By lemma~\ref{lem:mono-restrict} $\restrict_r((\Plb,\Pub))\neq(\Plb,\Pub)$ guarantees an existence of a couple $\omega,\omega'\in\Omega_v$
such that $\omega\preceq_v\omega'$, but $\Plb_{v,\omega}>\Plb_{v,\omega'}$ or $\Pub_{v,\omega}>\Pub_{v,\omega'}$.
Furthermore, from the definition of $\restrict_r$ we have $\Plb'_{v,\omega'}=\Plb_{v,\omega}$,
respectively $\Pub'_{v,\omega}=\Pub_{v,\omega'}$.

For any parametrisation $P\in p_R(T)$ such that $P_{v,\omega'}<\Plb_{v,\omega}$, respectively $P_{v,\omega}>\Pub_{v,\omega'}$,
we know $P\notin p_R(T\cup\{t\})$ by lemma~\ref{lem:mono-order}.
Thus, the bounds of $[p_R(T\cup\{t\})]$ must be $\lb{p_R(T\cup\{t\})}_{v,\omega'}\geq\Plb_{v,\omega}$,
respectively $\ub{p_R(T\cup\{t\})}_{v,\omega}\leq\Pub_{v,\omega'}$, giving us the coveted $[p_R(T\cup\{t\})]\subseteq(\Plb',\Pub')$.

\item $r$ is observability constraint, $s=\o$.
As we assume $(\Plb',\Pub')\neq(\Plb,\Pub)$, restriction $\restrict_r$ must have changed either the lower or upper bound.
By definition, $A_{u,v}(\Plb,\Pub)\neq\emptyset$ and at least one of the following: $|\underline{B}|=1$, $|\overline{B}|=1$.

Let us assume $\overline{B}=\{\omega\}$ for some $\omega\in\Omega_v$,
giving us $\Plb'_{v,\omega}=\Plb_{v,\omega}+1$.
From definition of $A_{u,v}(\Plb,\Pub)$, for any $\omega'\in\Omega_v\setminus A_{u,v}(\Plb,\Pub)$ and for all possible values $x_u$ of $u$,
all the lower bound and upper bound values $L_{v,\subst{\omega'} u {x_u}}$ and $\Pub_{v,\subst{\omega'} u {x_u}}$ are equal.
As such, there is no $\omega'$ outside of $A_{u,v}(\Plb,\Pub)$ that would allow observability satisfaction of $r$ as all upper and lower bounds are equal for any such $\omega'$ that differ in $u$ only.

Let us thus consider arbitrary $\omega'\in A_{u,v}(\Plb,\Pub)$ such that $\omega'\neq\omega$.
Since $\overline{B}\neq\emptyset$, it must hold that $\Plb_{v,\omega'}=\Plb_{v,\omega}$.
Furthermore we know that $\omega'$ is either $\preceq_v$-smaller than $\omega$ or incomparable.

Let us first consider $\omega'\parallel_v\omega$.
Since $\omega'\notin\overline{B}$, at least one of the three defining conditions must not hold.
We have already established $\Plb_{v,\omega'}=\Plb_{v,\omega}$.
Thus, either there exists $\omega''\in A_{u,v}(\Plb,\Pub)$ such that $\omega''\succ_v\omega'$,
in which case we repeat this analysis for $\omega''$,
ultimately leading to the following point by lemma~\ref{lem:mono-order-bounds} as $A_{u,v}(\Plb,\Pub)$ is finite.
Or $\Pub_{v,\omega'}=\Plb_{v,\omega'}=\Plb_{v,\omega}$.
Again, this gives us the same value, $\Plb_{v,\omega}$,
for all $\omega'\in A_{u,v}(\Plb,\Pub)$, $\omega'\parallel_v\omega$ and every parametrisation in $(\Plb,\Pub)$.

For the case $\omega'\preceq_v\omega$, lemma~\ref{lem:mono-order} gives us
$P_{v,\omega'}\leq P_{v,\omega}$ for any parametrisation $P\in(\Plb,\Pub)$.

Thus, for any parametrisation $P\in(\Plb,\Pub)$ such that $P_{v,\omega}=\Plb_{v,\omega}$
it also holds that $P_{v,\omega'}=\Plb_{v,\omega}$ for any $\omega'\in A_{u,v}(\Plb,\Pub)$.
Observability of $(u,v)$ is therefore not satisfied by any such $P$ giving us $P\notin p_R(T\cup\{t\})$
and $\lb{p_R(T\cup\{t\})}_{v,\omega}>\Plb_{v,\omega}$, thus $\lb{p_R(T\cup\{t\})}_{v,\omega}\geq\Plb'_{v,\omega}$.

The proof for the case $\underline{B}=\{\omega\}$ for some $\omega\in\Omega_v$ is symmetrical
and results in $\ub{p_R(T\cup\{t\})}_{v,\omega}\leq\Pub'_{v,\omega}$.

Since in case $\overline{B}\neq\emptyset$ only the lower bound is affected
and similarly in case $\underline{B}\neq\emptyset$ only the upper bound,
the combination of the results for both and the fact that $\Plb,\Plb'$
and $\Pub,\Pub'$ are always equal on all regulator states except $\omega$ gives us the coveted $[p_R(T\cup\{t\})]\subseteq(\Plb',\Pub')$.
\end{enumerate}

The above discussion concludes the proof of soundness, leaving only the inclusion in opposite direction to be proven.

$\pabs_R(T\cup\{t\})\subseteq[p_R(T\cup\{t\})]$ (best over-approximation):

Let us first remark that if the abstract parametrisation set is empty,
the infimum parametrisation in not smaller or equal to the supremum parametrisation,
we have $\pabs_R(T\cup\{t\})=\emptylattice\subseteq[p_R(T\cup\{t\})]$
regardless of the value of the smallest convex sublattice containing the concrete parametrisation set.
We therefore assume $\pabs_R(T\cup\{t\})\neq\emptylattice$.

If $t$ introduces no change to the smallest convex sublattice, $[p_R(T\cup\{t\})]=[p_R(T)]$
the result is trivial as $\pabs_R(T\cup\{t\})\subseteq\pabs_R(T)=[p_R(T)]$.
Let thus $\omega\in\Omega_v$ be such that $\lb{p_R(T\cup\{t\})}_{v,\omega}>\lb{p_R(T)}_{v,\omega}$
or $\ub{p_R(T\cup\{t\})}_{v,\omega}<\ub{p_R(T)}_{v,\omega}$
and let $t=x\xrightarrow{v,s}y$ where $s\in\{-1,1\}$.

Let us assume the lower bound has changed, $\lb{p_R(T\cup\{t\})}_{v,\omega}>\lb{p_R(T)}_{v,\omega}$,
instead of the upper bound and that the transition is increasing, $s=1$:

Note that the assumption of existence of $\omega$ expects $p_R(T\cup\{t\})\neq\emptyset$.
We therefore consider the change in lower bound of the value of $\omega$ as a disqualification of all parametrisations $P\in p_R(T)$
such that $P_{v,\omega}\leq\lb{p_R(T\cup\{t\})}_{v,\omega}$,
which in case of empty concrete parametrisation set translates to
all parametrisations with $P_{v,\omega}\leq\ub{p_R(T)}$, or simply all parametrisations in $p_R(T)$.
By abuse of notation we continue to use $\lb{p_R(T\cup\{t\})}$ and $\ub{p_R(T\cup\{t\})}$ in this sense even if $p_R(T\cup\{t\})=\emptyset$.
(Recall that we consider a lattice to be empty if the infimum parametrisation is not smaller or equal to supremum parametrisation.)

We now conduct a discussion on the relationship between $t$ and $\omega$ to show that
a restriction necessarily takes place in the abstract domain to reflect the change in concrete domain.
\begin{enumerate}[label=(\roman*)]
\item $\omega$ is the regulator state of $v$ in state $x$, $\omega=\omega_v(x)$
and the new lower bound for $\omega$ is the target value of $v$ of transition $t$, $\lb{p_R(T\cup\{t\})}_{v,\omega}=y_v$.

In this case, the change in $\omega$ can be attributed to the transition $t$ itself.
Since the value for $\omega$ changed, we have $y_v=\lb{p_R(T\cup\{t\})}_{v,\omega}>\lb{p_R(T)}_{v,\omega}$.
By definition of $\restrict_t$, $\lb{\restrict_t(\pabs_R(T))}_{v,\omega}=y_v$ and thus $\lb{\pabs_R(T\cup\{t\})}_{v,\omega}\geq y_v$.

\item $\omega\succeq_v\omega_v(x)$ and $\lb{p_R(T\cup\{t\})}_{v,\omega}=y_v$.

In this case the change in $\omega$ can be attributed to a combination of monotonicity constraints.
Again, since the value of $\omega$ changed, $y_v=\lb{p_R(T\cup\{t\})}_{v,\omega}>\lb{p_R(T)}_{v,\omega}=\pabs_R(T)$.
We already know from the previous point that $\lb{\pabs_R(T\cup\{t\})}_{v,\omega_v(x)}\geq y_v$.
Thus, by lemma~\ref{lem:mono-restrict} the monotonicity restrictions on $\restrict_t(\pabs_R(T))$
enforce $\lb{\pabs_R(T\cup\{t\})}_{v,\omega}\geq\lb{\pabs_R(T\cup\{t\})}_{v,\omega_v(x)}\geq y_v$.

\item For any other $\omega$ or $\lb{p_R(T\cup\{t\})}_{v,\omega}>y_v$.

In this case, the change in $\omega$ can be attributed to an observability constraint.
Let $(\Plb,\Pub)$ be the result of the restrictions discussed in the two previous points,
formally $(\Plb,\Pub)=\restrict_{\{(u,v,s)\in R\mid s\in\{-1,1\}\}}(\restrict_t(\pabs_R(T)))$.

We now show that the change in the lower bound for $\omega$ in the concrete parametrisation set
is reflected in the abstract parametrisation set by an observability restriction,
and furthermore, that no other restrictions are necessary after the observability restriction fires,
i.e. $\omega$ is the single regulator state of node $v$ with strictly higher value in $\lb{p_R(T\cup\{t\})}$ compared to $\Plb$.

Since $\lb{p_R(T\cup\{t\})}_{v,\omega}>\Plb$ we know that no parametrisation $P\in[p_R(T)]$
such that $P_{v,\omega(x)}=\Plb_{v,\omega_v(x)}=y_v$ and $P_{v,\omega}=\Plb_{v,\omega}$ belongs to $p_R(T)$.
One can observe that under the assumption $\omega_v(x)\neq\omega$ the above is easily applicable to theorem~\ref{lem:density}.
In fact, even in the case $\omega_v(x)=\omega$, by using arbitrary $\hat\omega\in\Omega_v\setminus\{\omega_v(x)\}$,
there is no parametrisation in $P\in[p_R(T)]$ with $P_{v,\hat\omega}=\Plb_{v,\hat\omega}$
and $P_{v,\omega(x)}=\Plb_{v,\omega(x)}$ in $p_R(T)$.

As the rest of the proof is independent of whether $\omega$ equals $\omega_v(x)$
we unify the notation for application of theorem~\ref{lem:density}:

\begin{center}
\begin{tabular}{ll}
$\omega_v(x)\neq\omega$ & $\omega_v(x)=\omega$\\\hline
$\overline\omega=\omega_v(x)$ & $\overline\omega = \hat\omega$\\
$\overline\omega'=\omega$ & $\overline\omega'= \omega_v(x)=\omega$
\end{tabular}
\end{center}

Here $\overline\omega$ and $\overline\omega'$ represent the $\omega$ and $\omega'$, respectively,
as used in definition of theorem~\ref{lem:density}.
The values denoted as $y_v$ and $z_v$ in definition of theorem~\ref{lem:density}
thus become $\Plb_{v,\overline\omega}$ and $\Plb_{v,\overline\omega'}$ respectively.
This assignment is valid as the requirements
$\overline\omega\preceq_v\overline\omega'\Rightarrow\Plb_{v,\overline\omega}\leq\Plb_{v,\overline\omega'}$,
respectively $\overline\omega\succeq_v\overline\omega'\Rightarrow\Plb_{v,\overline\omega}\geq\Plb_{v,\overline\omega'}$,
are satisfied by lemma~\ref{lem:mono-restrict}.

We thus use theorem~\ref{lem:density}.
As the second part of the implication does not hold in our case, the first part cannot hold either
and therefore no additional regulator state $\overline\omega''\in\Omega_v$
such that $\lb{p_R(T)}_{v,\overline\omega''}<\ub{p_R(T)}_{v,\overline\omega''}$ exists.
Formally, for any $\overline\omega''\in\Omega_v$:
\begin{align*}
\overline\omega''\prec_v\overline\omega&\Rightarrow\Plb_{v,\overline\omega''}=\Plb_{v,\overline\omega}\\
\overline\omega''\succ_v\overline\omega&\Rightarrow\Pub_{v,\overline\omega''}=\Plb_{v,\overline\omega}\\
\overline\omega''\prec_v\overline\omega'&\Rightarrow\Plb_{v,\overline\omega''}=\Plb_{v,\overline\omega'}\\
\overline\omega''\succ_v\overline\omega'&\Rightarrow\Pub_{v,\overline\omega''}=\Plb_{v,\overline\omega'}\\
\overline\omega''\parallel_v\overline\omega\wedge\overline\omega''\parallel_v\overline\omega'
&\Rightarrow\Plb_{v,\overline\omega''}=\Pub_{v,\overline\omega''}
\end{align*}

We now show that an observability constraint indeed enforces restriction.
Let us consider parametrisation $\Plb$, we know $\Plb\in p(T\cup\{t\})$ by definition
and thanks to lemma~\ref{lem:mono-order} also $\Plb\in\bigcap_{\{r=(u,v,s)\in R\mid s\in\{-1,1\}\}}\P_r$.
However, $\Plb\notin p_R(T\cup\{t\})$ meaning there must exist an observability constraint $r=(u,v,\o)\in R$ such that $\Plb\notin\P_r$.

We now explore $\restrict_r$ for the observability constraint $r$.
Recall that $\overline\omega'=\omega$ regardless of the equality between $\omega$ and $\omega_v(x)$.
A discussion on the nature of the lower and upper bounds of the value of $\overline\omega'$ follows.

Let us first assume $\Plb_{v,\overline\omega'}=\Pub_{v,\overline\omega'}$.
As such, $\lb{p_R(T\cup\{t\})}_{v,\overline\omega'}>\Pub_{v,\overline\omega'}\geq\ub{p_R(T\cup\{t\})}_{v,\overline\omega'}$
and thus $p_R(T\cup\{t\})=\emptyset$.

We will now show that $\Plb_{v,\overline\omega'}=\Pub_{v,\overline\omega'}$
leads to $A_{u,v}(\Plb,\Pub)=\emptyset$ and in turn $\pabs_R(T\cup\{t\})=\emptylattice$.
Since $\Plb\notin\P_r$ we know that the lower bound for any regulator state of $v$ is independent of the value of $u$, formally,
$\forall\omega'\in\Omega_v,\forall x_u\in\{1,\dots,\mv u\}:\Plb_{v,\subst{\omega'} u {x_u}}=\Plb_{v,\subst{\omega'} u {x_u-1}}$.

Furthermore, applying the results from theorem~\ref{lem:density},
any $\overline\omega''$ $\preceq_v$-incomparable to neither $\overline\omega$ nor $\overline\omega'$
has the upper bound equal to the lower bound.
The same must also hold for any $\overline\omega''\succeq_v\overline\omega$ or $\overline\omega''\succeq_v\overline\omega'$
since the upper bound of any such $\overline\omega''$ is at most $\Plb_{v,\overline\omega}$, respectively $\Plb_{v,\overline\omega'}$.
The lower and upper bounds are also equal for any $\overline\omega''\preceq_v\overline\omega'$
as the lower bound must be equal to $\Plb_{v,\overline\omega'}$
and the upper bound cannot exceed $\Pub_{v,\overline\omega'}$ by lemma~\ref{lem:mono-restrict}.

Finally, we show that $\Pub_{v,\overline\omega}$ is also equal to $\Plb_{v,\overline\omega}$
and thus the same must hold for any $\overline\omega''\preceq_v\overline\omega$.
If there exists $\omega'\in\Omega_v$ such that $\omega'\succeq_v\overline\omega$ then we have $\Pub_{v,\omega'}=\Plb_{v,\overline\omega}$
and by lemma~\ref{lem:mono-restrict} $\Pub_{v,\overline\omega}\leq\Pub_{v,\omega'}$ giving us the coveted equality of lower and upper bounds.

Assuming thus, $\overline\omega$ is $\preceq_v$-maximal,
we prove that $\Plb_{v,\overline\omega}=\Pub_{v,\overline\omega}$ by contradiction.
Let thus $\Plb_{v,\overline\omega}<\Pub_{v,\overline\omega}$.
Then for parametrisation $P=\update L {v,\overline\omega} {} {\Pub_{v,\overline\omega}}\in(\Plb,\Pub)$ it holds that $P\in p(T\cup\{t\})$,
by lemma~\ref{lem:mono-order} $P$ also satisfies all monotonicity constraints
and finally, by lemma~\ref{lem:observ_parity}, $P$ satisfies all observability constraints.
Thus, $P\in p_R(T\cup\{t\})$ which is a contradiction with $p_R(T\cup\{t\})=\emptyset$.

Clearly then for any $\omega'\in\Omega_v$ and any value $x_u\in\{1,\dots,\mv u\}$ we have
$\Plb_{v,\subst{\omega'} u {x_u}}=\Pub_{v,\subst{\omega'} u {x_u}}=\Plb_{v,\subst{\omega'} u {x_u-1}}=\Pub_{v,\subst{\omega'} u {x_u-1}}$
giving us the coveted $A_{u,v}(\Plb,\Pub)=\emptyset$ and $\pabs_R(T\cup\{t\})=\emptylattice$.

Let us therefore assume $\Plb_{v,\overline\omega'}<\Pub_{v,\overline\omega'}$.
This yields $\overline\omega'\in A_{u,v}(\Plb,\Pub)$.
We will now show that $\overline{B}=\{\overline\omega'\}$.

Let us first remark that under the assumption $\Plb_{v,\overline\omega'}<\Pub_{v,\overline\omega'}$
no $\overline\omega''\succeq_v\overline\omega'$ can exist.
Recall that as a result of theorem~\ref{lem:density} we know $\Pub_{v,\overline\omega''}=\Plb_{v,\overline\omega'}$.
This gives us $\Pub_{v,\overline\omega''}<\Pub_{v,\overline\omega'}$ which is a contradiction with theorem~\ref{lem:mono-restrict}.

We next show that $\overline\omega\parallel_v\overline\omega'\Rightarrow\Plb_{v,\overline\omega}=\Pub_{v,\overline\omega}$.
First remark that the conditions for any $\overline\omega''\succeq_v\overline\omega$ given by theorem~\ref{lem:density}
are analogous to those imposed on $\overline\omega''\succeq_v\overline\omega'$.
Thus, no such $\overline\omega''$ exists if $\Plb_{v,\overline\omega}<\Pub_{v,\overline\omega}$.
Furthermore as $\overline\omega\parallel_v\overline\omega$ there are no monotonic restrictions necessary on $\overline\omega'$
if value of $\overline\omega$ is changed.
Thus, parametrisation $\update \Plb {v,\overline\omega} {} {\Pub_{v,\overline\omega}}\in p_R(T\cup\{t\})$
by lemmas~\ref{lem:mono-restrict},\ref{lem:mono-order}~and~\ref{lem:observ_parity}
as long as $\Plb_{v,\overline\omega}<\Pub_{v,\overline\omega}$ holds.
Since we know no parametrisation $P$ with $P_{v,\overline\omega'}=\Plb_{v,\overline\omega'}$ belongs to $p_R(T\cup\{t\})$
it must hold that $\Plb_{v,\overline\omega}=\Pub_{v,\overline\omega}$.

By extension, $\Plb_{v,\overline\omega}=\Pub_{v,\overline\omega}$ gives us the same result,
$\Plb_{v,\overline\omega''}=\Pub_{v,\overline\omega''}$,
for any $\overline\omega''$ $\preceq_v$-smaller or $\preceq_v$-larger than $\overline\omega$.
As such, we can simplify the constraints given by theorem~\ref{lem:density} to only consider $\preceq_v$-relation to $\overline\omega'$
and treat $\overline\omega$ as any other regulator state.

Let now $\omega'\in A_{u,v}\setminus\{\overline\omega'\}$ be arbitrary.
The case $\omega'\preceq_v\overline\omega'$ is simple
as $\Plb_{v,\overline\omega''}=\Plb_{v,\overline\omega'}$ follows directly from the application of theorem~\ref{lem:density}.
We now discuss $\omega'\parallel_v\overline\omega'$.

We know that $\Plb_{v,\overline\omega''}=\Pub_{v,\overline\omega''}$
for any $\overline\omega''$ $\preceq_v$-incomparable to $\overline\omega'$.
Thus, $\omega'\in A_{u,v}(\Plb,\Pub)$ requires that $(u,v)$ is not monotonic
and $\subst{\omega'} u {\overline\omega'_u}\preceq_v\overline\omega'$.
By previous point $\Plb_{v,\subst{\omega'} u {\overline\omega'_u}}=\Plb_{v,\overline\omega'}$
and since $\Plb\notin\P_r$ we know that the lower bounds are equal for any regulator states differing in $u$ only,
giving us the coveted $\Plb_{v,\omega'}=\Plb_{v,\overline\omega'}$.

As such, $\overline{B}=\{\overline\omega'\}=\{\omega\}$
and the lower bound of $\omega$ is increased by one by observability restriction $\restrict_r$.
By lemma~\ref{lem:observ_parity} we know $\update \Plb {v,\omega} + 1$ satisfies all observability constraints
and since $\omega$ is $\preceq_v$-maximal
and $\update \Plb {v,\omega} + 1\in(\Plb,\Pub)$ we have $\update \Plb {v,\omega} + 1\in p_R(T\cup\{t\})$
by lemmas~\ref{lem:mono-restrict} and \ref{lem:mono-order}.
As such there may exist no other $\omega'\in\Omega_v$
such that $\lb{p_R(T\cup\{t\})}_{v,\omega'}>{\update \Plb {v,\omega} + 1}_{v,\omega'}$.
There may be however, an $\omega'\in\Omega_v$ such that $\ub{p_R(T\cup\{t\})}_{v,\omega'}<\Pub$.
Such $\omega'$ can be treated symmetrically to $\omega$ to show $\underline{B}=\{\omega'\}$.
\end{enumerate}

Since the last discussion in the third point is the only case when introducing an increasing transition can decrease the upper bound,
the above discussion indeed proves $\pabs_R{T\cup\{t\}}\subseteq[p_R(T\cup\{t\})]$ for an increasing transition $t$ by showing that
$\lb{\pabs_R{T\cup\{t\}}}_{v,\omega}\geq\lb{p_R(T\cup\{t\})}_{v,\omega}$,
respectively $\ub{\pabs_R{T\cup\{t\}}}_{v,\omega}\leq\ub{p_R(T\cup\{t\})}_{v,\omega}$, for arbitrary $\omega\in\Omega_v$

The proof is completely symmetrical for the case when $t$ is decreasing.
The regulator state $\omega$ being such that $\ub{p_R(T\cup\{t\})}_{v,\omega}<\ub{p_R(T)}_{v,\omega}$,
except for the final discussion in the third point, which is again, the only case when a lower bound can increase when introducing a decreasing transition.

As such, $\pabs_R{T\cup\{t\}}\subseteq[p_R(T\cup\{t\})]$ holds. Combined with soundness,
$[p_R(T\cup\{t\})]\subseteq\pabs_R{T\cup\{t\}}$, we obtain the coveted $\pabs_R(T\cup\{t\})=[p_R{T\cup\{t\}}]$.
\end{proof}

\end{document}